\documentclass[a4paper,UKenglish,cleveref, autoref, thm-restate]{lipics-v2021}
\pdfoutput=1
\usepackage[utf8]{inputenc}

\usepackage{amsthm, amsmath}
\usepackage{mathtools}
\usepackage{enumitem}
\usepackage{setspace}
\usepackage{makecell}
\usepackage{multirow}
\usepackage{hhline}

\usepackage{caption}
\usepackage{subcaption}
\usepackage{graphics}

\usepackage{tikz}
\usetikzlibrary{arrows,shapes,trees}
\usetikzlibrary{arrows.meta}

\usepackage{floatrow}
\newfloatcommand{capbtabbox}{table}[][\FBwidth]

\makeatletter
\@fleqnfalse
\makeatother

\title{DAG Scheduling in the BSP Model}

\author{P\'al Andr\'as Papp}{Computing Systems Lab, Huawei Zurich Research Center}{pal.andras.papp@huawei.com}{}{}
\author{Georg Anegg}{Computing Systems Lab, Huawei Zurich Research Center}{georg.anegg@huawei.com}{}{}
\author{Albert-Jan N. Yzelman}{Computing Systems Lab, Huawei Zurich Research Center}{albertjan.yzelman@huawei.com}{}{}

\authorrunning{P.\,A. Papp, G. Anegg and A.\,N. Yzelman}

\Copyright{P\'al Andr\'as Papp, Georg Anegg and Albert-Jan N. Yzelman}



\EventEditors{John Q. Open and Joan R. Access}
\EventNoEds{2}
\EventLongTitle{42nd Conference on Very Important Topics (CVIT 2016)}
\EventShortTitle{CVIT 2016}
\EventAcronym{CVIT}
\EventYear{2016}
\EventDate{December 24--27, 2016}
\EventLocation{Little Whinging, United Kingdom}
\EventLogo{}
\SeriesVolume{42}
\ArticleNo{23}
\nolinenumbers
\hideLIPIcs 
\begin{document}

\maketitle

\begin{abstract}
  We study the problem of scheduling an arbitrary computational DAG on a fixed number of processors while minimizing the makespan. While previous works have mostly studied this problem in fairly restricted models, we define and analyze DAG scheduling in the Bulk Synchronous Parallel (BSP) model, which is a well-established parallel computing model that captures the communication cost between processors much more accurately. We provide a taxonomy of simpler scheduling models that can be understood as variants or special cases of BSP, and discuss how the properties and optimum cost of these models relate to BSP. This essentially allows us to dissect the different building blocks of the BSP model, and gain insight into how these influence the scheduling problem.
  
  We then analyze the hardness of DAG scheduling in BSP in detail. We show that the problem is solvable in polynomial time for some very simple classes of DAGs, but it is already NP-hard for in-trees or DAGs of height $2$. We also prove that in general DAGs, the problem is APX-hard: it cannot be approximated to a $(1+\epsilon)$-factor in polynomial time for some specific $\epsilon>0$. We then separately study the subproblem of scheduling communication steps, and we show that the NP-hardness of this problem depends on the problem parameters and the communication rules within the BSP model. Finally, we present and analyze a natural formulation of our scheduling task as an Integer Linear Program.
\end{abstract}

\section{Introduction}

The optimal scheduling of complex workloads is a fundamental problem not only in computer science, but also in other areas like logistics or operations research. In a computational context, the most natural application of scheduling is when we have a complex computation consisting of many different subtasks, and we want to execute this on a parallel (multi-processor or multi-core) architecture, while minimizing the total time required for this. Unsurprisingly, this topic has been extensively studied since the 1960s, and has gained even more importance recently with the widespread use of manycore architectures.

In these scheduling problems, a computational task is represented as a Directed Acyclic Graph (DAG), where each node corresponds to an operation or subtask, and each directed edge $(u,v)$ indicates a dependency relation, i.e.\ that the processing of node $u$ has to be finished before the processing of node $v$ begins, since the output of operation $u$ is needed as an input for operation $v$. This DAG model of general computations is not only prominent in scheduling, but also in further topics such as pebbling.%, which studies time-memory tradeoffs in a given computation.

However, from a complexity-theoretic perspective, the scheduling of general DAGs is already a hard problem even in very simple settings, e.g.\ even in models that heavily simplify or completely ignore the communication costs between processors, which is the main bottleneck in many computational tasks in practice. Due to this, previous theoretical works have mostly focused on analyzing the complexity of scheduling in these rather simple models, and while more realistic models were sometimes introduced, the theoretical properties of scheduling in these more realistic models received little attention.

On the other hand, the parallel computing community has developed far more sophisticated models to accurately quantify the real cost of parallel algorithms in practice. One of the most notable among these is the Bulk Synchronous Parallel (BSP) model, which is still relatively simple, but provides a delicate cost function that also captures the volume of communicated data and the synchronization costs in a given parallel schedule. BSP (and other similar models) are fundamental tools for evaluating and comparing concrete practical implementations of parallel algorithms. However, previous theoretical works on BSP only focus on finding and analyzing parallel schedules for specific algorithms, and do not study BSP as a model for scheduling general DAGs, i.e.\ an arbitrary computational task.

Our goal in this paper is to bridge this gap between theory and practice to some extent, and understand the fundamental theoretical properties of DAG scheduling in the BSP model. This more detailed model results in a more complex scheduling problem with some entirely new aspects compared to classical models; as such, understanding its key properties is a crucial step towards designing efficient parallel schedules for computations in practice. Our hope is that our insights can inspire a new line of work, focusing on the theoretical study of scheduling in more advanced parallel computing models that are developed and applied on the practical side. More specifically, our main contributions are as follows:
\begin{enumerate}[label=(\roman*),leftmargin=15pt, itemsep=3pt, topsep=3pt]
\item We first define DAG scheduling in the BSP model. We then provide a taxonomy of scheduling models from previous works, and show that many of these can be understood as a special case of BSP. Analyzing the relations between these models essentially allows us to gain insight into how the different aspects of BSP affect the scheduling problem.
\item We then analyze the complexity of the BSP DAG scheduling problem, with the goal of understanding when (i.e.\ for which kind of DAGs) the problem becomes NP-hard. Our results show that BSP scheduling is still solvable in polynomial time for simple DAGs that consist of several connected chains, but it is already NP-hard for slightly more complex classes of DAGs, such as in-trees or DAGs of height $2$.
\item We show that for general DAGs, the problem is NP-hard to even approximate to a $(1+\epsilon)$ factor, for some constant $\epsilon>0$. In order words, BSP scheduling is APX-hard, and does not allow a polynomial-time approximation scheme unless P=NP.
\item We separately analyze the subproblem of scheduling communication steps, assuming that the assignment of tasks to processors and so-called supersteps is already fixed. We discuss the complexity of this subproblem for several different variants of BSP.
\item Finally, we present and analyze a natural formulation of the BSP scheduling task as an Integer Linear Programming (ILP) problem.
\end{enumerate}
Our main technical contributions correspond to points (ii)--(iv) above; however, the remaining points also provide valuable insight into different properties of the problem.

\section{Related work}

Scheduling is a fundamental problems in computer science, and has been studied extensively since the 1970s. DAGs are one of the most common models for such problems, since in many applications, the subtasks or computation steps have precedence relations between them.

The first results on DAG scheduling considered a simple setting where communication between the processors is free, which essentially corresponds to the PRAM model. The numerous results in this model include polynomial algorithms for $P _{\!} = _{\!} 2$ processors~\cite{DAG2proc1, DAG2proc2, DAG2proc3}, polynomial algorithms for special classes of DAGs (but $P _{\!} > _{\!} 2$)~\cite{oppforest, levelorder, boundedheight}, hardness results for $P_{\!} = _{\!} \infty$~\cite{PInf1, PInf2, PInf3}, and results for weighted DAGs~\cite{weighted1, weighted2, weighted3}. The results on some of these topics, e.g.\ approximation algorithms, are still rapidly improving in recent years~\cite{approx1, approx2, approx3, approx4}. On the other hand, some basic questions are still open even in this fundamental model: e.g.\ it is still not known whether scheduling for some fixed $P _{\!} > _{\!} 2$ is NP-hard.

A more realistic version of this model was introduced in the late 1980s~\cite{commDdef1, commDdef2}, where there is a fixed communication delay between processors. There are also numerous algorithms and hardness results for this setting, in particular for unit-length delays~\cite{commDunit1, commDunit2inf2} or infinitely many processors~\cite{commDinf1, commDunit2inf2, commDinf3}. The approximability of the optimal solution is also a central question in this model that receives significant attention even in recent years \cite{commDappR1, commDappR2, recent1, recent4, recent3}.

This communication delay model is still unrealistic in the sense that it allows an unlimited amount of data to be communicated in a single time unit. To our knowledge, the only more sophisticated DAG scheduling model which measures communication volume is the recently introduced single-port duplex model~\cite{SPD}. However, this model has not been studied before from a theoretical perspective, and in contrast to BSP, it cannot be extended by some real-world aspects such as synchronization costs.

%We discuss these scheduling models in detail in Section \ref{sec:tax} and Appendix \ref{app:def}.
There are also numerous extensions of these models with further aspects, e.g.\ heterogeneous processors, jobs requiring several processors, or deadlines for each task~\cite{unrelated, deadlines, shopsched, differentproc, recent2, recent5}.

On the other hand, BSP has been introduced as a prominent model of parallel computing in 1990~\cite{BSPintro}, and studied extensively ever since~\cite{BSPbook1, BSPbook2, BSPqa}. The model has also found its way into various applications, most notably through the BSPlib standard library~\cite{BSPlib} and its different implementations~\cite{BSPimpl1, BSPimpl2}. Fundamental results on BSP include the analysis of prominent algorithms in this model~\cite{BSPalg1, BSPalg2} or the extension of BSP to multi-level architectures~\cite{multiBSP1, multiBSP2}. However, the BSP model (and similar models with even more parameters, such as LogP~\cite{LogP}) have mostly been used so far to analyze the computational costs of specific parallel implementations of concrete algorithms. In contrast to this, in our work, we apply BSP as a general model to evaluate the scheduling of any computational DAG.

\section{Models and definitions} \label{sec:BSP}

Computational tasks are modelled as a \emph{Directed Acyclic Graph} $G$, with the set of nodes (subtasks) denoted by $V$, and the set of directed edges (dependencies) by $E$. We use $u$ and $v$ to denote individual nodes, and $n$ to denote the number of nodes $|V|$. An edge $(u,v)$ indicates that subtask $u$ has to be finished before the computation of subtask $v$ begins. %The \emph{indegree} and \emph{outdegree} of $v$ is the number of incoming and outgoing edges of $v$, respectively. 
We also use $[k]$ as a shorthand notation for the integer set $\{ 1, ..., k  \}$.

For scheduling, we assume that we have $P \in O(1)$ identical processors, and our goal is to execute all nodes of $G$ on these, while minimizing the total amount of time this takes.

\subparagraph{The BSP model.}

BSP is a very popular model for the design and evaluation of parallel algorithms. To our knowledge, general DAG scheduling has not been studied in this model before, but extending the interpretation of BSP to this setting is rather straightforward.

In contrast to classical models where nodes are assigned to concrete points in time, the BSP model instead divides the execution of nodes into larger batches, so-called \emph{supersteps}. Each superstep consists of two phases, in the following order:
\begin{enumerate}[label=\arabic*., leftmargin=12pt, itemsep=3pt, topsep=3pt]
 \item \emph{Computation phase}: each processor may execute an arbitrary number of computation steps, but no communication between the processors is allowed.
 \item \emph{Communication phase}: processors can communicate an arbitrary number of values to each other, but no computation is executed.
\end{enumerate}
The main motivation behind supersteps is to encourage executing the necessary communications in large batches, since in practice, inter-processor communication often has a large fixed cost (e.g.\ synchronization, network initialization) that is independent of the data volume.

\subparagraph{Formal definition.}

We denote the number of supersteps in our schedule by $S$. While $S$ is not a parameter of the problem and can be freely chosen in a schedule, we provide our definitions for a fixed $S$ for simplicity. A \emph{BSP schedule} with $S$ supersteps consists of:
%A \emph{BSP schedule} with a given number of supersteps $S$ consists of:
\begin{itemize}[leftmargin=15pt, itemsep=4pt, topsep=2pt]
 \item An assignment of nodes to processors $\pi: V \rightarrow [P]$ and to supersteps $\tau: V \rightarrow [S]$. For simplicity, we introduce the notation $H^{(s,p)}=\{v \in V \, | \, \pi(v)=p$, $\tau(v)=s \}$ for the set of nodes assigned to processor $p$ and superstep $s$. We can imagine the nodes of $H^{(s,p)}$ to be executed in an arbitrary (but topologically correct) order on $p$ in superstep $s$.
 \item A set $\Gamma$ of 4-tuples $ (v, p_1, p_2, s) \in V _{\!} \times _{\!} [P] _{\!} \times _{\!} [P] _{\!} \times _{\!} [S]$, indicating that the output of node $v$ is sent from processor $p_1$ to processor $p_2$ in the communication phase of superstep $s$. In this base variant of BSP, we only include $p_1$ in these 4-tuples for clarity, but we always assume $p_1=\pi(v)$, i.e.\ the value is sent from the processor where it was computed. 
\end{itemize}

\noindent A valid BSP schedule must satisfy the following conditions:
\begin{enumerate}[label=(\roman*),leftmargin=15pt, itemsep=4pt, topsep=2pt]
 \item A node $v$ can only be computed if all of its predecessors are available, i.e.\ they were computed on processor $\pi(v)$ in an earlier (or the same) superstep, or sent to $\pi(v)$ before the given superstep. That is, for all $(u,v) _{\!} \in _{\!} E$, if $\pi(u)=\pi(v)$ then we must have $\tau(u) \leq \tau(v)$, and if $\pi(u) \neq \pi(v)$ then we must have $(u, \pi(u), \pi(v), s) \in \Gamma$ for some $s _{\!} < _{\!} \tau(v)$.
 \item We only communicate values that are already computed: if $(v, p_1, p_2, s) _{\!} \in _{\!} \Gamma$, then $p_1=\pi(v)$, and $\tau(v) \leq s$.
\end{enumerate}

\subparagraph{Cost function.}

The computation phase can be executed in parallel on the different processors, so its cost (the amount of time it takes) in superstep $s _{\!} \in _{\!} [S]$ is the largest amount of computation executed on any of the processors. More formally, the \emph{work cost} of superstep $s$ (first for a given processor $p \in [P]$, and then in general) is defined as
\[ C_{work}\,\!^{(s,p)}=|H^{(s,p)}| \, \qquad \text{ and } \qquad C_{work}\,\!^{(s)} = \max_{p \in [P]} \, C_{work}\,\!^{(s,p)} \, . \]

Communication costs, on the other hand, are governed by two further problem parameters: $g _{\!} \in _{\!} \mathbb{N}$ is the cost of communicating a single unit of data, and $L _{\!} \in _{\!} \mathbb{N}$ is the fixed latency cost incurred by each superstep. BSP assumes that different values can be communicated in parallel in general, but any processor can only send and receive a single value in any time unit. As such, BSP considers the number of values sent and received by processor $p$ in superstep $s$, and then defines the \emph{communication cost} of a superstep (for $p$, or in general) as the maximum of these; this cost function is also known as a \emph{$h$-relation}. More formally, let
\[
C_{sent}\,\!^{(s,p)}=|\{ (v,p,p',s) \in \Gamma \}| \, \qquad \text{ and } \qquad C_{rec}\,\!^{(s,p)}=|\{ (v,p',p,s) \in \Gamma \}| \,
\]
for some fixed $s _{\!} \in _{\!} [S]$ and $p _{\!} \in _{\!} [P]$ (over all $v _{\!} \in _{\!} V$ and $p' _{\!} \in _{\!} [P]$), and then let
\[
C_{comm}\,\!^{(s,p)} = \max( C_{sent}\,\!^{(s,p)}\, , \, C_{rec}\,\!^{(s,p)}) \, \qquad \text{ and } \qquad C_{comm}\,\!^{(s)} = \max_{p \in [P]} \, C_{comm}\,\!^{(s,p)} \, .
\]
The cost $C\,\!^{(s)}$ of superstep $s$ and the cost $C$ of the entire schedule is then defined as:
\[
C\,\!^{(s)} = C_{work}\,\!^{(s)} \, + \, g \cdot C_{comm}\,\!^{(s)} + L \, \qquad \text{ and } \qquad C = \sum_{s \in [S]} \, C\,\!^{(s)} \, .
\]

For an example, consider the BSP schedule shown in Figure~\ref{fig:BSPexample}, and let $s_{\!}=_{\!}1$, $p_{1\!}=_{\!}1$, $p_{2\!}=_{\!}2$. Here processor $p_1$ computes $4$ nodes, and processor $p_2$ computes $5$ nodes, so $C_{work}\,\!^{(s,p_1)}_{\!}=_{\!}4$, $C_{work}\,\!^{(s,p_2)}_{\!}=_{\!}5$, and $C_{work}\,\!^{(s)}_{\!}=_{\!}\max(4,5)_{\!}=_{\!}5$ in the computation phase. In the communication phase, $p_1$ must send a single value to $p_2$ (so $C_{sent}\,\!^{(s,p_1)}_{\!}=_{\!}C_{rec}\,\!^{(s,p_2)}_{\!}=_{\!}1$), while $p_2$ must send two values to $p_1$ ($C_{sent}\,\!^{(s,p_2)}_{\!}=_{\!}C_{rec}\,\!^{(s,p_1)}_{\!}=_{\!}2$). This implies $C_{comm}\,\!^{(s,p_1)}_{\!}=_{\!}C_{comm}\,\!^{(s,p_2)}_{\!}=_{\!}\max(2, 1)_{\!}=_{\!}2$, and hence $C_{comm}\,\!^{(s)}_{\!}=_{\!}2$. The total cost of the superstep is $C\,\!^{(s)} = 5 + 2 _{\!}\cdot_{\!} g + L$. For more details on the BSP model, we refer the reader to~\cite{BSPbook1, BSPbook2}.

\subparagraph{As a scheduling problem.}

We can now formally define our DAG scheduling problem.

\begin{definition}
Given an input DAG, the goal of the \emph{BSP scheduling problem} is to find a feasible BSP schedule $(\pi, \tau, \Gamma)$ as described above, with minimal cost $C$. 
\end{definition}

In the decision version of the problem, we also have a maximal cost parameter $C_0$, and we need to decide if there is a BSP schedule with cost $C \leq C_0$. For simplicity, we will often focus on the case of $L _{\!} = _{\!} 0$, which is already similar to $L _{\!} \geq _{\!} 0$ in terms of hardness.

From a complexity perspective, it is important to note that we consider the parameters $P, g, L$ to be small fixed constants (properties of our computing architecture), and not parts of the problem input. We especially emphasize this for $P$, since in contrast to our work, some others assume that $P$ is an input variable that can be up to linear in $n$; however, this is unrealistic in most applications, and also makes the problem unreasonably hard even for trivial DAGs. In general, both settings (fixed $P$ and variable $P$) have been extensively studied before, and are distinguished by ``$Pm$'' and ``$P$'' in the classical 3-field notation~\cite{3field}.

\subparagraph{Further model extensions.}

While we mostly focus on the base model described above, we sometimes also point out how our claims carry over to two very natural extensions of the scheduling problem that occur frequently in the literature. A more formal definition of these extensions (as well as all other models in the paper) is provided in Appendix \ref{app:def}.
\begin{itemize}[leftmargin=15pt, itemsep=4pt, topsep=3pt]
 \item The model can be extended with \textit{node weights}, i.e.\ work weights $w_{work}: V \rightarrow \mathbb{Z}^+$ and/or a communication weights $w_{comm}: V \rightarrow \mathbb{Z}^+$, to capture that in practice, it often takes a different amount of time to compute different subtasks, or communicate their output.
 \item We can also allow \emph{duplications} (also called replications), i.e.\ to execute the same node multiple times, on different processors; this can decrease the amount of communication required between processors, hence sometimes reducing the total cost.
\end{itemize}

\section{Comparison to other models} \label{sec:tax}

\subsection{Taxonomy of scheduling models}

In the most basic scheduling model, which we call \emph{classical scheduling}, nodes are assigned to processors $\pi: V \rightarrow [P]$ and time steps $t: V \rightarrow \mathbb{Z}^+$, with two conditions: we cannot execute two nodes on the same processor at the same time ($\nexists \, u, v \in V$ with $\pi(u)=\pi(v)$, $t(u)=t(v)$), and we must respect the precedence constraints ($\forall (u,v) \in E$ we need $t(u)<t(v)$). The cost of a schedule is simply its makespan $\max_{v \in V\,} t(v)$. This setting essentially assumes no communication cost between the processors.

A more realistic version of this setting also considers a fixed communication delay between processors; we call this the \emph{commdelay} model. This model has a further parameter $g$, and only differs in the second validity condition: for all $(u,v) \in E$, in case if $\pi(u) \neq \pi(v)$, we now need to have $t(v)>t(u)+g$.

\begin{figure}
\begin{floatrow}
\ffigbox[0.42\textwidth]{
    \centering
    \resizebox{0.42\textwidth}{!}{
    \begin{tikzpicture}
	
    \begin{scope}[very thick, gray]
    \draw (-4pt,0pt) rectangle (114pt,40pt);
    \draw (-4pt,60pt) rectangle (114pt,100pt);
    \draw (230pt,60pt) -- (166pt,60pt) -- (166pt,100pt) -- (230pt,100pt);
    \draw (230pt,0pt) -- (166pt,0pt) -- (166pt,40pt) -- (230pt,40pt);
    \end{scope}

    \node[anchor=center, gray] at (55pt,-8pt) {\small \textbf{processor 2, superstep 1}}; 
    \node[anchor=center, gray] at (55pt,108pt) {\small \textbf{processor 1, superstep 1}};

    \node[anchor=center, gray] at (200pt,-8pt) {\small \textbf{proc. 2, sup. 2}}; 
    \node[anchor=center, gray] at (200pt,108pt) {\small \textbf{proc. 1, sup. 2}};

    \begin{scope}[thick, arrows=-stealth]
    \draw (10pt,30pt) -- (36pt,11pt);
    \draw (40pt,10pt) -- (66pt,29pt);
    \draw (70pt,30pt) -- (96pt,30pt);
    \draw (70pt,30pt) -- (96pt,11pt);

    \draw (10pt,70pt) -- (36pt,89pt);
    \draw (40pt,90pt) -- (67pt,73pt);
    \draw (10pt,70pt) -- (66pt,70pt);
    \draw (70pt,70pt) -- (96pt,89pt);

    \draw (100pt,90pt) -- (176pt,90pt);
    \draw (100pt,90pt) -- (177pt,33pt);
    \draw (100pt,30pt) -- (176.5pt,87pt);
    \draw (100pt,10pt) -- (177pt,68pt);
    \draw (100pt,10pt) -- (206pt,10pt);

    \draw (180pt,30pt) -- (206.5pt,12pt);
    \draw (180pt,70pt) -- (206.5pt,78pt);
    \draw (180pt,90pt) -- (206.5pt,82pt);
    \end{scope}
    \draw[thick] (210pt,10pt) -- (230pt,10pt);
    \draw[thick] (210pt,80pt) -- (230pt,80pt);

    \draw[black, fill=white] (10pt,30pt) circle (1.0ex);
    \draw[black, fill=white] (40pt,10pt) circle (1.0ex);
    \draw[black, fill=white] (70pt,30pt) circle (1.0ex);
    \draw[black, fill=white] (100pt,10pt) circle (1.0ex);
    \draw[black, fill=white] (100pt,30pt) circle (1.0ex);

    \draw[black, fill=white] (10pt,70pt) circle (1.0ex);
    \draw[black, fill=white] (40pt,90pt) circle (1.0ex);
    \draw[black, fill=white] (70pt,70pt) circle (1.0ex);
    \draw[black, fill=white] (100pt,90pt) circle (1.0ex);

    \draw[black, fill=white] (180pt,70pt) circle (1.0ex);
    \draw[black, fill=white] (180pt,90pt) circle (1.0ex);
    \draw[black, fill=white] (210pt,80pt) circle (1.0ex);
    \draw[black, fill=white] (180pt,30pt) circle (1.0ex);
    \draw[black, fill=white] (210pt,10pt) circle (1.0ex);

\end{tikzpicture}
    \hspace{-12pt}
    }
    }{
    \vspace{-7pt}
    \caption{Example BSP schedule for a DAG. The labelled boxes only represent the computation phase of a superstep; the superstep itself also consists of the communication phase that follows.}
    \label{fig:BSPexample}
}
\hspace{16pt}
\capbtabbox{
  \centering
  \resizebox{0.46\textwidth}{!}{
  \hspace{-11pt}
    \begin{tabular}{r || c | c | c |}
         & \makecell{free comm. \vspace{7pt} \\ \scriptsize (no cost) \vspace{13pt} \\} & \makecell{simplified \vspace{-1pt} \\ comm.\ cost \vspace{1pt} \\ \scriptsize (any amount \vspace{-3pt} \\ \scriptsize of data in a \vspace{-3pt} \\ \scriptsize single step)} & \makecell{exact \vspace{-1pt} \\ comm.\ cost \vspace{1pt} \\ \scriptsize (depends on \vspace{-3pt}\\ \scriptsize data volume) \vspace{9pt}} \\
        \hhline{=||=|=|=|}
        \Gape[3pt]{\makecell{comp.\ \& \vspace{-2pt} \\ comm. \vspace{-2pt} \\ simulta- \vspace{-2pt} \\neously }} & \makecell{ classical \vspace{-1pt} \\  scheduling \vspace{1pt} \\ \scriptsize \cite[\dots]{DAG2proc1, DAG2proc2} } & \makecell{commdelay \vspace{5pt} \\ \scriptsize \cite[\dots]{commDdef1, commDappR1} \vspace{-5pt} } & \makecell{single-port \vspace{-1pt} \\ duplex \vspace{1pt} \\ \scriptsize \cite{SPD} } \\
        \hline
         \Gape[3pt]{\makecell{comp.\ \& \vspace{-2pt} \\ comm.\ in \vspace{-2pt} \\ separate \vspace{-2pt} \\ phases}} & \makecell{ classical \vspace{-1pt} \\ scheduling \vspace{1pt} \\ \scriptsize \cite[\dots]{DAG2proc1, DAG2proc2}} & \makecell{commdelay \vspace{-2pt} \\ \scriptsize with phases \vspace{2pt} \\ \scriptsize \cite{CDphases} } & \makecell{\textbf{BSP} \vspace{2pt} \\ \scriptsize \textbf{[novel for} \vspace{-1pt} \\ \scriptsize \textbf{DAGs]} \vspace{-3pt}} \\
        \hline
        \noalign{\smallskip}
    \end{tabular}
    \hspace{-20pt}
    }
}{
  \vspace{2pt}
  \caption{Simplified taxonomy of DAG scheduling models. The horizontal axis shows different models of capturing communication cost, whereas the vertical axis shows whether simultaneous computation and communication steps are allowed.}
  \label{tab:tax1}
}
\end{floatrow}
\end{figure}

Note that BSP has two major differences from this commdelay model. Firstly, commdelay allows a processor $p$ to execute computations and communications simultaneously, while in BSP, these are explicitly separated into a computation/communication phase. Moreover, commdelay in fact allows any number of values to be sent from $p_1$ to $p_2$ simultaneously, whereas BSP also considers the communication volume: sending $k$ values from $p_1$ to $p_2$ takes $k$ times as long as sending a single value. Previous work has already briefly considered the extension of commdelay with both of these modifications separately:
\begin{itemize}[leftmargin=15pt, itemsep=4pt, topsep=2pt]
\item The work of~\cite{CDphases} considers a variant of commdelay where computation and communication can only happen in separate phases, and studies how this relates to the base model.
\item The work of~\cite{SPD} introduces a \textit{single-port duplex} (SPD) model, which extends commdelay with communication volume: besides the assignment $t_{\!}:_{\!}V_{\!} \rightarrow _{\!} \mathbb{Z}^+$, a schedule must also specifically assign the send/recieve steps to disjoint time intervals for each processor.
\end{itemize}
The models above are summarized in Table~\ref{tab:tax1}, which shows that DAG scheduling in BSP indeed fills a natural place in this taxonomy. We note that while the models in the last column seemingly use a different method to capture communication cost (time intervals in SPD, $h$-relations in BSP), one can show that these two methods are in fact essentially equivalent; we discuss this in detail in Appendix \ref{app:taxproofs}.

%Note that in the last column of Table~\ref{tab:tax1}, the difference between SPD and BSP is, in fact, twofold. Firstly, BSP assumes that at any point in time, a processor can only do either computation or communication, but not both. Secondly, BSP assumes that the whole execution is divided into supersteps, which can be understood as a \textit{barrier synchronization} requirement: in order to communicate a value $v$, we first need a point in time (after computing $v$) when no processor is computing, and hence they can initialize the process of communicating $v$ (and possibly other values). By separating these two properties, we can extend our taxonomy into Table \ref{tab:tax2} to discuss some further model variants related to BSP.

In the last column of Table~\ref{tab:tax1}, the difference between SPD and BSP is, in fact, twofold. Firstly, BSP assumes that at any time, a processor can only do either computation or communication, but not both. Secondly, BSP assumes that the execution is divided into supersteps, which can be understood as a \textit{barrier synchronization} requirement: to communicate a value $v$, we first need a point in time (after computing $v$) when no processor is computing, and hence they can initiate the process of sending $v$ (and other values). By separating these two properties, we can extend our taxonomy into Table \ref{tab:tax2} to include some further model variants.

\begin{table*}
        \centering
        \setlength\tabcolsep{7pt}
        \resizebox{0.92\textwidth}{!}{
        \begin{tabular}{r | c || c | c | c |}   
             \multicolumn{2}{c||}{} & \makecell{free communication \smallskip \\ \small (no cost) \vspace{9pt} \\} & \makecell{simplified comm.\ cost \smallskip \\ \small (any amount of data \vspace{-2pt} \\ \small in a single step)} & \makecell{exact comm.\ cost \smallskip \\ \small (cost depends on \vspace{-2pt}\\ \small data volume)} \\
            \hhline{==||=|=|=|}
           \multirow{2}{*}{\Gape[10pt]{\makecell{comp.\ \& comm. \\ simultaneously}}} & \makecell{no \vspace{-2pt}\\ sync} & \Gape[4pt]{\makecell{\large classical \\ \large scheduling}} & \makecell{\large commdelay } & \makecell{\large single-port duplex } \\
            \cline{2-3}
            \cline{4-5}
            & sync & \Gape[4pt]{\makecell{\large classical \\ \large scheduling}} & \makecell{\large commdelay \\ \small with sync points } & \makecell{\large maxBSP } \\
            \cline{1-3}
            \cline{4-5}
            \multirow{2}{*}{\Gape[10pt]{ \makecell{comp.\ \& comm. \\ separately}}} &  \makecell{no \vspace{-2pt}\\ sync} & \Gape[4pt]{\makecell{\large classical \\ \large scheduling}}
            %& \makecell{\large commdelay \\ \small with timeouts }
            & \makecell{\large $\alpha_{\!}-_{\!}\beta\:$ \small with $\beta_{\!}=_{\!}0$ \smallskip\\ (or subset-CD)}
            %& \makecell{\large subset-BSP \smallskip\\ (or $\alpha_{\!}-_{\!}\beta$ model)} \\
            & \makecell{\large $\alpha_{\!}-_{\!}\beta\:$ \small with $\alpha_{\!}=_{\!}0$ \smallskip\\ (or subset-BSP)} \\
            \cline{2-3}
            \cline{4-5}
            & sync & \Gape[4pt]{\makecell{\large classical \\ \large scheduling}} & \makecell{\large commdelay \\ \small with phases } & \large BSP \\
            \hline
            \noalign{\smallskip}
        \end{tabular}
        }
        \caption{Extended table of DAG scheduling models. The vertical axis is split according to two properties: (i) whether computation and communication are allowed simultaneously, and (ii) whether barrier synchronization is required for communication.}
        \label{tab:tax2}
    \end{table*}

\begin{itemize}[leftmargin=15pt, itemsep=1pt, topsep=4pt]
\item If global synchronization is required (so we have supersteps), but processors can compute and communicate simultaneously: in fact, the original definition of BSP~\cite{BSPintro} implicitly assumes this setting, defining $C\,\!^{(s)}$ as the maximum of $C_{work}\,\!^{(s)}$ and $g \cdot C_{comm}\,\!^{(s)} + L$. Let us call this model \emph{maxBSP}. In contrast, recent textbooks on BSP~\cite{BSPbook2} apply the definition in Section \ref{sec:BSP}, where the two terms are summed up. Note that defining a reasonable maxBSP model for DAGs scheduling requires further consideration, to ensure that the computation and communication phases of each superstep are indeed parallelizable.
\item If processors can only either compute or communicate at a given time, but no synchronization is needed: one can interpret this as the $\alpha_{\!}-_{\!}\beta$ model~\cite{alphabeta1} with a choice of $\alpha_{\!}=_{\!}0$, although the definition of this model varies. This allows e.g.\ $p_1$ and $p_2$ to stop computing and exchange values, while the rest of the processors keep computing in the meantime.
\end{itemize}

We can also apply the same separation idea to obtain $4$ slightly different variants of the commdelay model; in contrast to this, classical scheduling remains identical in all $4$ cases.

\subsection{Optimum costs and further properties}

%\subparagraph*{Optimum costs and further properties.}

%One of the most natural questions regarding this taxonomy is how the optimum costs in these models relate to each other for a given DAG. Let us denote this optimum cost by \textsc{OPT}. Firstly, note that in any of the models, one of the processors must have a work cost of $\frac{n}{P}$ at least; this provides a lower bound on the optimum. Moreover, executing the entire DAG on a single processor (without communication) always yields a valid solution of cost $n$, hence providing a factor $P$ approximation of the optimum.

One natural question in this taxonomy is how the optimum costs in these models (denoted by \textsc{OPT}) relate to each other for a given DAG. Firstly, note that in any of the models, one of the processors must have a work cost of $\frac{n}{P}$ at least; this provides a lower bound. Moreover, executing the entire DAG on a single processor (without communication) always yields a valid solution of cost $n$, hence always approximating the optimum to a factor $P$.

\begin{proposition} \label{prop:approx}
We have $\frac{n}{P} \leq \textsc{OPT} \leq n$ in any of these models.
\end{proposition}

Next we compare the optimum cost in the two fundamental models, classical scheduling and commdelay, to the optimum in BSP (denoted by $\textsc{OPT}_{class}$, $\textsc{OPT}_{CD}$ and $\textsc{OPT}_{BSP}$, respectively, assuming $L\!=\!0$ in BSP). These clearly satisfy $\textsc{OPT}_{class} \leq \textsc{OPT}_{CD} \leq \textsc{OPT}_{BSP}$. Finding the maximal difference is more involved; however, note that e.g.\ Proposition~\ref{prop:approx} already implies that the optimum costs in any two models differ by at most a factor $P$.

\begin{lemma} \label{lem:horizOpt}
We have $\textsc{OPT}_{BSP} \leq P \cdot _{\!} \textsc{OPT}_{class}$ for any DAG and parameters $P$ and $g$. Moreover, $\textsc{OPT}_{CD} \leq (1+g) \cdot \textsc{OPT}_{class}$.
These bounds are essentially tight: there are DAG constructions with $\frac{\textsc{OPT}_{CD}}{\textsc{OPT}_{class}}_{\!}=_{\!}P$, $\: \frac{\textsc{OPT}_{BSP}}{\textsc{OPT}_{CD}}_{\!}=_{\!}P$, and $\:\frac{\textsc{OPT}_{CD}}{\textsc{OPT}_{class}}_{\!}=_{\!}(1_{\!}+_{\!}g_{\!}-_{\!}\varepsilon)$ for any $\varepsilon _{\!}>_{\!}0$.
\end{lemma}

Due to our focus on BSP, we also analyze the relation between the models in the last column of Table \ref{tab:tax2} in detail. Let us denote their optimum costs by $\textsc{OPT}_{SPD}$, $\textsc{OPT}_{mBSP}$, $\textsc{OPT}_{\beta}$ and $\textsc{OPT}_{BSP}$ from top to bottom, again for $L\!=\!0$. The restrictiveness of the models implies $\textsc{OPT}_{SPD\!} \leq _{\!} \textsc{OPT}_{mBSP\!} \leq _{\!} \textsc{OPT}_{BSP}$ and $\textsc{OPT}_{SPD\!} \leq _{\!} \textsc{OPT}_{\beta\!} \leq _{\!} \textsc{OPT}_{BSP}$. We complement this by the following observations.

\begin{theorem} \label{th:optcosts}
For any DAG and parameters $P$, $g$, the optimum cost between any two models in the last column of Table \ref{tab:tax2} differs by a factor $2$ at most; or equivalently, we have $\textsc{OPT}_{BSP} \, \leq \, 2 \cdot \textsc{OPT}_{SPD}$. 
Moreover, we show DAG constructions that prove (for any $\varepsilon _{\!} > _{\!} 0$)
\begin{itemize}[leftmargin=12pt, itemsep=5pt, topsep=4pt]
 \item a matching lower bound of $(2_{\!}-_{\!}\varepsilon)$ for the ratios $\frac{\textsc{OPT}_{\beta}}{\textsc{OPT}_{SPD}}$, $\, \frac{\textsc{OPT}_{BSP}}{\textsc{OPT}_{SPD}}$, $\, \frac{\textsc{OPT}_{\beta}}{\textsc{OPT}_{mBSP}}$ and $\frac{\textsc{OPT}_{BSP}}{\textsc{OPT}_{mBSP}}$,
 \item a slightly looser lower bound of $(\frac{3}{2}-\varepsilon)$ for the ratios $\frac{\textsc{OPT}_{mBSP}}{\textsc{OPT}_{SPD}}$, $\: \frac{\textsc{OPT}_{BSP}}{\textsc{OPT}_{\beta}}$ and $\frac{\textsc{OPT}_{mBSP}}{\textsc{OPT}_{\beta}}$.
\end{itemize}
\end{theorem}

Finally, we discuss a few further properties of the models; while these can be proven with rather simple constructions, they still provide some valuable insight. Firstly, we note that while all cells in the first column of Table \ref{tab:tax2} seem identical, in case of weighted DAGs, synchronization can in fact make a significant difference even in this classical model. Secondly, we briefly analyze how duplication affects the different models in the taxonomy.

\begin{proposition} \label{lem:PRAMweights}
In classical scheduling with work weights, barrier synchronization can increase the optimum cost (by a factor $(2_{\!}-_{\!}\varepsilon)$ for any $\varepsilon _{\!}>_{\!}0$).
\end{proposition}

\begin{proposition} \label{lem:recomp}
Duplication can reduce the optimum cost in the following models:
\begin{itemize}[leftmargin=13pt, itemsep=2pt, topsep=2pt]
\item all models with communication cost (middle and right-hand column of Table \ref{tab:tax2}),
\item classical scheduling models with barrier synchronization, but only if we have work weights.
\end{itemize}
\end{proposition}

\section{Communication models within BSP} \label{sec:commodels}

For the rest of the paper, we focus on the BSP model. However, even within BSP, there are some different options to model the communication rules, and these seemingly small changes in the model definition can have a significant effect on the properties of the model.

We discuss two modelling choices that can be combined to form $4$ different \textit{communication models within BSP}. We name these models in Table~\ref{tab:commodels}. Our results in Sections \ref{sec:CS} and \ref{sec:ILP} will show that these submodels can indeed influence some properties of the problem. We note that these modelling choices only arise in models like BSP, which capture communication volume; in e.g.\ classical scheduling or commdelay, these variants are all equivalent.

\subparagraph*{Free movement of data.} \label{sec:commod_free}

For simplicity, our base BSP model assumed $\pi(v)=p_1$ for all $(v, p_1, p_2, s) \in \Gamma$, i.e.\ values are always sent from the processor where they were computed. In practice, there is no reason to make this restriction; to transfer a value from $p_1$ to $p_2$, one might as well send it from $p_1$ to a third processor $p_3$ first, and then from $p_3$ to $p_2$.

Moreover, there are simple examples where such unrestricted movement of data between processors can indeed result in a lower communication cost altogether. Consider the BSP schedule in Figure~\ref{fig:freedata} with $P\!=\!3$ processors. This schedule has a node that is computed on $p_3$ in superstep $1$, but later only needed on $p_1$ in superstep $3$. In case of direct transfer, we can send this value from $p_3$ to $p_1$ in either superstep $1$ or $2$; however, $p_1$ must already receive a value in superstep $1$, and $p_3$ must already send a value in superstep $2$, so both of these choices increase the communication cost in one of the supersteps. On the other hand, with free data movement, we can send the value from $p_3$ to $p_2$ in superstep $1$, and then from $p_2$ to $p_1$ in superstep $2$, without increasing the communication cost in either superstep.

\subparagraph*{Singlecast or broadcast.} \label{sec:commod_bc}

Another interesting question is what happens if processor $p$ wants to send a single value $v$ to multiple other processors $p_1, ..., p_k$ in the same superstep. In our base model, this requires a separate entry $(v, p, p_i, s)$ for all $i _{\!} \in _{\!} [k]$, and hence contributes $k$ units to the send cost $C_{send}\,\!^{(s, p)}$. This is a reasonable assumption e.g.\ if the communication topology is a fully connected graph; in this case, $p$ needs to send this value over $k$ distinct network links. However, in other cases, it is more reasonable to only charge a single unit of send cost for this, i.e.\ to assume that data transfers are \emph{broadcast operations}, and hence the values can be received by any number of processors. This can correspond to e.g.\ a star-shaped topology with a single communication device in the middle.

\renewcommand{\arraystretch}{1.5}
\begin{figure}
\begin{floatrow}
\ffigbox[0.43\textwidth]{
    \centering
    \resizebox{0.44\textwidth}{!}{
    \hspace{-4pt}
    \begin{tikzpicture}
	
    \begin{scope}[very thick, gray]
    \draw (-10pt,0pt) rectangle (40pt,20pt);
    \draw (-10pt,40pt) rectangle (40pt,60pt);
    \draw (-10pt,80pt) rectangle (40pt,100pt);
    \draw (70pt,0pt) rectangle (120pt,20pt);
    \draw (70pt,40pt) rectangle (120pt,60pt);
    \draw (70pt,80pt) rectangle (120pt,100pt);
    \draw (150pt,0pt) rectangle (200pt,20pt);
    \draw (150pt,40pt) rectangle (200pt,60pt);
    \draw (150pt,80pt) rectangle (200pt,100pt);
    \end{scope}
    
    \node[anchor=center, gray] at (15pt,-6pt) {\small \textbf{$p_3$, $\,s\!=\!1$}}; 
    \node[anchor=center, gray] at (15pt,34pt) {\small \textbf{$p_2$, $\,s\!=\!1$}};
    \node[anchor=center, gray] at (15pt,74pt) {\small \textbf{$p_1$, $\,s\!=\!1$}};

    \node[anchor=center, gray] at (95pt,-6pt) {\small \textbf{$p_3$, $\,s\!=\!2$}}; 
    \node[anchor=center, gray] at (95pt,34pt) {\small \textbf{$p_2$, $\,s\!=\!2$}};
    \node[anchor=center, gray] at (95pt,74pt) {\small \textbf{$p_1$, $\,s\!=\!2$}};

    \node[anchor=center, gray] at (175pt,-6pt) {\small \textbf{$p_3$, $\,s\!=\!3$}}; 
    \node[anchor=center, gray] at (175pt,34pt) {\small \textbf{$p_2$, $\,s\!=\!3$}};
    \node[anchor=center, gray] at (175pt,74pt) {\small \textbf{$p_1$, $\,s\!=\!3$}};

    \begin{scope}[arrows=-stealth]
    \draw[thin] (2pt,10pt) -- (12pt,10pt);
    \draw[thin] (15pt,10pt) -- (28pt,10pt);

    \draw[thin] (82pt,10pt) -- (92pt,10pt);
    \draw[thin] (95pt,10pt) -- (108pt,10pt);

    \draw[thin] (2pt,50pt) -- (12pt,50pt);
    \draw[thin] (15pt,50pt) -- (28pt,50pt);

    \draw[thin] (162pt,50pt) -- (172pt,50pt);
    \draw[thin] (175pt,50pt) -- (188pt,50pt);

    \draw[thin] (82pt,90pt) -- (92pt,90pt);
    \draw[thin] (95pt,90pt) -- (108pt,90pt);

    \draw[thin] (162pt,90pt) -- (172pt,90pt);
    \draw[thin] (175pt,90pt) -- (188pt,90pt);
    
    \draw[thick] (15pt,50pt) -- (94pt,86pt);
    \draw[thick] (95pt,10pt) -- (174pt,46pt);
    \draw[thick] (15pt,10pt) -- (174pt,86pt);
    
    \end{scope}

    \node[anchor=center] at (-1pt,6pt) {\small \textbf{...}};
    \draw[black, fill=white] (15pt,10pt) circle (0.85ex);
    \node[anchor=center] at (31pt,6pt) {\small \textbf{...}};

    \node[anchor=center] at (79pt,6pt) {\small \textbf{...}};
    \draw[black, fill=white] (95pt,10pt) circle (0.85ex);
    \node[anchor=center] at (111pt,6pt) {\small \textbf{...}};

    \node[anchor=center] at (-1pt,46pt) {\small \textbf{...}};
    \draw[black, fill=white] (15pt,50pt) circle (0.85ex);
    \node[anchor=center] at (31pt,46pt) {\small \textbf{...}};

    \node[anchor=center] at (159pt,46pt) {\small \textbf{...}};
    \draw[black, fill=white] (175pt,50pt) circle (0.85ex);
    \node[anchor=center] at (191pt,46pt) {\small \textbf{...}};

    \node[anchor=center] at (79pt,86pt) {\small \textbf{...}};
    \draw[black, fill=white] (95pt,90pt) circle (0.85ex);
    \node[anchor=center] at (111pt,86pt) {\small \textbf{...}};

    \node[anchor=center] at (159pt,86pt) {\small \textbf{...}};
    \draw[black, fill=white] (175pt,90pt) circle (0.85ex);
    \node[anchor=center] at (191pt,86pt) {\small \textbf{...}};

\end{tikzpicture}
    \hspace{-5pt}
    }
    }{
    \vspace{-16pt}
    \caption{Example BSP schedule where free data movement allows a lower communication cost than direct data transfer.}
    \label{fig:freedata}
}
\hspace{20pt}
\capbtabbox[0.48\textwidth]{
  \centering
  \resizebox{0.48\textwidth}{!}{
  \hspace{-6pt}
    \begin{tabular}{r|c|c|}
      & Singlecast & Broadcast \\
      \hline
      \makecell{Direct \\ transfer}& \Gape[5pt]{\makecell{\medskip \large \textbf{DS model} \normalsize \\ \smallskip CS: \textit{open problem}\\ ILP: $O(n _{\!} \cdot _{\!} P _{\!} \cdot _{\!} S)$ \textit{vars} }} & \Gape[5pt]{\makecell{\medskip \large \textbf{DB model} \normalsize \\ \smallskip CS: \textit{NP-hard} \\ ILP: $O(n _{\!} \cdot _{\!} P _{\!} \cdot _{\!} S)$ \textit{vars} }} \\
      \hline
      \makecell{Free data \\ movement} & \Gape[5pt]{\makecell{\medskip \large \textbf{FS model} \normalsize \\ \smallskip CS: \textit{NP-hard} \\ ILP: $O(n _{\!} \cdot _{\!} P^2 _{\!} \cdot _{\!} S)$ \textit{vars} }} & \Gape[5pt]{\makecell{\medskip \large \textbf{FB model} \normalsize \\ \smallskip CS: \textit{NP-hard} \\ ILP: $O(n _{\!} \cdot _{\!} P _{\!} \cdot _{\!} S)$ \textit{vars} }} \\
     \hline
    \end{tabular}
    \hspace{-8pt}
    }
}{
  \vspace{-1pt}
  \caption{Different communication models within BSP, and their properties established in Sections \ref{sec:CS} and \ref{sec:ILP}. Our main complexity results (Theorems~\ref{th:chains}--\ref{th:apx}) hold in any of these models.}
    \label{tab:commodels}
}
\end{floatrow}
\end{figure}

\section{NP-hardness} \label{sec:hardness}

One fundamental question regarding our scheduling problem is its complexity. Unsurprisingly, the problem is NP-hard in general DAGs. However, this raises a natural follow-up question, which has also been studied in simpler models: in which subclasses of DAGs is the problem still solvable in polynomial time, and when does it become NP-hard?

The simplest non-trivial subclass of DAGs is \emph{chain DAGs}, where both the indegree and outdegree of nodes is at most $1$. This subclass has been analyzed in different scheduling models before~\cite{chains, shopsched}, and has even been studied in BSP~\cite{chainsHard}, under slightly different assumptions (see Appendix \ref{app:hardness}). We also consider a slightly more realistic version of this subclass: we say that a DAG is a \emph{connected chain DAG} if it can be obtained by adding an extra source node $v_0$ to a chain DAG, and drawing an edge from $v_0$ to the first node in every chain. We show that for these relatively simple classes of DAGs, the optimal BSP schedule can still be found in polynomial time. The key observation in the proof is that the optimal BSP schedule in chain DAGs always consists of at most $P$ supersteps; this allows us to find the optimum through a rather complex dynamic programming approach.

\begin{theorem} \label{th:chains}
The BSP scheduling problem can be solved in polynomial time in $n$ for chain DAGs and connected chain DAGs.
\end{theorem}

On the other hand, it turns out that BSP scheduling already becomes NP-hard for slightly more complex DAGs. In particular, we consider (i) DAGs with height only $2$, where \emph{height} is the number of nodes in the longest directed path, and (ii) \emph{in-trees}, which are DAGs where every node has outdegree at most $1$. We prove that in these (still relatively simple) classes, the problem is already NP-hard. Since the scheduling problem is known to be polynomially solvable for these specific classes in simpler models~\cite{oppforest, commDinf1}, this shows that our more advanced model indeed comes at the cost of increased complexity.

\begin{theorem} \label{th:2level}
BSP scheduling is already NP-hard if restricted to DAGs of height $2$.
\end{theorem}

\begin{theorem} \label{th:trees}
BSP scheduling is already NP-hard if restricted to in-trees.
\end{theorem}

\renewcommand*{\proofname}{Proof sketch}

\begin{proof}
The proof of Theorem~\ref{th:2level} uses a reduction from the $k$-clique problem, and can be loosely understood as an adaptation of the reduction approach for partitioning in~\cite{hyperDAG} to our setting. Intuitively, the second level of the DAG consists of gadgets representing each node of the input graph, while the first level consists of gadgets representing the edges; each edge gadget is connected to the two incident node gadgets. Our construction ensures that at most $k$ of the node gadgets can be assigned to a given processor $p$ without incurring a too large work cost, and that each edge gadget incurs a communication step exactly if one of its incident node gadgets is not assigned to $p$. With the appropriate cost limit $C_0$, the DAG only admits a valid schedule if there are $k$ original nodes that induce ${k \choose 2}$ edges.

Theorem~\ref{th:trees} is our most technical proof, requiring $P\!=\!16$ processors; we only provide a brief intuitive overview here. Firstly, our choice of $C_0$ ensures that we can only have $\frac{n}{P}$ computations on any processor, and only $d$ communication steps altogether (for some $d$). 
Then on the one hand, our construction contains a large gadget that occupies a single processor $p$ for the entire schedule (otherwise the communication cost is too high), and $d$ further gadgets that each need to send a value to $p$. Any schedule needs to manage these gadgets carefully to ensure that one of these $d$ values is already available every time we do a communication step; this ensures that the communications can only happen at specific times at the earliest. 
On the other hand, we use critical paths to ensure that specific nodes need to be computed by a given time step. We then attach further gadgets to given segments of these paths, which forces us to split the work between multiple processors (and thus communicate) to satisfy these deadlines; hence communications must happen at specific times at the latest.
Together, these define the exact times when we need to communicate, otherwise the cost exceeds $C_0$. Due to this, there is essentially only one way to develop a valid schedule in our DAG. The underlying reduction then uses the $3$-partition problem with gadgets representing each input number $a_i$, and ensures that it is only possible to satisfy each communication deadline if the numbers $a_i$ are sorted into triplets that each sum up to a given value.
\end{proof}

\section{Inapproximability}

Given the NP-hardness of the problem, another follow-up question is whether we can at least approximate the optimal solution in polynomial time. However, we show that on general DAGs, BSP scheduling is APX-hard: there is a constant $\epsilon>0$ such that it is already NP-hard to approximate the optimum cost to a $(1+\epsilon)$ factor. To our knowledge, for the case of constant $P$, no similar hardness results are known for other DAG scheduling problems.

\begin{theorem} \label{th:apx}
BSP scheduling is APX-hard: there exists a constant $\epsilon>0$ such that it is NP-hard to approximate the optimal scheduling cost to a $(1+\epsilon)$ factor, already for $P=2$.
\end{theorem}

\begin{proof}
We provide a reduction from MAX-$3$SAT($B$), i.e.\ maximizing the satisfied clauses in a $3$SAT formula where each variable occurs at most $B$ times; this is already known to be APX-hard for $B \in O(1)$ \cite{maxsat3, maxsat4}. Our construction consists of a separate gadget for each clause and each variable of the given $3$SAT formula. On a high level, the variable gadgets contain two separate subgadgets that we need to assign to the two different processors; this choice corresponds to setting the variable to true or false in the underlying formula. The clause gadgets then contain subgadgets corresponding to the $3$ literals in the clause, which need to be assigned according to the truth value chosen in the corresponding variable gadget. Whenever a clause is unsatisfied, the corresponding gadgets in the schedule incur a slightly higher work cost than a satisfied clause. With this, altogether, the total cost of the BSP schedule has a $\Theta(n)$ term that is proportional to the number of unsatisfied clauses, which allows us to complete the reduction.

The technical part of the proof is to show that any reasonable schedule in our DAG is indeed structured as described above, representing a solution of MAX-$3$SAT($B$). Intuitively, our construction ensures that using any other kind of subschedule in our gadgets results in a higher work or communication cost. More formally, we can always apply a sequence of transformation steps to convert any other BSP schedule into a schedule that represents a valid $3$SAT assignment, while only decreasing the cost of the solution.
\end{proof}

\section{Problem within a problem: communication scheduling} \label{sec:CS}

Since $h$-relations are a complex communication metric, it is natural to wonder if the hardness of BSP scheduling lies only within the assignment of nodes to processors and steps (as in simpler models), or if the scheduling of communication steps also adds another layer of complexity on this. More specifically, assume that the nodes are already assigned to processors and supersteps, but we still need to decide when the values are communicated between the chosen processors; is this subproblem easier to solve? 

\begin{definition}
In the \emph{communication scheduling (CS) problem}, we are given a fixed $\pi$ and $\tau$, and we need to find a communication schedule $\Gamma$ that minimizes the resulting BSP cost.
\end{definition}

We implicitly assume that $\pi$ and $\tau$ ensure that there is a feasible communication schedule, i.e.\ for all $(u,v) \in E$, we have $\tau(u) \leq \tau(v)$ if $\pi(u)=\pi(v)$, and $\tau(u) < \tau(v)$ if $\pi(u) \neq \pi(v)$.

As a heuristic approach, one might consider an eager or lazy communication policy, i.e.\ to communicate each value immediately after it is computed (eager), or only in the last possible superstep before it is needed (lazy). However, both of these schedules might be suboptimal. Consider the example in Figure~\ref{fig:CS}, where most of the computed values need to be sent immediately in the same superstep, but $p_2$ also computes a value in $s\!=\!1$ that is only needed at $p_1$ in $s\!=\!4$, so $p_2$ can choose to send it in superstep $1$, $2$ or $3$. The eager and lazy policies send this value in superstep $1$ and $3$, respectively, but these both increase the cost of the given $h$-relation; in contrast to this, sending in superstep $2$ comes at no further cost.

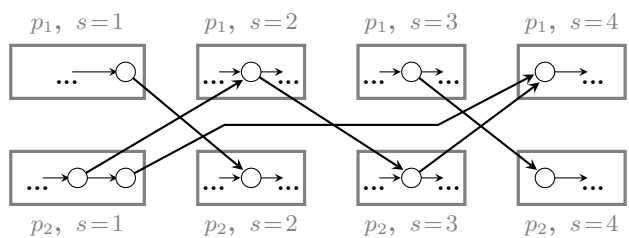
\begin{figure}
    \centering
    \begin{tikzpicture}
	
    \begin{scope}[very thick, gray]
    \draw (-10pt,0pt) rectangle (40pt,20pt);
    \draw (-10pt,40pt) rectangle (40pt,60pt);
    \draw (60pt,0pt) rectangle (100pt,20pt);
    \draw (60pt,40pt) rectangle (100pt,60pt);
    \draw (120pt,0pt) rectangle (160pt,20pt);
    \draw (120pt,40pt) rectangle (160pt,60pt);
    \draw (180pt,0pt) rectangle (220pt,20pt);
    \draw (180pt,40pt) rectangle (220pt,60pt);
    \end{scope}
    
    \node[anchor=center, gray] at (15pt,-8pt) {\small \textbf{$p_2$, $\,s\!=\!1$}}; 
    \node[anchor=center, gray] at (15pt,68pt) {\small \textbf{$p_1$, $\,s\!=\!1$}};

    \node[anchor=center, gray] at (80pt,-8pt) {\small \textbf{$p_2$, $\,s\!=\!2$}}; 
    \node[anchor=center, gray] at (80pt,68pt) {\small \textbf{$p_1$, $\,s\!=\!2$}};
    
    \node[anchor=center, gray] at (140pt,-8pt) {\small \textbf{$p_2$, $\,s\!=\!3$}}; 
    \node[anchor=center, gray] at (140pt,68pt) {\small \textbf{$p_1$, $\,s\!=\!3$}};

    \node[anchor=center, gray] at (200pt,-8pt) {\small \textbf{$p_2$, $\,s\!=\!4$}}; 
    \node[anchor=center, gray] at (200pt,68pt) {\small \textbf{$p_1$, $\,s\!=\!4$}};

    \begin{scope}[arrows=-stealth]
    \draw[thin] (2pt,10pt) -- (12pt,10pt);
    \draw[thin] (15pt,10pt) -- (30pt,10pt);

    \draw[thin] (68pt,10pt) -- (76pt,10pt);
    \draw[thin] (80pt,10pt) -- (92pt,10pt);

    \draw[thin] (128pt,10pt) -- (136pt,10pt);
    \draw[thin] (140pt,10pt) -- (152pt,10pt);

    \draw[thin] (190pt,10pt) -- (206pt,10pt);

    \draw[thin] (13pt,50pt) -- (30pt,50pt);

    \draw[thin] (68pt,50pt) -- (76pt,50pt);
    \draw[thin] (80pt,50pt) -- (92pt,50pt);

    \draw[thin] (128pt,50pt) -- (136pt,50pt);
    \draw[thin] (140pt,50pt) -- (152pt,50pt);

    \draw[thin] (190pt,50pt) -- (206pt,50pt);

    \draw[thick] (33pt,50pt) -- (77pt,13pt);
    \draw[thick] (80pt,50pt) -- (137pt,13pt);
    \draw[thick] (140pt,50pt) -- (187pt,13pt);

    \draw[thick] (15pt,10pt) -- (77pt,47pt);
    \draw[thick] (33pt,10pt) -- (70pt,30pt) -- (150pt,30pt) -- (186pt,49pt);
    \draw[thick] (140pt,10pt) -- (188pt,46pt);
    
    \end{scope}

    \node[anchor=center] at (0pt,6pt) {\small \textbf{...}};
    \draw[black, fill=white] (15pt,10pt) circle (0.85ex);
    \draw[black, fill=white] (33pt,10pt) circle (0.85ex);

    \node[anchor=center] at (66pt,6pt) {\small \textbf{...}};
    \draw[black, fill=white] (80pt,10pt) circle (0.85ex);
    \node[anchor=center] at (94pt,6pt) {\small \textbf{...}};

    \node[anchor=center] at (126pt,6pt) {\small \textbf{...}};
    \draw[black, fill=white] (140pt,10pt) circle (0.85ex);
    \node[anchor=center] at (154pt,6pt) {\small \textbf{...}};

    \draw[black, fill=white] (190pt,10pt) circle (0.85ex);
    \node[anchor=center] at (208pt,6pt) {\small \textbf{...}};

    \node[anchor=center] at (11pt,46pt) {\small \textbf{...}};
    \draw[black, fill=white] (33pt,50pt) circle (0.85ex);

    \node[anchor=center] at (66pt,46pt) {\small \textbf{...}};
    \draw[black, fill=white] (80pt,50pt) circle (0.85ex);
    \node[anchor=center] at (94pt,46pt) {\small \textbf{...}};

    \node[anchor=center] at (126pt,46pt) {\small \textbf{...}};
    \draw[black, fill=white] (140pt,50pt) circle (0.85ex);
    \node[anchor=center] at (154pt,46pt) {\small \textbf{...}};

    \draw[black, fill=white] (190pt,50pt) circle (0.85ex);
    \node[anchor=center] at (208pt,46pt) {\small \textbf{...}};

\end{tikzpicture}
    \vspace{-5pt}
    \caption{Illustration of the CS problem: both eager and lazy scheduling can be suboptimal.}
    \label{fig:CS}
\end{figure}

%This suggest that CS is indeed an interesting problem on its own. From the theoretical side, it is important to understand if our more realistic communication cost is, intuitively speaking, complex enough to already make the CS problem NP-hard. From the practical side, the CS problem has significantly less degree of freedom; as such, a heuristic algorithm might prefer to first find a good initial BSP schedule, and then try to further improve this by fixing $\pi$ and $\tau$, and reducing the communication cost separately.

This suggest that CS is indeed an interesting problem from a theoretical perspective. Furthermore, from the practical side, the CS problem has significantly less degree of freedom; as such, a heuristic scheduler could aim to first find a good initial BSP schedule, and then try to improve this by fixing $\pi$ and $\tau$, and reducing the communication cost separately.

For the simplest case of $P_{\!}=_{\!}2$, one can show that a greedy approach already finds the optimal solution (in any model of Table~\ref{tab:commodels}, since these are all equivalent for $P_{\!}=_{\!}2$).

\begin{lemma} \label{th:CSDP}
The CS problem can be solved in polynomial time for $P_{\!}=_{\!}2$ processors.
\end{lemma}

%Interestingly, for general $P$, it turns out that the hardness of CS may depend on the communication model. In particular, we prove that CS is already NP-hard in case of the DB, FS or FB models. However, the same proof approach does not work in case of DS; we leave it as an open question whether communication scheduling is also NP-hard in DS.

Interestingly, for general $P$, the hardness of CS may depend on the communication model. In particular, we prove that CS is already NP-hard in case of the DB, FS or FB models; however, we leave it as an open question whether the problem is also NP-hard in DS.

\begin{theorem} \label{th:commsched}
The CS problem is NP-hard in the DB, FS and FB models.
\end{theorem}

\renewcommand*{\proofname}{Proof sketch}

\begin{proof}
The proof uses a reduction from the $3$D-matching problem: given three sets $X$, $Y$, $Z$ of size $N$, and $M$ triplets from $X _{\!} \times _{\!} Y _{\!} \times _{\!} Z$, we need to select $N$ triplets that form a disjoint cover. We convert this into a CS problem where each triplet is represented by a value which needs to be sent from $p_0$ to several other processors, with the concrete deadlines depending on the given triplet. The construction consists of two parts. The first part allows us to send exactly $(M-N)$ values from $p_0$ to all other processors (the triplets that are not chosen) within the allowed cost. The second part allows us to send $N$ further values from $p_0$ to the other processors, but only if the corresponding triplets are disjoint.

Intuitively, the proof works in DB, FS and FB because sending a value from $p_1$ to $p_2$ and from $p_1$ to $p_3$ are not independent operations: with broadcast, a value sent from $p_1$ can be received by both $p_2$ and $p_3$ at the same time, and with free data movement, $p_1$ can send a value to another processor $p_a$, which relays this to both $p_2$ and $p_3$. These are crucial to ensure that $p_0$ indeed sends the same $(M-N)$ values to all other processors in the first part. In contrast to this, the communication steps from $p_1$ to $p_2$ and from $p_1$ to $p_3$ are essentially independent from each other in DS, so the same approach cannot be applied.
\end{proof}

On the other hand, if we have communication weights, then it becomes relatively simple to reduce CS to standard packing problems already for $P _{\!} = _{\!} 2$ processors.

\begin{lemma} \label{lem:CSweights}
With communication weights, the CS problem is already NP-hard for $P_{\!}=_{\!}2$.
\end{lemma}

\section{Formulation as an ILP problem} \label{sec:ILP}

Finally, we present a straightforward approach to formulate BSP scheduling as an Integer Linear Programming (ILP) problem. Since today's ILP solvers can often solve even fairly large instances in reasonable time, this naive ILP formulation may already be a viable approach in several applications, especially if the computation is modelled as a DAG at relatively coarse granularity. A similar approach with ILP solvers has already provided remarkable empirical results  for various combinatorial problems~\cite{jenneskens22}. Our ILP formulation also generalizes very naturally to the extensions with node weights and/or duplication.

\begin{proposition} \label{prop:ILP}
BSP scheduling can be formulated as an ILP problem on $O(n \cdot P \cdot S)$ variables in the DS, DB and FB models, and  on $O(n \cdot P^2 \cdot S)$ variables in the FS model.
\end{proposition}

\begin{proof}
For all $v _{\!} \in _{\!} V$, $p _{\!} \in _{\!} [P]$, $s _{\!} \in _{\!} [S]$, we introduce two binary variables: $\textsc{comp}_{v, p, s}$ indicates whether value $v$ is computed on processor $p$ in superstep $s$, while $\textsc{pres}_{v, p, s}$ indicates whether the output of $v$ is already present on processor $p$ after the computation phase of superstep $s$. We can then already express the work costs $C_{work}\,\!^{(s,p)}$ by a summation, and $C_{work}\,\!^{(s)}$ by requiring $C_{work}\,\!^{(s)} \geq C_{work}\,\!^{(s,p)}$ for all $p$. The main difference between the models is how we capture the actual communication steps:
\begin{itemize}[leftmargin=15pt, itemsep=4pt, topsep=2pt]
\item In broadcast models (DB and FB), we simply add binary variables $\textsc{sent}_{v, p, s}$ and $\textsc{rec}_{v, p, s}$ to indicate if processor $p$ sends/receives value $v$ in superstep $s$. The constraints here must ensure that if $\textsc{rec}_{v, p_1, s}=1$, then $\textsc{send}_{v, p_2, s}=1$ for some other processor $p_2$.
\item In FS, we need to add binary variables $\textsc{comm}_{v, p_1, p_2, s}$ to indicate the concrete data transfer, i.e.\ whether value $v$ is sent from processor $p_1$ to processor $p_2$ in superstep $s$. The validity of these operations can simply be checked using the variables $\textsc{pres}_{v, p, s}$.
\item Finally, in DS, a naive solution would also use the variables $\textsc{comm}_{v, p_1, p_2, s}$. However, since here only processor $\pi(v)$ is allowed to send $v$, we can use some auxiliary variables to also capture this case with $O(n \cdot P \cdot S)$ variables, similarly to DB and FB.
\end{itemize}

Note that a value $v$ is present on processor $p$ in superstep $s$ if it was already present before, it was computed in this superstep, or it was received in the previous superstep; hence e.g.\ in the broadcast models, the constraint $\textsc{pres}_{v, p, s} \, \leq \, \textsc{pres}_{v, p, (s-1)} + \textsc{comp}_{v, p, s} + \textsc{rec}_{v, p, (s-1)}$ ensures the validity of the variables $\textsc{pres}_{v, p, s}$. The precedence constraints in the DAG can be encoded by requiring $\textsc{comp}_{v, p, s} \leq \textsc{pres}_{u, p, s}$ for all $(u,v) \in E$. The entire set of variables and constraints for each model is discussed in Appendix \ref{app:ILP}.
\end{proof}

\newpage

\bibliography{references}

\newpage

\appendix

\section{Detailed definition of different model variants} \label{app:def}

We first give a more formal definition of the numerous scheduling models that were analyzed (or only briefly mentioned) in the paper. Note that besides BSP, the simplest scheduling models from the literature (classical scheduling and commdelay) were already defined in Section~\ref{sec:tax}.

We begin with the alternative communication models within BSP (the DB, FS and FB models), since these play a central role in our theoretical results in Sections~\ref{sec:CS} and \ref{sec:ILP}. We then continue with the remaining models that are only loosely related to BSP. Finally we discuss some extensions and alternative variants of the BSP scheduling problem.

\subsection{DB, FS and FB models}

To define the BSP scheduling problem with broadcast communication (and hence switch from our initial definition of the DS model to DB), we need to redefine the send cost $C_{sent}\,\!^{(s,p)}$ as
\[ C_{sent}\,\!^{(s,p)}=|\{ v \in V \: | \: \exists \, p' \in [P] \text{ such that } (v, p, p', s) \in \Gamma \}| \, , \]
while the rest of the definitions remain unchanged.

To allow free movement of data (switch from the DS model to FS), we need to introduce auxiliary variables similarly to the \textsc{pres}$_{v,p,s}$ variables in the ILP formulation; to avoid confusion, we develop a different notation for these in this context, referring to them as $\mu(v, p, s)$. Formally, we let the auxiliary variable $\mu(v, p, s) \in \{0, 1\}$ indicate whether the output of $v$ is already present on processor $p$ at the end of the computational phase of superstep $s$; that is, we have $\mu(v, p, s)=1$ if either $\pi(v)=p$ and $\tau(v) \leq s$, or if there is a $(v, p' p, s') \in \Gamma$ with $s' < s$. In the validity conditions of Section \ref{sec:BSP}, Condition (i) will now require for all $(u,v) \in E$ to have $\mu(u, \pi(v), \tau(v))=1$, and Condition (ii) will state that if $(v, p_1, p_2, s) \in \Gamma$, then we must have $\mu(v, p_1, s)=1$.

Note that the modified definition for broadcasting and for free data movement can be combined in a straightforward way to obtain the definition of the FB model, since the corresponding modifications affect different aspects of the model definition anyway.

\subsection{Single-port duplex model}

The SPD model has been studied in the work of \cite{SPD}; we slightly change this definition to adapt it to node-based communication.

Similarly to commdelay, an SPD schedule has a constant parameter $g$, and it consists of an assignment $\pi: V \rightarrow [P]$ and $t: V \rightarrow \mathbb{Z}^+$, but also a communication schedule $\Gamma \subset V \times [P] \times [P] \times \mathbb{Z}^+$, where $(v, p_1, p_2, t_0) \in \Gamma$ describes that value $v$ is sent from $p_1$ to $p_2$ at time step $t_0$. These communication steps in $\Gamma$ must satisfy two conditions for validity: the values that are sent must already be computed (i.e.\ if $(v, p_1, p_2, t_0) \in \Gamma$ then $t(v) \leq t_0$), and each processor can only send and receive a single value at a given time (i.e.\ if $(v_1, p_{1,1}, p_{2,1}, t_1), (v_2, p_{1,2}, p_{2,2}, t_2) \in \Gamma$ and the intervals $(t_1, t_1+g]$ and $(t_2, t_2+g]$ are not disjoint, then we must have $p_{1,1} \neq p_{1,2}$ and $p_{2,1} \neq p_{2,2}$).

For a correct SPD schedule, we still require that a processor only computes a single node at a time ($\nexists \, u, v \in V$ with $\pi(u)=\pi(v)$, $t(u)=t(v)$). For all $(u,v) \in E$, if $\pi(u)=\pi(v)$, then we still need to have $t(u)<t(v)$. For edges $(u,v) \in E$ with $\pi(u) \neq \pi(v)$, we require that there is a $(u, \pi(u), \pi(v), t_0) \in \Gamma$ with $t_0+g < t(v)$. As before, the cost of a schedule is simply its makespan $\max_{v \in V\,} t(v)$.

\subsection{MaxBSP}

As discussed in Section~\ref{sec:tax}, a reasonable interpretation of maxBSP for DAG scheduling requires some consideration. For instance, in the example in Figure~\ref{fig:maxBSP}, $v_2$ needs to be the last node to be computed on processor $p_1$ in superstep $2$, and $v_3$ must be the first node to be computed on $p_2$ in superstep $3$ in any case; this shows that the communication phase of superstep $2$ (sending $v_2$ from $p_1$ to $p_2$) cannot happen simultaneously with the computation phase of either superstep $2$ or $3$.

As such, simply applying the maxBSP cost function leads to an inaccurate model; we also need to consider which computations and communications can happen at the same time. For instance, $v_1$ is already computed in superstep $1$, and only required in superstep $3$, so it can indeed happen simultaneously with the computations in superstep $2$.

In general, DAG scheduling in maxBSP means that each superstep $s$ consists of a simultaneous computation and communication phase, and we can only send a value $v$ from $p_1$ to $p_2$ if $v$ was computed on $p_1$ in superstep $(s-1)$ or before, and is only needed on $p_2$ in superstep $(s+1)$ or after. For a formal definition of this, we only need minor modifications in the original BSP definition. In particular, in validity condition (ii), we need to require that if $(v, p_1, p_2, s) \in \Gamma$, then besides $p_1=\pi(v)$, we must have $\tau(v) < s$ with strict inequality, since the communication phase of $s$ now overlaps with its computation phase. Besides this, we simply need to adjust the cost of a superstep. One possible choice for this is
\[ C\,\!^{(s)} = \max \left( C_{work}\,\!^{(s)} \: , \: g \cdot C_{comm}\,\!^{(s)} + L \right) \, , \]
but as a natural alternative, we could also have
\[ C\,\!^{(s)} = \max \left( C_{work}\,\!^{(s)} \: , \: g \cdot C_{comm}\,\!^{(s)} \right) + L \, . \]

\begin{figure}
    \centering
    \begin{tikzpicture}
	
    \begin{scope}[very thick, gray]
    \draw (0pt,-6pt) rectangle (50pt,20pt);
    \draw (0pt,40pt) rectangle (50pt,66pt);
    \draw (80pt,-6pt) rectangle (150pt,20pt);
    \draw (80pt,40pt) rectangle (150pt,66pt);
    \draw (180pt,-6pt) rectangle (230pt,20pt);
    \draw (180pt,40pt) rectangle (230pt,66pt);
    \end{scope}
    
    \node[anchor=center, gray] at (25pt,-14pt) {\small \textbf{$p_2$, $\,s\!=\!1$}}; 
    \node[anchor=center, gray] at (25pt,74pt) {\small \textbf{$p_1$, $\,s\!=\!1$}};

    \node[anchor=center, gray] at (115pt,-14pt) {\small \textbf{$p_2$, $\,s\!=\!2$}}; 
    \node[anchor=center, gray] at (115pt,74pt) {\small \textbf{$p_1$, $\,s\!=\!2$}};
    
    \node[anchor=center, gray] at (205pt,-14pt) {\small \textbf{$p_2$, $\,s\!=\!3$}}; 
    \node[anchor=center, gray] at (205pt,74pt) {\small \textbf{$p_1$, $\,s\!=\!3$}};

    \begin{scope}[thick, arrows=-stealth]
    \draw (19pt,53pt) -- (35pt,53pt);
    \draw (40pt,53pt) -- (86pt,53pt);
    \draw (91pt,53pt) -- (106pt,53pt);
    \draw (124pt,53pt) -- (135pt,53pt);
    \draw (140pt,53pt) -- (188pt,11pt);
    \draw (40pt,53pt) -- (186pt,8pt);
    \draw (191pt,7pt) -- (211pt,7pt);
    \end{scope}

    \draw[black, fill=white] (39pt,53pt) circle (1.3ex);
    \node[anchor=center] at (12pt,49pt) {\textbf{$\ldots$}};
    \draw[black, fill=white] (91pt,53pt) circle (1.3ex);
    \draw[black, fill=white] (140pt,53pt) circle (1.3ex);
    \node[anchor=center] at (115pt,49pt) {\textbf{$\ldots$}};
    \draw[black, fill=white] (191pt,7pt) circle (1.3ex);
    \node[anchor=center] at (218pt,4pt) {\textbf{$\ldots$}};

    \node[anchor=center] at (39.4pt,52.5pt) {$v_{_{\!}1}$};
    \node[anchor=center] at (140.8pt,52.5pt) {$v_{_{\!}2}$};
    \node[anchor=center] at (191.4pt,6.5pt) {$v_{_{\!}3}$};

\end{tikzpicture}
    \caption{Example for interpreting DAG scheduling in the maxBSP model.}
    \label{fig:maxBSP}
\end{figure}
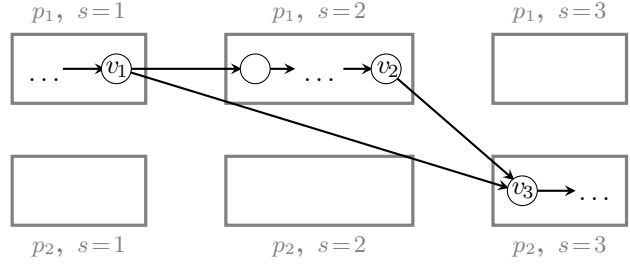

\subsection{The $\alpha-\beta$ model and subset-BSP} \label{app:subs:ABsubset}

The third cell in the last column of Table~\ref{tab:tax2} can be interpreted in several different ways for DAG scheduling; below we discuss two of these. The main difference between these two variants is whether they require \textit{subset synchronization} between the nodes that are currently communicating, similarly to how BSP requires it between all processors. That is, one natural interpretation of this setting is what we call subset-BSP (also studied before~\cite{subsetBSP1, subsetBSP2}), where the schedule is still divided into computation and communication phases, just like BSP, but the different processors can decide to do these communication phases at separate points in time: each such communication phase can happen only between a given subset of processors, while the rest of the processors keep computing in the meantime. The other natural interpretation is closer to the spirit of the $\alpha-\beta$ model (with $\alpha=0$), where no synchronization is required at all; in other words, this is essentially identical to the SPD model, with the minor difference that while communicating, a processor is unable to compute.

To demonstrate the difference between these two models, assume for instance that $p_1$ computes a value $v_1$ in time step $t_0$, and wants to send this to $p_2$ in the time interval $(t_0, t_0+g]$, while $p_3$ computes another value $v_2$ in time step $t_0+1$, and wants to send this to $p_1$ in the time interval $(t_0+1,t_0+g+1]$. This does not violate any condition in the $\alpha-\beta$ model, but it is not allowed in subset-BSP, since $p_3$ is not yet finished with the computation of $v_2$ at $t_0$, so it cannot begin a communication phase together with $p_1$ at this point.

We first begin with the definition of the $\alpha-\beta$ model, since this is surprisingly simple. We can simply take the definition of the SPD model, and extend it with a simple further condition on validity: if there is a communication step $(u, p_1, p_2, t_0) \in \Gamma$, then we cannot have any $v \in V$ such that $\pi(v) \in \{ p_1, p_2\}$ and $t_0<\tau(v) \leq t_0+g$. We once again point out that since the $\alpha-\beta$ model is mostly used as a cost function and not a complete model in the related work, this is just one possible interpretation of it; some other works assume e.g.\ that a processor is not even allowed to send and receive at the same time, which would make this setting significantly different from all other models in our taxonomy. 

On the other hand, formally defining the subset-BSP scheduling model much more technical. For convenience, we modify the definition of supersteps for this, and assume that each superstep on a processor only contains computations. The communication phases are defined separately as a set $\{ \zeta_1, \zeta_2, ... \}$, where each $\zeta_i$ is a $4$-tuple $(\Gamma, P_0, t_s, t_f)$, with $\Gamma$ being a set of communication steps as before, $P_0 \subseteq [P]$ is a subset of processors of size $|P_0| \geq 2$, and $t_s, t_f \in \mathbb{N}$ are integers denoting the starting and finishing time of the communication phase, respectively. The set $\Gamma$ now only consists of $3$-tuples $(v, p_1, p_2) \in V \times [P] \times [P]$, and these must satisfy $p_1, p_2 \in P_0$. The communication cost of $\zeta_i$ is defined as the cost of the communication phase before: $C_{comm}\,\!^{(\zeta_i,p)}$ is the maximum of $C_{sent}\,\!^{(\zeta_i,p)}$ and $C_{rec}\,\!^{(\zeta_i,p)}$, and $C_{comm}\,\!^{(\zeta_i)}$ is the maximum of these values, with $p$ restricted to the processors in $P_0$. The starting and finishing times must satisfy $t_s+C_{comm}\,\!^{(\zeta)}=t_f$.

We also need to establish the starting and ending times of the computational phases (supersteps) with functions $t'_s, t'_f: [P] \times [S] \rightarrow \mathbb{N}$; these must satisfy $t'_f(p, s) = t'_s(p, s) + C_{work}\,\!^{(s,p)}$. Furthermore, the supersteps need to be ordered correctly ($t'_s(p, s) \geq t'_f(p, s')$ for $s'<s$), and disjoint from the communication phases containing $p$ (if $\zeta_i$ has $p \in P_0$, then for all $s$ either $t'_f(p,s)\leq t_s$ or $t'_s(p,s)\geq t_f$ for the $t_s$, $t_f$ in $\zeta_i$). Note that if the desired number of computation phases is different for different processors, then we can simply add supersteps of length $0$ for some processors at the end of their schedule.

The validity conditions in this model are as follows. For condition (i), in case of $\pi(u) \neq \pi(v)$, we must now require that there exists a $\zeta_i=(\Gamma, P_0, t_s, t_f)$ such that $(u, \pi(u), \pi(v)) \in \Gamma$ and $t_f \leq t'_s(\pi(v),\tau(v))$. For condition (ii), we need for each $\zeta_i=(\Gamma, P_0, t_s, t_f)$ and $(v, p_1, p_2) \in \Gamma$ to have $p_1=\pi(v)$ and $t'_f(p_1, \tau(v)) \leq t_s$. Finally, the cost of a schedule can be defined as
\[  \max\left( \max_{p \in [P]} t'_f(p, S)\, , \: \max_{\zeta_i} t_f \right) \, . \]

Subset-BSP can also easily be extended with a latency cost at each subset synchronization. On the other hand, there is no straightforward way to do the same in the $\alpha-\beta$ model.

\subsection{Commdelay model variants}

While the different variants of the commdelay model (in the middle column of Table \ref{tab:tax2}) are further away from the main focus of the paper, we briefly show how these can be defined in a straightforward way. These $4$ commdelay variants are indeed different from each other, and hence the properties on the vertical axis of the table (simultaneous computation and communication, and barrier synchronization) also have a significant effect in this simplified model where data volume is not measured.

To define the different commdelay variants, it is easier to specifically include the communication steps in the schedule; that is, just like in SPD, the schedule again contains a $\Gamma \subset V \times [P] \times [P] \times \mathbb{Z}^+$, where $(v, p_1, p_2, t_0) \in \Gamma$ indicates that $v$ is sent from $p_1$ to $p_2$ in the interval $(t_0, t_0+g]$. The standard commdelay model can of course also be redefined this way; we must require that any $(v, p_1, p_2, t_0) \in \Gamma$ must have $t(v) \leq t_0$, and for edges $(u,v) \in E$ with $\pi(u) \neq \pi(v)$, we need a $(u, \pi(u), \pi(v), t_0) \in \Gamma$ with $t_0+g < t(v)$. In contrast to SPD, there is no condition requiring that the communication time intervals on a processor cannot intersect. One can observe that this definition is indeed equivalent to the original commdelay model if we construct $\Gamma$ by adding the $4$-tuple $(u, \pi(u), \pi(v), t(u)) \in \Gamma$ for each $(u,v) \in E$ with $\pi(u) \neq \pi(v)$.

To define a commdelay model with sync points, we essentially assume that the start of each communication step requires a barrier synchronization between all processors; this essentially means that no processors can be in the middle of a computation or a communication step at this point. If each computation has a unit work weight, then every integer time $t_0$ is naturally a point where no computation is happening currently; as such, for commdelay with sync points, we only need to ensure that no processor is in the middle of a communication step when another communication begins. Formally, this means adding the following condition to our schedule: for any two $(v_1, p_{1,1}, p_{2,1}, t_1), (v_2, p_{1,2}, p_{2,2}, t_2) \in \Gamma$, we must have either $t_1 \leq t_2$ or $t_1 \geq t_2+g$.

The commdelay model with phases, which has already been studied in \cite{CDphases}, is essentially identical to BSP, but every communication phase lasts for exactly $g$ time units; as such, its definition can indeed be obtained from BSP by setting $C_{comm}\,\!^{(s)}=g$ for any superstep. One can also add a latency $L$ to this model, but this essentially only creates a new delay parameter $g'=g+L$. Alternatively, the model can be defined by starting from commdelay with sync points, and further requiring that if $(v, p_1, p_2, t_0) \in \Gamma$, then for all $v \in V$ we have $t(v) \leq t_0$ or $t(v) > t_0+g$.

In case of the third column of the table we once again have two possible interpretations, similarly to those discussed in Section~\ref{app:subs:ABsubset}. If no synchronization is required at all, then we can get a commdelay model with ``timeouts'', i.e.\ where a processor has to stop computing while it sends/receives. This is essentially again the $\alpha-\beta$ model, but now with $\beta=0$; we can define this by starting from commdelay, and adding the condition that if $(u, p_1, p_2, t_0) \in \Gamma$, then we cannot have a $v \in V$ with $\pi(v) \in \{ p_1, p_2\}$ and $t_0<\tau(v) \leq t_0+g$. On the other hand, if we require subset synchronization between the communicating nodes, then we can also define a subset-commdelay model analogously to subset-BSP by changing the communication costs to $C_{comm}\,\!^{(\zeta_i)}=g$ in subset-BSP.

\subsection{BSP model extensions}

Real-world computations often consist of nodes with significantly different computational or communication cost. As such, one natural extension is to consider computation (i.e.\ work) weights for nodes. That is, we assume that we also have a \emph{work weight} function $w_{work}: V \rightarrow \mathbb{Z}^+$ that defines the amount of time required to execute each node; similarly to other works, we assume for simplicity that these weights are integers. Most scheduling models are easy to extend with these work weights. E.g.\ in case of our BSP, the work cost of a superstep is then understood as the sum of node weights in the superstep, i.e.\ $ C_{work}\,\!^{(p, s)} = \sum_{v \in H^{(s,p)}} \, w_{work}(v)$.

Besides heterogeneous work weights, the size of the output data of nodes can also differ significantly, and hence it can take a different amount of time to communicate specific values between the processors. As such, another extension is to assume \emph{communication weights} $w_{comm}: V \rightarrow \mathbb{Z}^+$ which describe the data size (e.g.\ number of bytes) to be communicated in case of each node (alternatively, some works consider this to be a function $w_{comm}: E \rightarrow \mathbb{Z}^+$; we briefly discuss this in Section \ref{app:subs:edgebased}). The communication cost can again be adjusted accordingly, by multiplying each term in $C_{sent}\,\!^{(s,p)}$ and $C_{rec}\,\!^{(s,p)}$ by the respective $w_{comm}(v)$ weight.

Besides work/communication weights, another extension of the base model is to allow \emph{duplications} (also called replications). That is, our initial definitions assign every computational task to a single processor and time step; in contrast to this, the extension with duplication also allows us to execute the same node multiple times, on different processors. Formally, a schedule in this case becomes a function $\tau: V \times [P] \rightarrow \mathbb{Z}^+ \cup \{\emptyset\}$ such that for all $v \in V$ there exists a $p \in [P]$ with $\tau(v,p) \neq \emptyset$. Naturally, these duplication steps come with more total time spent on computation, but on the other hand, they can reduce the required amount of communication between the processors, hence decreasing the scheduling cost altogether (see Proposition \ref{lem:recomp} for examples).

\subsection{Edge-based communication} \label{app:subs:edgebased}

Regarding the BSP model, we also point out that (as per our definitions in Section \ref{sec:BSP}) we assume a \textit{node-based} communication model throughout the paper. That is, each communication step corresponds to sending the output of some node $u$ to another processor, which is then saved on this other processor indefinitely. In particular, if $u$ is computed on $p_1$, and it has many outgoing edges $(u, v_1), ..., (u, v_k) \in E$ to some nodes $v_1, ..., v_k$ which are all assigned to $p_2$, then this only incurs a single communication step in our model. This is in line with several recent works that also assume such a node-based communication model~\cite{hyperDAG, DAH}.

In contrast to this, some other works (e.g.\ ~\cite{SPD}) assume an edge-based model, where the situation above incurs $k$ distinct communication steps, one for each of the edges $(u, v_i)$. This edge-based model might be motivated by a setting where the memory of processors is limited; however, in general, it can overestimate the real communication cost by a factor $k$.

Note that if we have communication weights, then the natural choice for $w_{comm}$ is a function $V \rightarrow \mathbb{Z}^+$ in this node-based model. On the other hand, other works assuming an edge-based model sometimes use a function $w_{comm} : E \rightarrow \mathbb{Z}^+$ instead.

For completeness, we also provide a definition of the edge-based alternative, i.e.\ the BSP scheduling problem with edge-based communication. In this edge-based model, the communication steps $\Gamma$ are defined as a set of pairs $((u,v),s) \in E \times [S]$ (instead of $4$-tuples), indicating that the data corresponding to edge $(u,v)$ is transferred from $\pi(u)$ to $\pi(v)$ in the communication phase of superstep $s$. In case of $\pi(u) \neq \pi(v)$, validity condition (i) of the BSP schedule will now state that we must have $((u,v), s) \in \Gamma$ for some $s < \tau(v)$. Validity condition (ii) is also changed accordingly: if $((u,v), s) \in \Gamma$, then we need to have $\tau(u) \leq s$. Finally, the definition of communication costs is modified as follows:
\begin{gather*}
C_{sent}\,\!^{(s,p)}=|\{ ((u,v),s_1) \in \Gamma \, | \, \tau(u)=p, s_1=s \}| \, , \\
C_{rec}\,\!^{(s,p)}=|\{ ((u,v),s_1) \in \Gamma \, | \, \tau(v)=p, s_1=s \}| \, , \\
\end{gather*}

\section{Proofs for Section \ref{sec:tax}} \label{app:taxproofs}

This section provides a detailed discussion of our claims regarding the relations between the different models in the taxonomy.

\subsection{Communication time intervals and $h$-relations}

Regarding the last column of Table \ref{tab:tax1}, note that BSP measures data volume through the simpler notion of $h$-relations, while SPD assigns communications to concrete time intervals. It is a natural question whether these two approaches indeed describe the same cost function (for $L _{\!} = _{\!} 0$ in BSP). That is, to justify the use of $h$-relations in BSP, we need to show that if each processor $p$ sends $C_{sent}\,\!^{(s,p)}$ and receives $C_{rec}\,\!^{(s,p)}$ values, then the communications can indeed be scheduled into $C_{comm}\,\!^{(s)}$ time slots according to the rules of SPD, i.e.\ such that each processor only sends and receives a single value in a step. More specifically, assume that for each pair of processors $p \neq p'$, we have $\Lambda(p, p')$ values (all available at time $t$) that we need to send from $p$ to $p'$. Naturally, $\Lambda(p, p)=0$ for all $p \in [P]$. We show that, at least in the unweighted case, we can indeed communicate all these values by the time $t+g \cdot C_{comm}\,\!^{(s)}$ according to the rules of the SPD model, where
\[ C_{comm}\,\!^{(s)} = \max_{p, p' \in [P]} \Lambda(p, p') \]
as in the definition of BSP. This is easy with standard graph-theoretic tools, by decomposing a bipartite multigraph into perfect matchings.

\begin{lemma} \label{lem:matching}
If any processor sends and receives at most $C$ values in a communication phase, then the phase can be scheduled into $C$ time slots (according to the rules of SPD).
\end{lemma}

\renewcommand*{\proofname}{Proof.}
\begin{proof}
Consider a bipartite multigraph on $P+P$ nodes $u_1, ..., u_P$ and $v_1, ..., v_P$, where for each $p, p' \in [P]$, we draw $\Lambda(p, p')$ distinct edges from $u_p$ to $v_{p'}$. Until there is a node which has degree less than $C_{comm}\,\!^{(s)}$, we can select an arbitrary such node in both $\{ u_1, ..., u_P \}$ and $\{ v_1, ..., v_P \}$, and add an edge between them (we will call this an ``artificial'' edge). After several such steps, the graph becomes $C_{comm}\,\!^{(s)}$-regular.

It is known that such a regular bipartite graph can be decomposed into $C_{comm}\,\!^{(s)}$ perfect matchings. In particular, the condition of Hall's theorem is always satisfied for regular graphs, so there exists a perfect matching; if we remove these edges, the remaining graph is again regular. The fact that our graph is a multigraph does not affect this method at all, since a multiplicity of a given edge contributes the same number of edges to a subset $U \subseteq \{u_1, ..., u_P \}$ and its neighborhood when applying Hall's theorem.

These perfect matchings naturally divide the communications into $C_{comm}\,\!^{(s)}$ consecutive rounds where every processor only sends and receives a single value in each round. As such, we can transfer the values corresponding to the first matching in the time interval $(0, g]$, the values in the second matching in the interval $(g, 2\cdot g]$, and so on. When an edge from $u_{p_1}$ to $v_{p_2}$ is artificial, then $p_1$ simply sends no value to $p_2$ in that round.
\end{proof}

The same claim does not hold however in case of communication weights. Consider the case when $\Lambda(p_1, p_2)=\Lambda(p_2, p_3)=\Lambda(p_3, p_1)=3$, each of these induced by a single value of communication weight $w_{comm}=3$, and we furthermore have a processor $p_4$ that needs to send a single value of communication weight $1$ to all of $p_1$, $p_2$ and $p_3$. The cost of the $h$-relation here is $C_{comm}\,\!^{(s)}=4$, since $p_1$, $p_2$ and $p_3$ receive $4$ units of data in total. However, we are unable to schedule this into the slots $(0, g],(g, 2g],(2g, 3g],(3g, 4g]$ in the SPD model: the values sent between $p_1$, $p_2$ and $p_3$ each occupy a continuous interval of length $3g$, so each of $p_1$, $p_2$ and $p_3$ can only receive values from $p_4$ either in $(0, g]$ or in $(3g, 4g]$. However, $p_4$ has to send a value to all three processors, which is not possible in these two intervals.

In case of communication weights, the same equivalence would only hold if the communication steps are preemptive, i.e. we can split $w_{comm}(v)$ into smaller integer-sized data units; in this case, we could apply the same technique as in the unweighted case. However, the study of this preemptive case is beyond the scope of this paper.

\subsection{Lemma~\ref{lem:horizOpt}}

The upper bound of factor $P$ in Lemma~\ref{lem:horizOpt} follows directly from Proposition~\ref{prop:approx}. As for the upper bound of a factor $(1+g)$, consider a schedule $(\pi, t)$ in the classical model, and assume that after each step, we keep all the processors idle for $g$ time steps and communicate every value computed in the given time step to all other processors. In other words, the consider the timing $t'(v)=(1+g) \cdot t$ for all $v \in V$; this is a valid schedule of the same DAG in the commdelay model, and its makespan is at most a factor $(1+g)$ larger.

For the matching lower bounds, consider a layered DAG of length $\ell$ and width $k$, i.e.\ a DAG consisting of $n=k \cdot \ell$ nodes indexed $v_{i,j}$ for $i \in [\ell], j \in [k]$, such that every node in the $i$-th layer has an edge to every node in the $(i+1)$-th layer, i.e.\ $(v_{i, j_1}, v_{(i+1),j_2}) \in E$ for all $i \in [\ell-1]$ and $j_1, j_2 \in [k]$. With $P=k$ processors, we have $\textsc{OPT}_{class}=\ell$ in this DAG, since computation can be parallelized completely. However, in the commdelay model, we have to wait $g$ time units between each consecutive pair of layers unless both layers are entirely computed on the same processor. As such, if $P \leq (1+g)$, then the optimal strategy is to execute the entire DAG on a single processor, resulting in a makespan of $P \cdot \ell$. If $P > (1+g)$, then it is indeed better to wait for $g$ time units, which gives an execution time of $(\ell-1) \cdot (1+g) +1$ (since we do not have to wait after the last layer). In the first case, we get a ratio of $P$ exactly, whereas in the second case, we get a ratio of
\[ \frac{(\ell-1)(1+g) +1}{\ell}=(1+g) - \frac{g}{\ell} \, \]
which is larger than $(1+g)-\varepsilon$ for any $\varepsilon>0$ if $\ell$ is large enough.

For the lower bound for classical scheduling and BSP, we can even achieve the factor $P$ lower bound when $(1+g)<P$. For this, let us add even more edges to the construction above, such that $(v_{i_1, j_1}, v_{i_2,j_2}) \in E$ for all $i_1, i_2 \in [\ell]$, $i_1<i_2$ and $j_1, j_2 \in [k]$. Here it is intuitively even more clear that the optimum solution uses a single processor only. In particular, if two nodes anywhere in the layers $\frac{\ell}{2}+1$, ..., $\ell$ are assigned to different processors, then one of these processors must receive at least $\frac{1}{4} \cdot \ell \cdot P$ values altogether from the first half of the DAG, resulting in a communication cost of $\frac{1}{4} \cdot \ell \cdot P \cdot g$; this is already larger than $\ell \cdot P$ if $g \geq 4$. As such, the second half of the DAG must be assigned to a single processor $p_1$. Assume that there is a solution with cost below $\ell \cdot P$ that also uses another processor besides $p_1$; we will even further reduce the communication cost of this solution by disregarding every communication cost in our analysis except the values received by $p_1$. We can iterate through the nodes that are assigned to a different processor than $p_1$ (from layer $\frac{\ell}{2}$ to layer $1$), and reassign them one-by-one to $p_1$; each such step may increase computation cost by $1$, but decreases communication cost by $g$, so it decreases the cost in total. After at most $\frac{\ell}{2} \cdot P$ steps, everything is scheduled in $p_1$; however, this solution has cost $\ell \cdot P$, which is a contradiction.

Finally, for a lower bound for the case of commdelay and BSP, consider again a layered DAG of length $\ell$ and width $P$, but now with layer $i$ only having edges to the nodes in layers $\{ i+g+1, ..., \ell \}$, i.e.\ we have $(v_{i_1, j_1}, v_{i_2,j_2}) \in E$ for all $i_1, i_2 \in [\ell]$, $i_1+g<i_2$ and $j_1, j_2 \in [k]$. In the commdelay model, since all values are available on all processors after $g$ time units, we can parallelize each layer of this DAG perfectly and achieve a total cost of $\ell$. In case of BSP, however, we can apply the same proof as above to show that there is again no better schedule than computing all nodes on a single processor $p_1$, at a cost of $P \cdot \ell$. This again amounts to a factor $P$ difference.

We note that the same lower bound construction techniques may also work if we replace the BSP model by SPD; however, in this case, it becomes more technical to prove that an SPD schedule of lower cost does not exist.

\subsection{Theorem~\ref{th:optcosts}}

We now present the proof of Theorem~\ref{th:optcosts}. We note that the same upper bound proof and lower bound construction techniques can be used to show the same relations between the different commdelay model variants (middle column of Table~\ref{tab:tax2}).

\subsubsection{Upper bound}

We begin by showing the upper bound part of Theorem~\ref{th:optcosts}. Consider an SPD schedule of makespan $t_{max}$ on a given DAG; we show how to turns this into a BSP schedule of cost at most $2 \cdot t_{max}$, hence proving $\textsc{OPT}_{BSP} \, \leq \, 2 \cdot \textsc{OPT}_{SPD}$. The remaining upper bounds follow from this, since all the optimum costs are at most as large as $\textsc{OPT}_{BSP}$ and at least as large as $\textsc{OPT}_{SPD}$.

Assume we have a DAG and an SPD schedule of this DAG; then let us divide this to time intervals of length $g$, i.e.\ $(0,g]$, $(g,2g]$, ..., up to $(t_{max}-g,t_{max}]$ (the last interval may be shorter if $t_{max}$ is not divisible by $g$). We turn each such interval $(t_0, t_0+g]$ into a separate superstep $s$ in our BSP schedule. In this superstep $s$, the computational phase on processor $p$ executes the same nodes $v$ that have $\tau(v)=p$ and $t(v) \in (t_0, t_0+g]$ in our SPD schedule; since there are at most $g$ such nodes, we will have $C_{work}\,\!^{(s)} \leq g$. The communication phase of superstep $s$, on the other hand, will send all the communication steps $(v, p_1, p_2, t') \in \Gamma$ that have $t_0 \leq t' < t_0+g$; since the sending and receiving intervals in SPD have length $g$ and are disjoint, this means that each processor can only send and only receive a single value in superstep $s$. As such, we have $C_{comm}\,\!^{(s)} \leq 1$, and hence $C\,\!^{(s)} \leq 2 \cdot g$. This means that the total cost of the BSP schedule is indeed at most
\[ \frac{t_{max}}{g} \cdot 2 \cdot g = 2 \cdot t_{max} \, . \]

We note that if the last interval has length smaller than $g$, then it results in a computational phase with cost equal to the length of the interval, and an empty communication phase, since we cannot begin transmitting any more values at this point. As such, the corresponding SPD and BSP costs are equal for this interval.

\subsubsection{Lower bound of $(2-\varepsilon)$}

Let us now analyze the lower bound for the model pairs
$(\textsc{OPT}_{SPD\,}\!\leq\!\textsc{OPT}_{\beta})\,$, $\;(\textsc{OPT}_{SPD\,}\!\leq\!\textsc{OPT}_{BSP})$, $(\textsc{OPT}_{mBSP\,}\!\leq\!\textsc{OPT}_{BSP})\;$ and $\;(\textsc{OPT}_{mBSP\,}\!\leq\!\textsc{OPT}_{\beta})$. For these cases, it suffices to develop a DAG construction where computation and communication can be completely parallelized in the models where this is permitted.

Assume $P>2$, and consider another parameter $k$, and let $\ell \:= k \cdot g$. Let us construct a layered DAG consisting of $n=P \cdot (\ell + 1)$ nodes indexed $v_{i,j}$ for $i \in \{0\}\!\cup\![\ell], j \in [P]$, having the edges $(v_{(i-1),j},v_{i,j}) \in E$ for all $i \in [\ell]$, $j \in [P]$, as well as the edges
 \[ (v_{(i-1) \cdot g),j} \,, \: v_{i \cdot g+1,\,((j+1) \text{ mod } P)}) \in E \]
for all $i \in [k]$, $j \in [P]$, where we now understand the modulo operation to return a remainder from the set $[P]$ for simplicity.

The DAG above can be scheduled with a cost of $(\ell+1)$ in any of the models that allow simultaneous computation and communication, i.e.\ in the SPD and maxBSP models. More specifically, for all $p \in [P]$, we can assign all the nodes $v_{i, p}$ to processor $p$, and in the SPD model, assign each node $v_{i,j}$ to time step $t(v_{i,j})=i+1$, whereas in the maxBSP model, assign nodes $v_{0,j}$ to superstep $1$, and all other nodes $v_{i,j}$ with $i \in \{ (s-2) \cdot g + 1, ..., (s-1) \cdot g \}$ to superstep $s$ for $s \in \{2, ..., k+1 \}$. This ensures that each processor $p$ can keep computing the given nodes $v_{i, p}$ in order without any interruption. In SPD, the value of $(v_{(i-1) \cdot g,j})$ is transferred from processor $j$ to processor $((j+1) \text{ mod } P)$ in the interval $((i-1) \cdot g+1, i \cdot g+1]$ for all $i \in [k]$, $j \in [P]$, and in maxBSP, the same value is transferred in the communication phase of superstep $s=i+1$.

On the other hand, consider the models that do not allow simultaneous computation and communication: BSP, $(\alpha-\beta)$ with $\alpha=0$, or subset-BSP. In these models, the optimum solution is to once again assign all the nodes $v_{i, p}$ to processor $p$, but now before every node $v_{i \cdot g, p}$, we need to allocate $g$ time units to communicate the values between the processors. As such the optimum cost in this case is $(1+2\cdot g \cdot k)$, hence the ratio between the two optima is
\[  \frac{1 + 2 \cdot g \cdot k}{1 + g \cdot k} \, , \]
which is indeed arbitrary close to $2$ for $k$ large enough.

\subsubsection{Lower bound of $(\frac{3}{2}-\varepsilon)$}

The lower bound construction for the model pairs
$(\textsc{OPT}_{SPD\,}\!\leq\!\textsc{OPT}_{mBSP})\,$, $(\textsc{OPT}_{\beta\,}\!\leq\!\textsc{OPT}_{BSP})\;$ and $\;(\textsc{OPT}_{\beta\,}\!\leq\!\textsc{OPT}_{mBSP})$ is slightly more technical. We first outline the main idea. Our DAG will consist of $g$ distinct layered subDAGs; we will number these by $k \in [g]$. Each of the components can individually be scheduled with a cost of $2 \cdot g + 1$; however, for this, the $k$-th component needs to execute computations in times steps $\{ 1, ..., k \}$, then a communication for $g$ time units, and then again computations for the time steps $\{ g + k + 1 , ...,  2 \cdot g + 1 \}$. That is, the first subDAG requires communications in the time interval $(1, g+1]$, the second one in the time interval $(2,g+2]$, and so on, until $(g,2\cdot g]$. However, this is not possible with barrier synchronization; if synchronization is required, the best we can do is as follows: (i) execute the computation steps before the required communication on all processors, in time $g$, then (ii) execute all the communications together for another $g$ time units, and finally (iii) compute all the post-communication nodes in another $g$ time units. This results in a cost of $3 \cdot g$, and hence a ratio of
\[ \frac{3 \cdot g}{ 2 \cdot g + 1} \, , \]
which is arbitrarily close to $\frac{3}{2}$ in an example where the parameter $g$ is large enough.

Let us introduce a parameter $k_0$ which will be the width of each component. Each component consists of $k_0 \cdot (g+1)$ nodes indexed $v_{i,j}$ for $i \in \{0\}\!\cup\![g], j \in [k_0]$. We add the edges $(v_{(i-1),j},v_{i,j}) \in E$ for all $i \in [g]$, $j \in [k_0]$, as well as the edges
 \[ (v_{k,j} \,, \: v_{k+g+1,\,((j+1) \text{ mod } k_0)}) \in E \]
in the $k$-th component for all $j \in [k_0]$, where the modulo operation once again returns a remainder from the set $[k_0]$. Finally, we select $P=k \cdot k_0$, i.e.\ there is a sufficient number of processors to assign a separate processor $p$ to the values $v_{i,j}$ in every component for every fixed $j \in [k_0]$.

This DAG can be scheduled at a cost of $2 \cdot g + 1$ in the SPD, the $\alpha-\beta$ with $\alpha=0$ and subset-BSP models, as described above: in the $k$-th component, we can execute the computations $v_{1,j}$, ..., $v_{k,j}$ in the time steps $\{ 1, ..., k \}$, then communicate the necessary values for $g$ time units, and then again compute $v_{k+1,j}$, ..., $v_{g+1,j}$ in time steps $\{ g+k+1, ..., 2 \cdot g + 1 \}$.

On the other hand, in the models with barrier synchronization, we need to wait for the computations $v_{1,j}$, ..., $v_{k,j}$ to finish in \textit{all} components before starting a new phase/superstep (in BSP/maxBSP, respectively), then execute all communications simultaneously, and then execute the remaining computations in a following phase/superstep. This results in a total cost of $3 \cdot g$. Note that the ability to execute computations and communications simultaneously in maxBSP does not offer any benefit in this DAG, since each of the nodes $v_{k+1,j}$ require a communication step. Alternatively, in maxBSP, it can also be optimal to have two separate communication phases: for some $k \in [g]$, we organize a separate communication phase for the components $\{1, ..., k \}$ in the time interval $(k;k+g]$, and then another one for the components $\{k+1, ..., g \}$ in the interval $(k+g;k+2 \cdot g]$. However, this again requires an initial superstep with $k$ computations in component $k$ and a final superstep with $(g-k)$ computations in component $(k+1)$, and hence it also has a cost of $3 \cdot g$. Finally, another possible solution would be to reduce the number of processors used (per component) in order to avoid communication. However, if we select $k_0 \geq 3$, then computing each component on a single processor already becomes suboptimal, and if we use at least $2$ processors, then at least one communication step is already required within each component.

\subsection{Propositions~\ref{lem:PRAMweights} and \ref{lem:recomp}}

We now outline the proofs of Propositions~\ref{lem:PRAMweights} and \ref{lem:recomp} from Section~\ref{sec:tax}.

When referring to the classical scheduling model with work weights, we assume that a schedule must fulfill slightly different properties: we need to have $\nexists \, u, v \in V$ with $\pi(u)=\pi(v)$ such that the intervals $(t(u), t(u)+w_{work}(u)]$ and $(t(v), t(v)+w_{work}(v)]$ intersect, and $\forall (u,v) \in E$ we need to have $t(u)+w_{work}(u) \leq t(v)$. The makespan in this case is understood as $\max_{v \in V\,} t(v)+w_{work}(v)$. If we add barrier synchronization to this, we require that for each $(u,v) \in E$ with $\pi(u) \neq \pi(v)$ there exists a possible synchronization point $t_0$ within the interval $[t(u)+w_{work}(u), t(v)]$ such that $\nexists \, v_0 \in V$ with $t_0 \in (t(v_0), t(v_0)+w_{work}(v_0)]$, i.e.\ no processor is computing at time $t_0$.

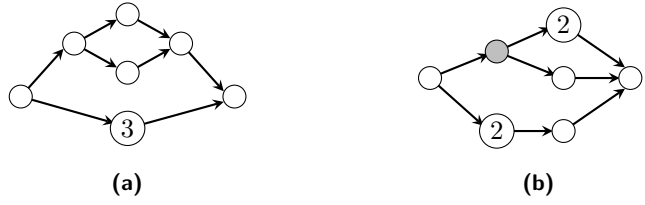
\begin{figure}
    \centering
    \begin{subfigure}[b]{0.23\textwidth}
        \centering
        \begin{tikzpicture}

        \begin{scope}[thick, arrows=-stealth]
	\draw (0pt,17pt) -- (16pt,35pt);
        \draw (20pt,37pt) -- (36pt,46pt);
        \draw (20pt,37pt) -- (36pt,28pt);
        \draw (40pt,26pt) -- (56pt,35pt);
        \draw (40pt,48pt) -- (56pt,39pt);
        \draw (60pt,37pt) -- (76pt,19pt);
        \draw (0pt,17pt) -- (35pt,7pt);
        \draw (40pt,5pt) -- (76pt,15pt);
        \end{scope}
	
	%---
	
	\draw[black, fill=white] (0pt,17pt) circle (1.0ex);
	\draw[black, fill=white] (20pt,37pt) circle (1.0ex);
	\draw[black, fill=white] (40pt,26pt) circle (1.0ex);
	\draw[black, fill=white] (40pt,48pt) circle (1.0ex);
	\draw[black, fill=white] (60pt,37pt) circle (1.0ex);
	\draw[black, fill=white] (80pt,17pt) circle (1.0ex);
	\draw[black, fill=white] (40pt,5pt) circle (1.5ex);
        \node[anchor=center] at (40pt,5pt) {$3$};

\end{tikzpicture}
        \vspace{-9pt}
        \captionsetup{justification=centering}
        \caption{}
        \label{fig:classWW}
    \end{subfigure}
    \hspace{0.14\textwidth}
    \begin{subfigure}[b]{0.23\textwidth}
        \centering
        \begin{tikzpicture}

        \begin{scope}[thick, arrows=-stealth]
	\draw (0pt,20pt) -- (20pt,2pt);
        \draw (0pt,20pt) -- (22pt,29pt);
        \draw (25pt,0pt) -- (46pt,0pt);
        \draw (25pt,30pt) -- (47pt,21pt);
        \draw (25pt,30pt) -- (45pt,39pt);
        \draw (50pt,0pt) -- (73pt,16pt);
        \draw (50pt,20pt) -- (71pt,20pt);
        \draw (50pt,40pt) -- (73pt,24pt);
        \end{scope}
	
	%---
	
	\draw[black, fill=white] (0pt,20pt) circle (1.0ex);
	\draw[black, fill=white] (25pt,0pt) circle (1.5ex);
	\draw[black, fill=lightgray] (25pt,30pt) circle (1.0ex);
	\draw[black, fill=white] (50pt,0pt) circle (1.0ex);
	\draw[black, fill=white] (50pt,20pt) circle (1.0ex);
        \draw[black, fill=white] (50pt,40pt) circle (1.5ex);
	\draw[black, fill=white] (75pt,20pt) circle (1.0ex);
        \node[anchor=center] at (25pt,0pt) {$2$};
        \node[anchor=center] at (50pt,40pt) {$2$};

\end{tikzpicture}
        \vspace{3pt}
        \captionsetup{justification=centering}
        \caption{}
        \label{fig:recomp}
    \end{subfigure}
    \vspace{-6pt}
    \caption{Example DAGs with work weights in the classical scheduling model with barrier synchronization, for $P_{\!}=_{\!}3$. The weights are only shown for nodes of non-unit weight.}
\end{figure}

\subsubsection{Proof of Proposition~\ref{lem:PRAMweights}}

Consider the weighted DAG in Figure~\ref{fig:classWW} with $P\!=\!3$ processors. Let us refer to the bottom node with weight $3$ as $u$. Without barrier synchronization, if we assign the topmost node to $p_1$, node $u$ to $p_2$ and all the other nodes to $p_3$, the DAG admits a straightforward schedule with a makespan of $5$. On the other hand, when barrier synchronization is required, the same DAG will have an optimal makespan of $6$ in classical scheduling. If we set $t(u)=1$, then no communication is possible in the interval $(1,4]$, and hence it is not possible use the other two processors to parallelize the upper part of the DAG between two processors. As such, we will still have $2$ uncomputed nodes at time step $4$, increasing the best makespan to $6$. Alternatively, setting $t(u)=2$ also results in a makespan of $6$.

One can also generalize this into a construction where the optimum cost essentially becomes twice as large when barrier synchronization is introduced. Let $k$ be a constant parameter, and consider the construction from the proof of Lemma~\ref{lem:horizOpt} with width and length $k$, i.e. a DAG sorted into $k$ layers of width $k$ such that every node in layer $i \in [k-1]$ has an edge to every node in layer $(i+1)$. Besides this, let us add an isolated node $u$ with $w_{work}(u)=k$ (we can also add a common source and sink node to the DAG if we prefer it to be connected). Finally, let us select $P=k+1$.

Without barrier synchronization, this DAG can be scheduled with a makespan of $k$. However, if synchronization is required, then whenever node $u$ is computed, the remaining processors are not able to communicate any values within each other, so they must stay idle for $(k-1)$ time steps (or a single processor must compute an entire layer, which essentially has the same effect, resulting in a delay of $(k-1)$). Hence the best makespan with synchronization is $2k-1$. The ratio between the two optima is $\frac{2k-1}{k}$; we can make this higher than $2-\varepsilon$ for any $\varepsilon>0$ with a sufficiently large choice of $k$.

\subsubsection{Proof of Proposition~\ref{lem:recomp}}

For the first point, consider a simple DAG consisting of two directed paths of length $\ell$, and a common source node $v_0$ that has an edge towards the initial node of both paths. For $\ell$ large enough, any of the models with communication assigns the two paths to two different processors $p_1$ and $p_2$, otherwise the cost is much larger. Without latency, this results in a work cost of $\ell+1$ and a communication cost of $1 \cdot g$ in each of the models, so the total cost/makespan is $\ell+g+1$. On the other hand, if we can duplicate $v_0$ on both $p_1$ and $p_2$, then no communication is required, and we easily obtain a cost of $\ell+1$, which is smaller.

In case of classical scheduling without work weights, each integer point in time allows us to communicate the computed values to all processors (even with barrier synchronization). As such, if a node $v$ is computed at time step $t$ at the earliest, then its value can be distributed to all processors by the beginning of time step $(t+1)$, so there is no motivation to compute $v$ on further processors either in time step $t$ or later. Similarly, in case we have work weights but no barrier synchronization, if $t(v)+w_{work}(v)$ is the earliest time when one of the processors finishes computing $v$, then the value of $v$ already becomes available on all other processors by this time $t(v)+w_{work}(v)$, so there is no motivation for any other processor to begin computing $v$. As such, duplication cannot decrease the optimum cost in these cases.

For the case when we have classical scheduling with both work weights and barrier synchronization, consider Figure~\ref{fig:recomp} with $P=3$. Without barrier synchronization, the DAG admits a straightforward schedule of makespan $5$. However, with synchronization, a makespan of $5$ is not possible here: the top and bottom paths are both critical, so a makespan of $5$ is only possible if no communication happens in the intervals $(1,3]$ and $(2,4]$. This implies that the value of the grey node, which becomes available at $t=2$ earliest, cannot be sent to a third processor until $t=4$, and as such, the node on the middle path cannot be computed by a third processor. If we only start computing this node at $t=4$ or on one of the first two processors, then the best makespan is $6$. On the other hand, if duplication is allowed, then we can compute the bottom path on $p_1$, duplicate the grey node on both $p_2$ and $p_3$, and then compute the top and middle paths on $p_2$ and $p_3$, respectively, again achieving a makespan of $5$.

\section{Complexity for specific DAG classes} \label{app:hardness}

\subsection{Proof of Theorem~\ref{th:chains}}

\subsubsection{Chain DAGs} We first show that for chain DAGs, there exists an algorithm that finds the best BSP schedule, and has a running time that is polynomial in $n$.

\renewcommand*{\proofname}{Proof of Theorem~\ref{th:chains} for chain DAGs}
\begin{proof}
The key observation is that the optimal BSP schedule in chain DAGs always consists of either $1$, $2$, ..., or $P$ supersteps; more specifically, the optimum uses at most $(P-1)$ rounds of communication (for clarity, we prove this separately below in Lemma~\ref{lem:maxPss}). We then consider each case separately, i.e.\ the best possible solution assuming that the number of supersteps is $1$, ..., or $P$, and then select the best of these to find the optimum.

Now assume that the number of supersteps is some fixed $S \in [P]$. More specifically, we will assume that the total communication cost is $(S-1)$, and then for convenience, we assume that these $(S-1)$ rounds of communication indeed happen in different supersteps. If one of the supersteps in the optimum actually had a communication cost of more than $1$, then we can split this into several supersteps of communication cost $1$ each, and simply leave the computation phase of the newly created extra supersteps empty.

If the communication cost of a BSP schedule is $S$, then this means that there are at most $P \cdot S$ values communicated altogether. In fact, this implies that we can iterate over all the possible values $v$ and source and target processors $p_1$, $p_2$ that took part in this communication: with $n$ nodes and $P$ processors, there are only $n \cdot P^2+1$ values for each communication step (including the possibility that this communication step is not used), and hence the number of possible configurations for communication is upper bounded by $(n \cdot P^2)^{P \cdot S}$, which is still polynomial in $n$. We will iterate through all these configurations, find the best possible cost for each of them, and simply return the best of these. Note that many of the configurations are actually invalid, e.g.\ if $p_1=p_2$, or $v$ is sent multiple times from $p_1$ to $p_2$, but filtering these out is just a technicality.

Hence consider a given configuration for communication. Note that for each of the chains that have one (or more) of the communicated nodes, this already determines the assignment of each node in the chain to a processor: if a node $v$ is sent from $p_1$ to $p_2$, the only possible reason for this in a chain DAG is that the part of the chain up to $v$ was executed on $p_1$, whereas the part starting after $v$ on $p_2$. For convenience we will even assume more than this: that for these chains, we also know the assignment $\tau$ for each node. Indeed, since there are at most $P \cdot S$ chains with communications, and splitting each chain (of length at most $n$) between the $S$ supersteps monotonically can happen in at most ${(n+S-1) \choose S}=O(n^S)$ ways, this again increases the number of possible cases to study by a factor $O(n^{P \cdot S^2})$ that is polynomial in $n$.

Note that in each case we are studying, the communication costs for each superstep are already fixed, and the analyzed chains already contribute a known amount of work cost to each processor and superstep. As such, our only goal is to consider the remaining chains (\textit{free chains}) that take part in no communication (and hence each of them needs to be assigned to a single processor entirely), and distribute them between the processors and supersteps in a way that minimizes the total computation cost of the schedule.

This can be accomplished with a method that is essentially a more complex variant of the well-known dynamic programming solution for the multiway number partitioning problem~\cite{GJ79}. That is, we create a boolean table $T$ of $(P \cdot S + 1)$ dimensions, where the first $P \cdot S$ dimensions correspond to a combination of a given processor and superstep, and the last dimension corresponds to the free chains. The first $P \cdot S$ dimensions of $T$ will have cells labeled from $0$ to $n$, whereas the last dimension has cells labeled from $0$ to the number of free chains. A given cell of the table indexed by $(1, 1)=a_{1,1}$, ... ,  $(P, S)=a_{P, S}$ and by $i$ in the last dimension will describe whether it is possible to assign the first $i$ free chains to processors and supersteps such that the number of nodes in each processor $p \in [P]$ and superstep $s \in [S]$ is exactly $a_{p,s}$. Initially, we set each cell of the table to false, except the single cell that describes the exact number of nodes sorted into each processor/superstep earlier due to communications, and has $0$ as its last index.

We can then fill this table out systematically, iterating over the last index $i$ (i.e.\ iterating through the free chains). Given a current chain, there are $P$ different processors that the chain could be assigned to, and for a chain of length $\ell$, there are ${(\ell+S-1) \choose S} = O(\ell^S) \leq O(n^P)$ ways to distribute it into the $S$ supersteps. For each of these possible assignments, we update the corresponding values in the subtable indexed by $(i+1)$ in the last dimension. That is, if we have an assignment $(1, 1)=b_{1,1}$, ... ,  $(P, S)=b_{P, S}$, and some cell $T\left( (1, 1)=a_{1,1} ; \,  ... ; \, (P, S)=a_{P, S}; \;  i \right)$ of the table was set to true, then we also set $T\left( (1, 1)=a_{1,1}+b_{1,1} ; \,  ... ; \, (P, S)=a_{P, S}+b_{P,S}; \;  i+1 \right)$ to true. This allows us to fill the entire table appropriately.

In the end, it simply remains to check all schedules in the last subtable (where all free chains are assigned), compute the cost of their schedules, and select the one with lowest cost; this gives us the best obtainable schedule for the given communication configuration. Note that the table has size $O(n^{P^2+1})$, and there are $O(n)$ free chains and $O(n^P)$ possible assignments for each free chain, so the running time of the algorithm is bounded by $O(n^{P^2+P+2})$, which is indeed polynomial in $n$. Running this process for each possible configuration still results in polynomial time.
\end{proof}

It remains to show that we indeed need at most $P$ supersteps.

\begin{lemma} \label{lem:maxPss}
The total communication cost in the optimal BSP schedule for a chain DAG is at most $(P-1)$.
\end{lemma}

\renewcommand*{\proofname}{Proof}
\begin{proof}
Note that the minimum work cost in a BSP schedule is lower bounded by both $\frac{n}{P}$ and the length $\ell_{max}$ of the longest chain. We show that it is always possible to achieve this work cost with a total of $(P-1)$ communication rounds, and thus get a total cost of $\max(\frac{n}{P}, \ell_{max})+(P-1) \cdot g$. Since a higher number of communication steps could only yield a higher cost, any such solution is clearly suboptimal.

For this, let us set $P'=P$ and $n'=n$ initially, and consider the following iterative procedure. In each step, we take the currently longest chain $\ell'_{max}$. If $\ell'_{max} \geq \frac{n'}{P'}$, then we assign the entire chain to a single processor, and discard both this chain and the processor, hence updating our values to $P' \leftarrow (P'\!-\!1$) and  $n' \leftarrow (n'\! -\! \ell'_{max})$, and continue the process. Otherwise, if $\ell'_{max} < \frac{n'}{P'}$, then we terminate the iterative process, and develop a schedule for the remaining chains in a single step. For this, we simply number the remaining nodes in all chains from $1$ to $n'$, considering the chains in an arbitrary order and iterating though each chain from beginning to end, and then we assign the nodes $\{ 1, ..., \frac{n'}{P'} \}$ to the first processor, the nodes $\{ \frac{n'}{P'}+1, ..., \frac{2 \cdot n'}{P'} \}$ to the second processor, and so on.

We can then execute the nodes on the last $P'$ processors in any desired order; whenever we reach a node on $p_1$ with a successor that is assigned to a different processor $p_2$, we simply close the computation phase, and add a communication phase to send this value from $p_1$ to $p_2$ (and possibly others that are available). Since there are at most $(P'-1)$ nodes with a successor assigned to a different processor, this results in a communication cost of at most $P'-1 < P - 1$. The total work cost on each of the last $P'$ processors is exactly $\frac{n'}{P'}$; note that our updates of $P'$ and $n'$ always ensure that $\frac{n'}{P'} \leq \frac{n}{P}$. On the processors discarded earlier (assigned to a single chain), we simply compute the next node in each step, hence the work cost on these is bounded by $\ell_{max}$.
\end{proof}

\subsubsection{Discussion} Note that the polynomial algorithm above can easily be adapted to the case of having non-zero latency $L>0$: since each of the solutions we consider has a known number of supersteps, we can simply add the latency costs to each candidate solution and compare the different solutions this way.

As for adapting the algorithm to the other communication models in Table~\ref{tab:commodels}: switching to broadcast operations has no advantage on chain DAGs, since each node has outdegree at most $1$, so each value is required on at most one further processor. In case of free data movement, some further configurations of communication become valid, e.g.\ if $v$ is transferred from $p_1$ to $p_2$ and then from $p_2$ to $p_3$, then we can assume that the chain part starting at the successor of $v$ is computed on $p_3$. However, this is only a slight change in the conditions for evaluating the validity of each configuration.

We also note that the earlier work of~\cite{chainsHard} also studies optimal BSP scheduling for chain DAGs. In particular, this work claims that the BSP scheduling problem is already NP-hard on chain DAGs; this might look surprising at first glance, since it seemingly contradicts our result. The proof of~\cite{chainsHard} uses a reduction to the partitioning problem (i.e.\ multiway number partitioning with only $2$ classes) in the case of $S\!=\!1$ supersteps, and invokes the fact that this partitioning problem is NP-complete. This raises the underlying technical question of how our problem input is actually encoded. That is, given some integers $a_1, ..., a_k$, the dynamic programming algorithm for partitioning is a so-called pseudo-polynomial time algorithm. This means that it only takes polynomial time in terms of the sum of the numbers $\sum a_i$ (which corresponds to $n$ in our case), but if the numbers $a_1, ..., a_k$ are provided in binary form in the input, then this $\text{poly}(n)$ time is in fact exponential in the number of input bits; as such, from a complexity theoretical perspective, the problem is still NP-hard. 

For the BSP scheduling of chain DAGs, this means the following. If we restrict our interest to chain DAGs, then it indeed becomes possible to select a different, more compressed encoding format for our inputs (only describing the length of each chain, as a binary number) which makes the problem NP-hard. However, if we consider the input of the problem to always be a DAG description, listing the nodes and edges in the DAG (which we believe to be a more natural problem formulation in general), then the problem is not NP-hard anymore, since our algorithm above has a running time that is polynomial in the number of nodes $n$.

We also point out that for this polynomial solution, it is once again a critical assumption that the number of processors $P$ is a constant, and does not scale arbitrarily with $n$. If $P$ were a part of the input (as assumed in some other works), then the problem very easily becomes unreasonably hard; e.g.\ a reduction from $3$-partition shows that the problem would already be NP-hard for the simple case of chain DAGs and only $S\!=\!1$ superstep, as also noted in~\cite{chainsHard}.

Finally, we note that~\cite{chainsHard} also carries out a more detailed analysis of the largest possible number of supersteps that is required for an optimal schedule in chain graphs, showing both an upper and lower bound of essentially $\frac{P}{2}$ for this. However, for our algorithm above, the simpler upper bound of $P$ was already sufficient.

\subsubsection{Connected chain DAGs} We now discuss the modifications required if our chains also have a common source node $v_0$.

\renewcommand*{\proofname}{Proof of Theorem~\ref{th:chains} for connected chain DAGs}

\begin{proof}
The simplest possible schedule in this case is to have a single superstep, assigning the entire DAG to $p_1$ at a cost of $n$; we will compare this to the rest of the solutions.

If our schedule has $S \geq 2$ supersteps (and actual communication between them), then naturally $v_0$ will be in the first superstep. Furthermore, we can show that it is always optimal to place $v_0$ in a separate superstep. In particular, assume that some chain up to another node $v$ is also assigned to superstep $1$ (and hence $\pi(v_0)$). If $v$ is the last node of the chain, or if it has a successor $v'$ assigned to $\pi(v')=\pi(v_0)$, then we can simply reassign the chain to superstep $2$: this decreases $C_{work}\,\!^{(1)}$ by the length of the chain, increases $C_{work}\,\!^{(2)}$
by at most the same amount, and incurs no communication. Otherwise, if $v$ has a successor $v'$ assigned to $\pi(v') \neq \pi(v_0)$, then we can reassign the whole chain up to $v$ to processor $\pi(v')$ and superstep $\tau(v')$. This still requires us to send one value from $\pi(v_0)$ to $\pi(v')$ (we can do this in the same superstep as before), and as before, the work cost of superstep $1$ is decreased by at least as much as the work cost of superstep $\tau(v')$ is increased.

After all these reassignments, $v_0$ is in a superstep of its own, and only affects the remaining part of the DAG in a very simple way: if we have a chain starting with a processor different from $\pi(v_0)$, then this will incur an extra communication step from $\pi(v_0)$. However, with $P$ processors and a single value $v_0$, this means at most $(P-1)$ extra communication steps. Indeed, applying the method in Lemma~\ref{lem:maxPss} still provides a solution of cost at most $\max(\frac{n-1}{P}, \ell_{max})+2 \cdot (P-1) \cdot g + 1$, even if all processors appear in the beginning of some chain and thus $v_0$ is sent to them all in the first superstep, so $C_{comm}\,\!^{(1)}=P-1$. This shows that the optimal schedule still consist of at most $(2 \cdot P - 1)$ supersteps, since otherwise it would produce a higher cost.

Hence we can still execute the same dynamic programming approach as before, but now allowing $(2 \cdot P - 2)$ different communication steps, including steps that send the value of $v_0$. Naturally, the communications of $v_0$ must be handled slightly differently. As they determine the superstep $s$ when the value of $v_0$ becomes available on each processor $p$, the solutions need to be filtered out according to this: before superstep $s$, no segment of any chain can be assigned to this $p$. Nonetheless, it is still only a technicality to check the validity of each configuration (in polynomial time), and then divide the free chains into supersteps and processors according to the availability of $v_0$. The table can be filled in each configuration the same way as before to find the optimal schedule.
\end{proof}

In the FS model, the method above only needs a minor modification, as $v_0$ may now be relayed between the processors; however, the communication steps in each configuration still determine when $v_0$ becomes available on each processor. On the other hand, in the broadcast models, the algorithm actually becomes significantly simpler, since in a single round of communication, $v_0$ reaches all the processors in superstep $1$, and from this point the construction behaves like a chain DAG, not affected by $v_0$ at all.

\subsection{Proof of Theorem~\ref{th:2level}}

We prove Theorem~\ref{th:2level} through a reduction from the clique problem: given an undirected graph $G'(V',E)'$ and a number $k$, the task is to find a subset of $k$ nodes in $V'$ such that any two of them is adjacent. This problem is long known to be NP-complete~\cite{GJ79}. Our reduction will only have $P=2$ processors, hence it applies to any of the communication models equivalently.

\renewcommand*{\proofname}{Proof of Theorem~\ref{th:2level}}

\begin{proof}
Consider some fixed parameters $M_1$, $M_2$. The main tool in our construction is a so-called \textit{2-level block} $U=(U^{(1)}, U^{(2)})$, which consists of (i) a set $U^{(1)}$ of $M_1$ nodes which will be on the first level of the DAG, and (ii) another set $U^{(2)}$ of an unspecified number of nodes on the second level of the DAG, such that each node in $U^{(1)}$ has an edge to each node in $U^{(2)}$. We call $|U^{(2)}|$ the \text{size} of the block; in our construction, this will always be a multiple of $M_2$. Intuitively, such a block will ensure that each node in $U^{(2)}$ has to be assigned to the same processor; if the block is split, then each node in $U^{(1)}$ will need to be communicated to at least one of the processors used in $U^{(2)}$, hence resulting in a communication cost of at least $M_1$, which will be too large to provide a reasonable solution.

Given a clique problem on a graph $G'(V',E)'$ on $n'=|V'|$ nodes and a parameter $k$, we turn this into a BSP scheduling problem as follows.
\begin{itemize}[leftmargin=13pt, itemsep=4pt]
\item For each node $v' \in V'$, we create a $2$-level block $U_{v'}$ of size $|U_{v'}^{(2)}|=M_2$. Besides, we add two further blocks $U_{p_1}$ and $U_{p_2}$ of sizes $(2 \cdot n' - k) \cdot M_2$ and $(n' + k) \cdot M_2$, respectively.
\item Then for each edge $(v'_1, v'_2) \in E'$, we add a single node $v_{e'}$ on the first level of the DAG, and draw an edge form $v_{e'}$ to (i) an arbitrary node in $U_{v'_1}^{(2)}$, (ii) an arbitrary node in $U_{v'_2}^{(2)}$, (iii) an arbitrary node in $U_{p_1}^{(2)}$.
\item We add $|E'|-{k \choose 2}$ more nodes $\hat{v}_i$ to the first level of the DAG, and add an edge from $\hat{v}_i$ to (i) an arbitrary node in $U_{p_1}^{(2)}$, and (ii) an arbitrary node in $U_{p_2}^{(2)}$.
\item Finally, we add $n' \cdot M_1 + {k \choose 2}$ isolated nodes to our DAG; we will imagine these nodes to be a part of level $1$.
\item For our parameters, we select $M_1:=|E'|+1$ and $M_2= M_1 \cdot (n'+1) + 2 \cdot |E'| \cdot g$. Note that the size
$n$ of the resulting DAG construction is still polynomial in the size of $G'$.
\item In the derived BSP scheduling problem, we set $P=2$, $L=0$, choose $g$ as a small constant (e.g.\ $g=2$), and set the allowed cost to $C_0 = \frac{n}{2} + (|E'| - { k \choose 2}) \cdot g$.
\end{itemize}

We first outline the main idea of the construction, and then discuss the technical details separately.
\begin{enumerate}[leftmargin=13pt, itemsep=4pt, label=(\roman*)]
\item The computational cost of any schedule is at least $\frac{n}{2}$, so any schedule within the allowed cost can have at most $|E'| - {k \choose 2}$ communication steps. As outlined before, assigning a block to several processors results in a cost of $M_1> |E'|$; hence in each valid schedule, the level-$2$ part of each block must entirely be assigned to the same processor.
\item In the second level of our DAG, the blocks in our construction essentially allow us to adapt the packing technique in the recent work of \cite{hyperDAG} to our BSP scheduling problem, ensuring that we have the desired number of blocks assigned to both processors in a reasonable solution. That is, assigning both $U_{p_1}^{(2)}$ and $U_{p_2}^{(2)}$ to the same processor would result in a work cost of at least $3 \cdot n' \cdot M_2$, which is significantly larger than $C_0$. As such, $U_{p_1}^{(2)}$ and $U_{p_2}^{(2)}$ are assigned to different processors; assume w.l.o.g. that these are $p_1$ and $p_2$, respectively.
\item Now consider the second level of the blocks $U_{v'}$, which all have size $M_2$. One can then observe that we need to assign exactly $k$ of these to $p_1$ and $(n'-k)$ of them to $p_2$, otherwise, the work cost will be split in a very imbalanced way between $p_1$ and $p_2$ on the second level, and since most of the nodes are in this level, this also makes the total work cost too high.
\item As for communication costs, the nodes $\hat{v}_i$ all have a successor on both processors, hence they will all incur a communication step, regardless of which processor they are assigned to. Also, since $U_{p_1}^{(2)}$ is assigned to $p_1$, each value $v_{e'}$ is required on processor $p_1$ at some point; as such, node $v_{e'}$ will incur a communication step if and only if one of its other two successors is assigned to $p_2$.
\item The construction allows us to distribute the work cost evenly between the two processors on both levels, hence resulting in a total work cost of $\frac{n}{2}$. As such, a valid solution must have a communication cost of $(|E'| - { k \choose 2}) \cdot g$ at most, which means at most $2 \cdot (|E'| - { k \choose 2})$ communication steps, half of them from $p_1$ to $p_2$ and the rest in the reverse direction.
\item The isolated nodes incur no communication, and neither do the first levels of blocks if they are assigned to the same processor as the second level. The nodes $\hat{v}_i$ already incur $|E'| - { k \choose 2}$ communication steps. As such, the communication cost will be determined by the number of $v_{e'}$ that have a successor assigned to $p_2$. To limit the number of communications to $2 \cdot (|E'| - { k \choose 2})$ altogether, we need to ensure that the nodes $v_{e'}$ also incur at most $|E'| - { k \choose 2}$ communication steps, i.e.\ there are ${ k \choose 2}$ distinct nodes $v_{e'}$ that have all of their successors assigned to $p_1$.
\item Altogether, this means that we need to select $k$ of the second-level blocks $U_{v'_1}^{(2)}$ to assign to $p_1$, such that the number of edges induced by these $k$ nodes is at least ${ k \choose 2}$, i.e.\ the nodes form a  clique. Hence a BSP of cost $C_0$ exist if and only if $G'$ contained a clique of size $k$.
\end{enumerate}

More formally, given a clique of size $k$ in $G'$, we can assign (both levels of) the corresponding $k$ blocks and $U_{p_1}$ to $p_1$, and (both levels of) the remaining blocks to $p_2$. This ensures that the number of level-$2$ nodes assigned to both processors is $2 \cdot n' \cdot M_2$, and the number of level-$1$ nodes assigned to $p_1$ and $p_2$ is $(k+1) \cdot M_1$ and $(n'-k+1)\cdot M_1$, respectively. We assign all the nodes $v_{e'}$ also to $p_1$, the nodes $\hat{v}_i$ to $p_2$, and as for the remaining isolated nodes, we assign $(n'-k) \cdot M_1$ of them to $p_1$, and $k \cdot M_1 + {k \choose 2}$ of them to $p_2$. This ensures that the number of level-$1$ nodes assigned to both processors is $(n'+1) \cdot M_1 + |E'|$. As such, we can assign all nodes on the first and second levels, respectively, to supersteps $1$ and $2$, and the total work cost of the two supersteps will be $\frac{n}{2}$. Finally, the nodes $\hat{v}_i$ will be sent from $p_2$ to $p_1$, and the $v_{e'}$ not corresponding to the clique edges will need to be sent from $p_1$ to $p_2$, resulting in a communication cost of $(|E'|-{k \choose 2}) \cdot g$. This indeed provides a BSP schedule of total cost $C_0$.

Conversely, assume that a BSP schedule of cost $C_0$ exists. As discussed before, this implies that the second level of all blocks is assigned to a single processor, $U_{p_1}^{(2)}$ and $U_{p_2}^{(2)}$ are assigned to $p_1$ and $p_2$, respectively, and exactly $k$ of the $U_{v'}^{(2)}$ are assigned to $p_1$. The work cost of the solution must be $\frac{n}{2}$ at least. For the communication cost to be at most $(|E'| - { k \choose 2}) \cdot g$, we can have at most $2 \cdot (|E'| - { k \choose 2})$ communicated values; after subtracting the $\hat{v}_i$, only $|E'| - { k \choose 2}$ communicated values remain. All the edges in $E'$ not induced by the $k$ chosen nodes will result in a communication step, so to keep the number of these steps below $(|E'| - { k \choose 2})$, the $k$ nodes must induce ${k \choose 2}$ edges at least, i.e.\ they must form a clique. 

It remains to discuss a few technical details regarding our parameters. Note that the number of nodes in the construction is 
\[ n = 2 \cdot \left( M_1 \cdot (n'+1) + |E|' \right) + 4 \cdot n' \cdot M_2 = O(n'^2 \cdot |E'|) \, . \]
For point (ii) of the argument above, we also need to show that $3 \cdot n' \cdot M_2 > C_0$; indeed, since
\[ C_0 = M_1 \cdot (n'+1) + |E|' + 2 \cdot n' \cdot M_2 + \left(|E'| - { k \choose 2} \right) \cdot g \, , \]
we only need to show
\[ n' \cdot M_2 > M_1 \cdot (n'+1) + |E|' + \left(|E'| - { k \choose 2} \right) \cdot g \, , \]
which follows from our choice of $M_2=M_1 \cdot (n'+1) + 2 \cdot |E'| \cdot g$. Also, for point (iii), note that if we choose to assign more or less than $k$ of the blocks $U_{v'}$ to $p_1$, then one of the processors will be assigned at least $(2 \cdot n' +1) \cdot M_2$ nodes on the second level. Even if all nodes on the first level are assigned to the other processor, the number of nodes on this processor will be too high. In particular, we have $(2 \cdot n' +1) \cdot M_2 > C_0$; this is equivalent to
\[ M_2 > M_1 \cdot (n'+1) + |E|' + \left(|E'| - { k \choose 2} \right) \cdot g \, , \]
which once again holds due to our choice of $M_2$. As such, a work cost of at least $(2 \cdot n' +1) \cdot M_2$ indeed cannot provide a valid BSP schedule, so second-level blocks must be distributed in a completely balanced fashion.
\end{proof}

\subsection{Proof of Theorem~\ref{th:trees}}

For in-trees, we provide a reduction from the so-called $3$-partition problem. In this problem, we are given a list of $m$ positive integers $a_1, \ldots, a_m$, with $m= 3 \! \cdot \! m'$ for some other integer $m'$. The integers are known to satisfy $\sum_{i=1}^{m}\, a_i \, = \, m' \cdot T$ for some integer $T$, and also $\frac{T}{4} _{\!} < _{\!} a_{i\!}< _{\!} \frac{T}{2}$ for all $i \in [m]$. The goal is to divide the numbers into $m'$ distinct subsets of size $3$ each, such that the sum in each subset is exactly $T$. This problem is known to be NP-hard~\cite{GJ79} (in the strong sense, i.e.\ even if the input has length $m \cdot T$).

Our construction to convert a $3$-partition problem to a BSP scheduling task in in-trees will be described in several steps. We first describe the gadgets used in the proof, and then prove the main properties of our construction. We then show in a separate step how to embed the input $3$-partition problem in our construction. Finally, we sort out the technical details.

For ease of presentation, we will often describe our BSP schedule as if it was a classical scheduling approach, i.e.\ as if the nodes within each superstep were ordered in an arbitrary topological ordering, and hence each single node is executed in a given integer time step.

\subsubsection{Main ideas and tools}

Consider some integer parameters $M$, $M_0$, $A$, $B$; we will decide the values of these later. Intuitively, we will have $M>M_0$ and $B>A$, and $M$ and $M_0$ will be much larger than $B$ and $A$. The parameters will satisfy $M_0 + A < M < M_0 + B$. Altogether, our construction will initially consist of $n= 4 \cdot m \cdot (2 \cdot M + M_0) + 4$ nodes, and $P _{\!} = _{\!} 4$ processors. We will choose $g$ as an arbitrary constant, e.g.\ $g=2$.

The main gadgets in our construction will be so-called \textit{cones}: a cone on $k$ nodes consists of a single node $u$, and $(k-1)$ further nodes that have an edge to $u$. We call node $u$ the \textit{top} of the cone. A cone essentially ensures that we need to assign most of these $k$ nodes to the same processor; indeed, if $k_0$ of the predecessors of $u$ are assigned to a different processor than $u$, then this results in a communication cost of at least $k_0$, since these values all need to be sent to $u$. Due to this, we will refer to the processor assigned to the top node of the cone as the processor assigned to the cone itself.

The allowed scheduling cost in our derived BSP problem will be set to $C_0 = \frac{n}{P} + 3 \! \cdot \! m \! \cdot \! g$. Since the work cost is at least $\frac{n}{P}$, this will mean that any BSP schedule can have at most $ 3 \! \cdot \! m$ communication rounds. We will show that this amount of communication is indeed needed in any reasonable schedule, and hence the work cost needs to be exactly $\frac{n}{P}$. In other words, this implies that every processor needs to execute a node in every time step.

Our construction will contain a \textit{final cone} of size $\frac{n}{P}$ where the top node $u_0$ is a sink node of the DAG. Furthermore, there will be $3 _{\!}\cdot_{\!}m$ further cones in the DAG (called \textit{semi-final} cones, to be described in detail later), each of size at least $M_0$, and the top node of each of them will have an edge to $u_0$. Altogether, this implies that the top node of each semi-final cone has to be assigned to a different processor than $u_0$; if not, then the given top node and $u_0$ together have more than $\frac{n}{P}+M_0$ immediate predecessors. This cannot be a solution below the cost limit: if at least $\frac{n}{P}+M_0-3_{\!}\cdot_{\!}m$ of these predecessors are on the same processor as $u_0$, then that yields a too high work cost (since our parameters will ensure $M_0>3 \cdot m \cdot (g+1)$), otherwise we need to communicate more than $3 _{\!}\cdot_{\!}m$ values to this processor, yielding a too high communication cost.

This implies that to each of the $3 _{\!}\cdot_{\!}m$ semi-final cones, we will need to assign a different processor than $\pi(u_0)$. This already gives us $3 _{\!}\cdot_{\!}m$ values that need to be communicated to $u_0$ from a different processor, hence establishing a communication cost of at least $3 _{\!}\cdot_{\!}m$, and as a result, a work cost of exactly $\frac{n}{P}$.

Note that this also means that each of the $(\frac{n}{P}-1)$ predecessors of $u_0$ within the final cone needs to be assigned to $\pi(u_0)$, otherwise the communication cost grows even higher. As such, there will be one processor in our schedule, let us call it $p_4$, which is not particularly interesting: it will compute the source nodes of the final cone in the first $(\frac{n}{P}-1)$ steps, $u_0$ in the last step, and nothing else in the DAG.

Also note that having $3 _{\!}\cdot_{\!}m$ values to transfer from semi-final cones to $p_4$, which is exactly the allowed communication cost, implies that one of these values need to be sent to $p_4$ in every round of communication. This means that whenever we need to communicate any value within our DAG, we always need to have a new semi-final cone already computed and ready. Hence intuitively, the semi-final cones ensure in our solution that each of the communication rounds must happen \textit{at earliest} at a given point in time, since the next communication round can only be used once a new semi-final cone is computed.
 
Besides the cones, our construction will also have a critical path of length $\frac{n}{P}$; hence to have a work cost of $\frac{n}{P}$, we will need to make sure that the $i$-th node of this path is computed in the $i$-th step for all $i \in [\frac{n}{P}]$. However, there will be further cones attached to this path, which, intuitively speaking, require you to have communication rounds \textit{at latest} in given supersteps, otherwise we cannot execute all the predecessors of the path node and hence cannot compute the $i$-th node in time. Together with the semi-final cones above, this will ensure that communications have to happen \textit{exactly} at given time steps. This will be crucial for our reduction.

\subsubsection{Construction outline} \label{app:subs:outline}

Our construction will essentially consist of $m$ identical copies of a set of gadgets. We now describe these gadgets and discuss the properties they ensure. We will later show how to customize each gadget slightly to represent a concrete number $a_i$ of the $3$-partition problem.

As for the semi-final cones, $2 \! \cdot \! m$ of these will simply be cones of size $M$, with their top nodes having an edge to $u_0$. The remaining $m$, however, we be turned into so-called \textit{triple cones}. For each triple cone, we create nodes $u_1$, $u_2$, $u_3$ that are the top nodes of a cone of size $M_0$, $B$ and $A$, respectively, and we then add the edges $(u_3, u_2)$ and $(u_2, u_1)$ to our DAG (besides the edge $(u_1, u_0)$, which was already mentioned before).

As for the critical path, we will divide it into $m$ continuous \textit{segments} of size $(2 \! \cdot \! M + M_0)$, plus a single final node. For each of the segments, we add $3$ separate cones such that the top node of the cone is a given node of the segment. In particular, for each segment, we add (i) a cone of size $(M-A)$ with its top being the $(M+1)$-st node of the segment (from the source), (ii) a cone of size $(M-B)$ with its top being the $(2_{\!} \cdot _{\!} M + 1)$-st node of the segment, and (iii) a cone of size $M_0$ with its top being the \ $(2_{\!} \cdot _{\!} M + M_0 +1)$-st node, i.e.\ the node immediately after the segment. These $3$ top node split the segment into $3$ parts.

This implies that the $i$-th top node on the critical path (let us denote it by $v_i$) has at least $i \cdot (2 \! \cdot \! M_0 - B)$ predecessors. We will use this to show that at least one communication round is always needed between $v_{(i-1)}$ and $v_i$ for all $i \in [3 \! \cdot \! m]$. For $i=1$, this is clear: the top node $v_1$ has $(2 \cdot M - A)$ predecessors, and recall that it has to be computed exactly in the $(M+1)$-st step. With $(2_{\!}\cdot _{\!} M - A) > M$, this results in a too high work cost if done on a single processor; as such, these predecessors must be split among at least two processors, and hence at least one communication round is required.

Also recall that one of the semi-final cone values has to be transmitted in each communication round, so in the first $M$ steps, another work cost of at least $M$ has to be invested into one of the semi-final cones (either a triple cone or a regular one). Since there is at most $A$ work cost available that is not used on the critical path or the cone attached to it, the earliest possible time for finishing the semi-final cone (and hence having the computation step) is at time $(M-A)$. Also, since all these nodes have to be computed, the processors can compute at most $A$ further nodes in the first $M$ steps besides those mentioned so far.

The same argument can be continued in an inductive way. Node $v_i$ has $i \cdot (2 \cdot M_0 - B)$ predecessors at least. Up to the time when we compute $v_{(i-1)}$, there are at most $(i\!-\!1) \cdot B$ nodes in the DAG that have been computed, except for (i) the critical path component up to $v_{(i-1)}$, (ii) $i$ semi-final cones that were completely finished, and (iii) the final cone which we disregarded for convenience. As such, out of the at least $(2 \cdot M_0 - B)$ nodes in the path between $v_{(i-1)}$ and $v_i$ and the cone attached to $v_i$, there are still at least $(2 \cdot M_0 - B) - (i-1) \cdot B$ nodes that are not yet computed. There are either $M$ or $M_0$ time units between computing $v_{(i-1)}$ and $v_i$; if we have $(2 \cdot M_0 - B) - m \cdot B > M$, this ensures that we cannot execute all the remaining nodes on a single processor, so another round of communication is required before $v_i$. Once again, at least $M$ work cost is needed to finish another semi-final cone by this round; this shows that the number of extra nodes that are computed (besides the path up to $v_i$ and the new semi-final cone) is again upper bounded by $(i \! \cdot \! B)$.

Hence one round of communication is needed between each two top nodes on the critical path; since we can only have $3 \cdot m$ communication steps altogether, this means that there is \textit{exactly} one round of communication between each two top nodes. This also ensures the convenient property that the cones attached to the path are all entirely assigned to single processor: otherwise, they would incur further communication cost without reasonably splitting the work cost between $v_{(i-1)}$ and $v_i$. As such, between each pair $v_{(i-1)}$ and $v_i$, the path consists of two continuous parts, with the nodes in the first and second parts are assigned to to different processors.

As a detail, note that we also add $2$ isolated nodes to the construction to increase the number of sink nodes from $2$ to $4$, and hence ensure that all processors can still compute in the last step.

\subsubsection{Scheduling the construction} \label{app:subs:conebehave}

We now go over the segments in another induction from source to sink, and show that each segment can only be scheduled in a specific way without violating the allowed cost $C_0$. Let us consider the first segment in detail. Recall that all three parts of this segment must have exactly one communication round.

We begin with the first part of the segment, up to $v_1$. This has a path of length $(M+1)$ ending in $v_1$, and a cone of size $(M-A)$ with $v_1$ as its top node. To compute $v_1$ by time step $(M+1)$, we need to compute all the $(2M-A)$ predecessors of $v_1$; this leaves only $(M+A)$ further nodes we can compute on our $3$ processors in the first $M$ times steps. Note that we also need to finish one of the semi-final cones by time step $M$ in order to send a value to $p_4$ in the first communication round. As the triple cones all have size $M_0+A+B > M$, we can only achieve this by computing one of the semi-final cones of size $M$.

There are also several other observations we can make. Firstly, to avoid further communication, the semi-final cone must be assigned entirely to a single processor; let us call this w.l.o.g.\ $p_3$. This implies that the communication must happen after exactly $M$ steps (at the last possible point in time), i.e.\ we have $C_{work}\,\!^{(1)}=M$. As such, another processor $p_1$ must process the first $M$ nodes of the path, and a third one $p_2$ must process the $(M-A)$ nodes in the cone on the path (and also $v_1$, but this is already in the computational phase of the next superstep). As such, in the communication phase of superstep $1$, $p_1$ sends the $M$-th node of the path to $p_2$, and $p_3$ sends the top node of a semi-final cone to $p_4$. Besides this, $p_2$ can compute $A$ further nodes in another part of the DAG.

Now consider the second part of the segment, up to $v_2$, and time steps $(M+1)$ to $2 \! \cdot \! M$. This is again a path of length $(M+1)$ ending in $v_2$, and a cone of size $(M-B)$ attached to $v_2$. However, also note that the beginning of the path must be assigned to $\tau(v_1)=p_2$, even up to the $(2M-B)$-th node, otherwise the remaining path and the cone would contain more than $M$ nodes altogether, and hence would require another communication step. This means that most of the capacity of $p_2$ in this interval is used on the path. In particular, we cannot use $p_2$ to compute another semi-final cone, since even with the leftover computations from earlier, we have at most $(A + B)$ steps to do this, and both $M$ and $M_0$ are significantly larger than $(A + B)$.

Hence another semi-final cone must be computed on $p_1$ or $p_3$. Note that the role of these two processors is symmetrical here: the nodes previously computed by them have no direct connection to the rest of the DAG, and they both sent a value in the first communication step, and did not receive one. Assume that the semi-final cone is computed on $p_1$. This cannot be a triple cone, since even with the leftover from the earlier steps, we can compute at most $(A+M) < (A + B + M_0)$ nodes on a single processor, and using more processors requires extra communication. As such, $p_1$ must compute a regular semi-final cone of size $M$ in the interval $[M+1, 2 \cdot M]$, which again shows that the second round of communication happens after time step $2 \cdot M$. This means that $p_2$ computes the path up to length $2 \! \cdot \! M$, and $p_3$ computes the attached cone, and then node $v_2$ in step $(2 \! \cdot \! M +1)$. In the communication round, we send a value from $p_2$ to $p_3$ and from $p_1$ to $p_4$. There are $B$ leftover computations that $p_3$ can execute in this time interval.

Finally, consider the third part up to $v_3$; recall that this is a shorter part, corresponding only to the time steps $(2 \! \cdot \! M+1)$ to $(2 \! \cdot \! M+M_0)$. Once again the beginning of the path is assigned to $p_3$, up to at least node $(2 \! \cdot \! M + M_0)-B$, otherwise the rest could not be completed in time by a single processor, even with the leftover from earlier. Then $2 \! \cdot \! B$ is again much smaller than $M_0$, so the semi-final cone must be computed by a processor other than $p_3$. However, $p_1$ and $p_2$ only have $M_0$ time steps in this interval, plus $p_2$ has $A$ leftovers, and $M_0 + A < M$; recall that $p_3$ also has $B$ leftovers, but this cannot be combined in a semi-final cone with $p_1$ or $p_2$.

This means that it is not an option to compute a regular semi-final cone by time step $(2\!\cdot \! M + M_0)$; the only way to transmit a value to $p_4$ in time is to compute a triple cone. Note that a triple cone has $A+B+M_0$ nodes, and the path component between $v_2$ and $v_3$ has $2 \! \cdot \! M_0$ nodes, so this is also only possible if we use all the $3 \! \cdot \! M_0$ computations in our time interval, plus the $(A+B)$ leftovers from earlier; in other words, no leftovers remain after this step. The top of the triple cone cannot be finished before step $(2 \! \cdot \! M+M_0)$, since then all three processors would need to actively work on the path and the attached cone in the remaining steps, resulting in extra communication; hence the top node $u_1$ is only computed in step $(2 \cdot M+M_0)$, and communication can only follow afterwards, again at the last possible point in time.

This also means that $p_3$ is computing the path until node $(2 \! \cdot \! M+M_0)$, and not doing any other computations in this interval. Then the $B$ leftover steps of $p_3$ from earlier cannot be used anywhere else than in the cone of size $B$ in the triple cone; all other cones are smaller or larger (i.e.\ sizes $A$ or $M_0$), so otherwise some cone would be split, incurring extra communication. Since these leftovers are computed between time steps $(M+1)$ to $2 \! \cdot \! M$, the cone of size $A$ in the triple cone must be computed by the leftovers of $p_2$ in the interval $1$ to $M$, to make $u_3$ already available to $u_2$ by time step $2 \cdot M$ at the latest. The cone attached to $v_3$ and the largest cone in the triple cone are then assigned to $p_1$ and $p_2$ between time steps $(2 \! \cdot \! M+1)$ to $(2 \! \cdot \! M+M_0)$ in either way; we will assume that $p_1$ computes the attached cone and $p_2$ computes the semi-final cone, since this nicely maintains the periodic use of the processors for the next segment that follows after $v_3$.

Note that the assignment described above indeed provides a schedule with the desired cost: all processors are computing in each time step, so the work cost is $(2 \cdot M+M_0)$. The assignments in the triple cone also require further communication: sending $u_3$ from $p_2$ to $p_3$ in superstep $1$, and sending $u_2$ from $p_3$ to $p_2$ in superstep $2$. However, this does not increase the cost of the given $h$-relations.

The same proof can be used in an inductive fashion for the next segments that follow. In particular, the arguments above ensure that we need to use all the $3$ processors in all $(2 \cdot M+M_0)$ time steps, so no other nodes are computed. The only connection between the already assigned part of the DAG and the rest is the node $v_3$, which is assigned to $p_1$; however, the role of the processors in the rest of the schedule is symmetrical, so this is not a restriction in any sense. The number of regular semi-final cones is decreased by $2$, the number of triple cones by $1$; as such, the unassigned part of the DAG behaves exactly like the same construction, but now containing only $(m-1)$ segments.

This induction proves that in each segment of the DAG, a valid scheduling in this in-tree is only possible by repeating the steps described above for each segment from the beginning to the end of the critical path.

\subsubsection{Making the communication times fixed}

In our actual reduction, the different segments will not be so flawlessly separated from each other as described above; in particular, we might have a very small number of leftover nodes from the previous segments. While the number of these nodes will be much less than $A$ or $B$, they could still influence with the behavior of the construction, unfortunately, so we need a further technical step to prevent this.

More specifically, we will sort the segments into consecutive triplets, i.e.\ segments $1$ to $3$, segments $4$ to $6$, and so on, referring to each such triplet as a $3$-segment. Note that each such $3$-segment has exactly $9$ communication steps. The main idea of the reduction is that each triple cone will correspond to a number $a_i$ in the original $3$-partition problem. Recall from above that any schedule assigns each triple cone to a given segment, so a $3$-segment corresponds to $3$ consecutive numbers. We will ensure that the schedule is only valid if the $3$ numbers (i.e.\ triple cones) assigned to each of our $3$-segments sum up to exactly $T$.

However, as mentioned, this method will result in a few leftover nodes within the $3$-segment, from the first to the second segment, and from the second to the third segment (but no leftover after the entire $3$-segment). The problem with this is that e.g.\ in segment $2$, these leftover nodes from segment $1$ might be used to pre-compute some nodes of the semi-final cones earlier, which might allow us to execute e.g.\ the first communication of the segment after \textit{less than} $M$ steps. This in turn can allow us to assign $(A+1)$ nodes to processor $p_2$, and only $(M-1)$ to $p_1$, in contrast to what we described above, and with that the correctness of the reduction would not be guaranteed.

In order to resolve this, we further expand our construction to ensure that the communication steps happen at exactly the fixed time steps described in Section~\ref{app:subs:conebehave}. More specifically, we create $4$ identical copies of the DAG described so far (and also increase the number of processors to $P=16$, so the work cost $\frac{n}{P}$ remains unchanged). While this might sound like a major modification, in fact, each of the $4$ copies of the construction will still behave in the same way as outlined before. In particular, the same induction method as in Section \ref{app:subs:outline} still shows that a separate communication round is needed between each $v_{(i-1)}$ and $v_i$: there are still at most $4 \cdot (i-1) \cdot B$ leftover nodes in the DAG until computing $v_{(i-1)}$, so $(2 _{\!} \cdot _{\!} M_0 - B) - 4 _{\!} \cdot _{\!} m _{\!} \cdot _{\!} B > M$ implies that the next part of the path again requires a communication step, and since we also need to compute $4$ semifinal cones (of size at least $M$), the number of leftover nodes for the next induction step is at most $(4 _{\!} \cdot _{\!} i _{\!} \cdot _{\!} B)$.

It then follows that each of the $4$ copies can once again only be scheduled in the way outlined in Section \ref{app:subs:conebehave}. In particular, the same arguments show that the $4$ groups of $A$ leftover nodes from the first part and the $4$ groups of $B$ leftover nodes from the second part of the segment can only be on $8$ distinct processors; this means that none of these leftovers can be combined in the same cone anyway, and thus each of the $4$ segments can still only be scheduled as discussed before. As a technical note, in the third part of the segment, we can now assume $B> 4 _{\!} \cdot _{\!} A$ for simplicity to make it easier to see that the size-$B$ leftover groups can only be used for the cones of size $B$.

The reason for creating $4$ distinct copies of the DAG is as follows. In our reduction, we will adjust the cone sizes in one of the $4$ copies of the construction (we call this the \textit{main copy}) to model the $3$-partition problem. However, as mentioned before, this embedding of $3$-partition in the main copy will also result in up to $3$ very small groups of leftover nodes between the segments, on (at most) $3$ distinct processors. Having $4$ copies ensures that in each part of every segment, there is at least one copy where none of these inter-segment leftover nodes is used, and hence it is guaranteed that the semi-final cones in this copy are indeed only finished after $M$/$M$/$M_0$ time steps (in the first/second/third part of a segment, respectively). As such, regardless of the remaining $3$ copies, the earliest possible times to execute each communication step in our schedule are still the time steps discussed in Section \ref{app:subs:conebehave}.

We note that this method only ensures that the communication times are fixed; the small leftover groups may still be used to pre-compute a part of some cones, and hence finish these cones slightly earlier. However, in this case, we end up with the same number of leftovers on the same processor after the cone is finished, and hence this does not affect the properties of our construction.

\subsubsection{Embedding $3$-partition}

Note that to prove the properties of the construction above, it was only necessary to have an approximate value for the parameters $A$ and $B$; intuitively speaking, any values in the same magnitude would guarantee these properties. As such, the key idea of the reduction is to consider the main copy of our construction, and fine-tune the values $A$ and $B$ here to model our $3$-partition problem.

In particular, we will consider the $m$ triple cones in the main copy, and turn each of them into a representation of the number $a_i$ in the $3$-partition problem. For this, we replace the cone of size $A$ in the triple cone by a cone of size $(A+a_i)$, and we replace the cone of size $B$ by a cone of size $(B+2 \! \cdot \! T-a_i)$.

Furthermore, the cones attached to the critical path are also modified in the main copy. In particular, in each $3$-segment, the sizes of the cones attached to $v_i$, ..., $v_{(i+8)}$ were originally
\[ (M-A, \, M-B, \, M_0, \; M-A, \, M-B, \, M_0, \; M-A, \, M-B, \, M_0) \]
in order; we replace this by the new sequence
\begin{gather*}
(M-(A+T/2), \, M-(B+ 5 \cdot T), \, M_0, \; M-(A+T/2), \\ M-B, \, M_0, \; M-A, \, M-B, \, M_0) \, .
\end{gather*}
This creates at most $6 \! \cdot \! T$ further leftover computation steps within each such $3$-segment in the main copy. As discussed before, this does not affect the time when the communication steps happen.

Furthermore, note that the number of nodes in any triple cone sum up to $(A+B+2 \! \cdot \! T)$ (not considering the cone of size $M_0$), so three triple cones add up to $3 \cdot (A+B+2 \! \cdot \! T)$. On the other hand, in any $3$-segment, the number of leftover nodes from the attached cones of the path is now also
\[ 2 \cdot (A+T/2) + (B+5 \cdot T) + A + 2 \cdot B = 3 \cdot (A+B+2 \cdot T) \, . \]
This means that in order to finish three triple-cones by the end of the $3$-segment, i.e.\ in any valid schedule, there can be no leftover nodes at all by the end of the $3$-segment: every computation step of each processor is required.

Assume w.l.o.g. that the triple cone chosen in the first segment (of a given $3$-segment) corresponds to the number $a_1$. Since the triple cones of the main copy have cone sizes larger than $A$ and $B$, in the first segment these can only be computed with the corresponding leftovers groups of size $(A+T/2)$ and $(B+ 5 _{\!} \cdot _{\!} T)$ from the attached cones in the main copy. This implies that a valid scheduling of the first segment results in $(A+T/2)-(A+a_1)=T/2-a_1$ leftover nodes on one processor, and $(B+5_{\!} \cdot _{\!} T)-(B+ 2 _{\!} \cdot _{\!} T - a_1)=3 _{\!} \cdot _{\!} T+a_1$ leftover nodes on another processor after the segment. As discussed above, these must be used to pre-compute other nodes in the same $3$-segment; in particular, they need to be used to compute some of the ``$A$-type'' and ``$B$-type'' cones (i.e.\ the cones of size  $(A+a_i)$ and $(B+2 _{\!} \cdot _{\!} T-a_i)$) in the main copy, since the corresponding processors will not have enough computation steps for these cones in segments $2$ and $3$ by design.

Now consider the second segment, corresponding to some number $a_2$. Here the cone of size $(B+2 \! \cdot \! T-a_2)$ cannot be computed by any of the leftover groups in this segment (they all have size at most $B$), so we need to use the leftovers from the previous segment. Moreover, the $(T/2-a_1)$ leftovers from the previous $A$-type cone cannot be used for this cone, since these nodes would still not be enough to compute it in time, as
\[ (T/2-a_1)+B < B+2 \cdot T-a_2 \, . \]
Hence to compute the $B$-type cone in the second segment, the leftover nodes from the previous $B$-type cone need to be used, and hence this second $B$-type cone in the main copy needs to be assigned to the same processor. As for the cone of size $(A+a_2)$, this cannot be combined from a leftover group of size $A$ and the $(T/2-a_1)$ remaining leftovers from the previous segment, since $a_2 > T/2 - a_1$ due to $a_i > T/4$. Hence this is computed again with the group of size $(A+T/2)$ from the main copy. Then after the segment, this results in a group of
\[ B + (3\!\cdot\!T+a_1) - (B+2 \cdot T-a_2) = T + a_1 + a_2 \]
leftover nodes due to the $B$-type cone, and a group of
\[ (A+T/2)-(A+a_2) = T/2 - a_2 \]
leftover nodes due to the $A$-type cone, besides the further $(T/2-a_1)$ leftovers from the first segment. The $(T + a_1 + a_2)$ leftovers cannot be on the same processor as the other two groups, since the all the $A$-type and $B$-type cones are assigned to different processors in every segment. However, the $(T/2 - a_1)$ and $(T/2 - a_2)$ may be on the same processor, forming a common group of $(T - a_1 - a_2)$ leftover nodes; in fact, this is the only possibility that will allow us to finish the triple cone in the third segment.

In the third segment (where the triple cone corresponds to $a_3$), there is again a cone of size $(B+2 _{\!} \cdot _{\!} T-a_3)$. We again cannot use the at most $(T - a_1 - a_2)$ leftovers from the $A$-cones here, since even together with a group of size $B$, this is still only
\[ (T - a_1 - a_2)+B < B + T <  B+2 \cdot T-a_3  \]
nodes, which is once again not enough. Hence for the cone of size $(B+2 _{\!} \cdot _{\!} T-a_3)$, we again need to use the $(T + a_1 + a_2)$ leftovers from the previous $B$-cones. Furthermore, recall that all the leftovers are needed at this point, so the remaining $(T - a_1 - a_2)$ leftovers must be used on the cone of size $(A+a_3)$. In particular, to finish both the $A$-cone and the $B$-cone in time, we need to have
\begin{gather*}
(T - a_1 - a_2)+A \geq A + a_3 \\
(T + a_1 + a_2) + B \geq B+2 \cdot T-a_3 \, ,
\end{gather*}
which only hold together if we have exactly $a_1+a_2+a_3=T$.

This shows that we can only produce a valid schedule for each $3$-segment if we can partition the numbers into groups of $3$ that each sum up to $T$. On the other hand, given such a $3$-partitioning, the method above indeed produces a schedule of cost $C_0$; in particular, note that our method described in Section~\ref{app:subs:conebehave} uses the same processor for each $A$-cone and $B$-cone, respectively, so it can indeed use the leftover nodes in the way described above. This completes the reduction.

\subsubsection{Technical details}

It remains to discuss some technicalities.

Note that with the exception of final cones, it could still happen that a cone is not entirely assigned to the same processor. This does result in extra communication towards the top node, but this is not always prohibitive, because some processors are idle in given communication steps, and hence in some cases, they could execute these extra communications without increasing the cost. We know that any cone must have strictly less than $3\! \cdot \! m$ such \textit{outlier nodes} that are assigned to a different processor than the top, otherwise this already produces a cost of more than $C_0$. However, these few outlier nodes in the cones could still cause small problems in the construction: they may allow us to finish the cones (and hence possibly execute some communication steps) earlier. As a result, the communication step on the main path could happen several nodes before $v_i$, and the workload may be split differently between the processors assigned to the path and to the attached cone (producing leftover groups of slightly different size than discussed before).

To prevent this, we consider a further parameter $D$, and we multiply the size of each gadget in our construction by a factor $D$ in the end; in particular, the size of all cones, and the size of each part $[v_i, v_{(i+1)})$ of the critical paths. Note that there can be no leftovers at all by the end of a $3$-segment, and during each $3$-segment, we process only $O(1)$ different gadgets ($9$ parts of the path, $9$ attached cones, semi-final cones consisting of $15$ cones altogether, in $4$ copies; altogether $132$). Each of these gadgets can only increase/decrease the number of leftover nodes on any processor by at most $3_{\!} \cdot _{\!} m$, by e.g.\ bringing a communication step slightly earlier, or using some outliers instead of the processor assigned to a cone. Hence even if these would add up, the number of leftover nodes on any processor can increase/decrease by at most $132 _{\!} \cdot _{\!} 3 _{\!} \cdot _{\!} m$. As such, if we select $D= O(1) _{\!}  \cdot _{\!} m$ with a large enough constant (e.g. $3 _{\!} \cdot _{\!} 132 _{\!} \cdot _{\!} 3$), then these small offsets in group sizes do not allow us to combine the leftovers in our cones in any other way than described before; in particular, each side in our inequalities can change by at most $\frac{D}{3}$, while an original difference of $1$ now corresponds to $D$ nodes. As such, even if we have a few outliers in some of the cones, this does not influence the behavior of the construction.

As for our other parameters, note that the embedding of $3$-partition can only result in less than $6 \cdot T$ further leftover nodes (besides $A$'s and $B$'s) in a given $3$-segment. It would be enough to make $A$ larger than this; however, to emphasize the size difference between the different kinds of gadgets in our constructions, let us select $A = m \cdot T$ (assuming that $m$ is much larger than $6$). To ensure $B > 4 \cdot A + 6 _{\!} \cdot _{\!} T$, we choose $B = 8 _{\!} \cdot _{\!} A$. We then need to ensure $2 \cdot M_0 > M + (4 _{\!} \cdot _{\!} m+1) \cdot (B + 6 _{\!} \cdot _{\!} T)$; since we will have $M < M_0 + B$, this is satisfied if we ensure $M_0 > (4 _{\!} \cdot _{\!} m+2) _{\!} \cdot _{\!} (B + 6 _{\!} \cdot _{\!} T)$. Since $m _{\!} > _{\!} 2$ and $B _{\!} > _{\!} 6 _{\!} \cdot _{\!} T$, we can do this e.g.\ by selecting $M_0 = 10 _{\!} \cdot _{\!}  m _{\!} \cdot _{\!} B$. To have $M_0 + A + 6 _{\!} \cdot _{\!} T < M < M_0 + B$, we then choose $M=M_0 + A + 7 _{\!} \cdot _{\!} T$. Finally, recall that the size of all gadgets is further increased by a factor $D=O(m)$. However, even in this case, the size $n$ is still polynomial in $m$.

To ensure that $T/2$ is an integer, we also assume for convenience that $T$ is divisible by $6$; otherwise, we can simply multiply all the $a_i$ by $6$ to obtain an equivalent $3$-partition problem.

Finally, note that the proof applies not only to the DS model, but also the other BSP variants in Table~\ref{tab:commodels}. In particular, broadcasting offers no advantages in in-trees, and free data movement also does not affect any of the properties of our construction.

\section{APX-hardness} \label{app:APX}

To show that the problem does not allow a PTAS on general DAGs, we provide a reduction from the MAX-$3$SAT($B$) problem: given a boolean formula over $N$ variables and $M$ clauses in conjunctive normal form, with each clause containing exactly $3$ literals, and such that any variable appears in at most $B$ clauses for some constant bound $B$, the goal is to find a boolean assignment to the variables that maximizes the number of clauses satisfied. This problem has been studied exhaustively before \cite{maxsat1, maxsat2, maxsat3}, and previous work has shown that it is already APX-hard for very small $B$ values, including $B=4$ \cite{maxsat4}.

We also note that any MAX-$3$SAT($B$) problem has a solution where at least $\frac{7}{8} \cdot M$ clauses are satisfied, and such a solution can be found in polynomial time (by applying standard derandomization techniques to a simple randomized algorithm \cite{sat78_1, sat78_2}).

Finally, we point out that the constants in our construction are chosen rather large, in order to simplify the analysis; with some further effort, it is possible to further improve upon these constants significantly.

\subsection{Construction}

\subparagraph*{General tools.}

The dominant tool in our construction is a \textit{rigid chain gadget}: a chain of nodes $u_1, u_2, ..., u_{\ell}$, with $(u_i, u_j) \in E$ for all $i<j$, $i, j \in [\ell]$. The gadget is a simple way to ensure that any good schedule assigns all these nodes to a single processor. If the nodes $u_{1}, ..., u_{\ell}$ are not assigned all to the same processor in a schedule, we say that the rigid chain is \textit{broken}, and if even the nodes $u_{\ell/2+1}, ..., u_{\ell}$ are not all assigned to the same processor, we say that it is \textit{critically broken} (assuming $\ell$ is even for simplicity). If the chain is critically broken, i.e.\ both processors appear in the second half of the chain, then this means that each node $u_1, ..., u_{\ell/2}$ needs to be sent to the other processor, resulting in at least $\frac{\ell}{2}$ communication steps altogether.

Our construction will be split into $M_{\!}+_{\!}N_{\!}+_{\!}2$ gadgets, and these will be connected in a serial fashion. That is, each gadget will have a last layer of nodes $S_l$ (sinks within the gadget), and a first layer of nodes $S_1$ (sources within the gadget), and all nodes in the last layer $S_l$ of the $i$-th gadget will have an edge to all nodes in the first layer $S_1$ of the $(i+1)$-th gadget. Intuitively, this encourages that all nodes of the current gadget are computed (and if necessary, communicated to the other processor) before we start processing the next gadget.

\subparagraph*{Construction details.}

Our construction has parameters $\lambda_c, \lambda_v$, both constant integers, and another integer parameter $\lambda_s$ that is linear in $M$.

We begin the construction by creating two rigid chains (called the \emph{long rigid chains}) of length $\lambda_s + M _{\!} \cdot _{\!} \lambda_c + N _{\!} \cdot _{\!} \lambda_v + 1$ each. Furthermore, we take another rigid chain of length $M _{\!} \cdot _{\!} \lambda_c$ (called the \emph{side chain}), and artificially attach it to the middle of one of the long rigid chains. This side chain will mimic the behavior of the segment of the long rigid chain between the $(\lambda_s+1)$-th and $(\lambda_s+M _{\!} \cdot _{\!} \lambda_c)$-th node. That is, the $j$-th node of the side chain (with $j \in [M _{\!} \cdot _{\!} \lambda_c]$) will have an incoming edge form (i) the $i$-th node of the selected long rigid chain, for all $i \in [\lambda_s]$, and from (ii) the $i$-th node of the side chain itself, for all $i \in [j-1]$. Similarly, $j$-th node of the side chain will have an outgoing edge to (i) the $i$-th node of the selected long rigid chain for all $i \in [\lambda_s + M _{\!} \cdot _{\!} \lambda_c + N _{\!} \cdot _{\!} \lambda_v + 1] \setminus [\lambda_s + M _{\!} \cdot _{\!} \lambda_c]$, and to (ii) the $i$-th node of the side chain itself, for all $i \in [\lambda_c] \setminus [j]$. We will refer to the long rigid chain with the side chain attached as the \emph{false chain}, and the other long rigid chain as the \emph{true chain}.

The first gadget of our DAG will simply consist of the first $\lambda_s$ nodes of both long rigid chains; we call this the \emph{initialization gadget}. The last layer of this gadget consists of only $2$ nodes: the $\lambda_s$-th nodes of both long rigid chains.

This is then followed by $M$ distinct clause gadgets, each of which represent the next clause of the input boolean formula. Each clause gadget consists of:
\begin{itemize}[itemsep=1pt, topsep=3pt]
 \item the next $\lambda_c$ nodes of the true chain, the next $\lambda_c$ nodes of the false chain, and the next $\lambda_c$ nodes of the side chain,
 \item $3$ distinct rigid chains of length $\lambda_c$, each representing one of the literals in the clause,
 \item $2$ distinct rigid chains of length $\lambda_c$, called padding chains,
 \item $2$ extra nodes.
\end{itemize}
The first layer of the gadget consists of the first nodes of the given segments of the true, false, and side chains, the first nodes of the $5$ additional rigid chains, and the $2$ extra nodes. The last layer of the gadget consists of the last nodes in the given segments of the true, false, and side chains, the last nodes of the $5$ additional rigid chains, and again the $2$ extra nodes. As such, all first and last layers consist of $10$ nodes in these gadgets.

The DAG then continues with $N$ distinct variable gadgets, each representing the next variable of the boolean formula. Each of the variable gadgets consist of:
\begin{itemize}[itemsep=2pt, topsep=3pt]
 \item the next $\lambda_v$ nodes of the true chain, and the next $\lambda_v$ nodes of the false chain,
 \item $2$ distinct rigid chains of length $\lambda_v$, with one of them representing the variable itself, and the other one representing the negated version of the variable,
 \item $2$ extra nodes.
\end{itemize}
The first layer of these gadget consist of the first nodes of the segments of the true/false chains, the first nodes of the $2$ additional rigid chains, and the $2$ extra nodes. The last layers of the gadget consists of the last nodes in the segments of the true/false chains, the last nodes of the $2$ rigid chains, and again the $2$ extra nodes. All first and last layers consist of $6$ nodes in these gadgets.

Furthermore, to capture the structure of the formula, we take each rigid chain $R$ in a variable gadget that represents a literal (i.e.\ the original or negated version of a variable), and consider all the (at most $B$) rigid chains $R_i$ in the clause gadgets that represent the same literal. For all such rigid chains, we draw an edge from all nodes of $R_i$ to all nodes of $R$.

Finally, after the variable gadgets, we add a closing dummy gadget, consisting of only the last nodes of the two long rigid chains. This is simply to ensure that we do not need to analyze the last variable gadget differently.

As mentioned, between every pair of consecutive gadgets, we draw an edge from all nodes in the last layer of the preceding gadget to all nodes in the first layer of the following gadget.

Altogether, our DAG consists of $n=2 \cdot \lambda_s + (8 \cdot \lambda_c + 2) \cdot M + (4 \cdot \lambda_v + 2) \cdot N + 2$ nodes. With $\lambda_s \in O(M)$, $\lambda_c, \lambda_v \in O(1)$ and $N \in O(M)$, the number of nodes is in $O(M)$.

\subparagraph*{Choice of our parameters.}

We analyze the BSP scheduling of this DAG on $P=2$ processors, with a choice of $g=8$ and $L=0$. The constant $B$ is already defined in the input satisfiability problem. Note that due to $P=2$, our reduction holds equivalently in all communication models within BSP.

As for the parameters in the construction, we select $\lambda_c=3$, and $\lambda_v = B \cdot \lambda_c +1$ to ensure $\lambda_v > B \cdot \lambda_c$. We discuss the choice of $\lambda_s$ later; for now, assume that $\lambda_s \in \Theta(M)$.

\subparagraph*{Intuition and clean schedules.}

The intuition behind the construction is as follows. Any reasonable schedule will assign all nodes in the true chain to a single processor (assume w.l.o.g. $p_1$), and all nodes in the false chain and its side chain to the other processor; otherwise, either the work or the communication costs become unreasonably high. Furthermore, none of the smaller rigid chains within the gadgets will be broken in a good solution either. Each gadget will be computed in a single superstep in a reasonable solution, first computing all nodes of the gadget in a computation phase, and then sending all values of the last layer of the gadget to the other processor in a single communication phase.

In the variable gadgets, the two rigid chains will be assigned to different processors, to avoid a too high work cost. We understand the literal assigned to $p_1$ to represent a true value, and the literal assigned to $p_2$ to represent a false value. Furthermore, when $R$ is a rigid chain representing a literal in a variable gadget, all the rigid chains $R_i$ representing the same literal in a clause gadget should be assigned to the same processor, to avoid communications. Finally, the two extra nodes should be assigned to different processors (arbitrarily).

In the clause gadgets, we have $8 _{\!} \cdot _{\!} \lambda_c + 2$ nodes altogether; out of these, $2 _{\!} \cdot _{\!} \lambda_c$ are always assigned to $p_2$ (false chain and side chain) and $\lambda_c$ are always assigned to $p_1$ (true chain). Then, if all $3$ literals in a clause are false, then we prefer to assign both of the padding chains and both of the extra nodes to $p_1$, but even in this case, the superstep of the gadget will have a work cost of $5 _{\!} \cdot _{\!} \lambda_c$. On the other hand, if $1$, $2$, or $3$ of the literals in the clause are true, then we can assign $2$, $1$ or $0$ of the padding chains to $p_1$, respectively (and the others to $p_2$), and assign the two extra nodes to different processors. This ensures that the work cost of the superstep is the minimal possible value of $4 _{\!} \cdot _{\!} \lambda_c + 1$.

If a schedule satisfies the properties above, we call it a \emph{clean schedule}. Note that in these clean schedules, the nodes $S_l$ in the last layer of the gadgets are always balanced between the two processors, and sent in a single communication phase, thus incurring a communication cost of $\frac{|S_l|}{2} _{\!} \cdot _{\!} g$. As such, over the whole construction, there are $\Upsilon:=2 + 10 _{\!} \cdot _{\!} M + 6 _{\!} \cdot _{\!} N$ communication steps, adding up to a communication cost of exactly $\frac{\Upsilon}{2} \cdot g$.

The work cost in the initialization gadget is $\lambda_s$, the work cost in each variable gadget is exactly $2 \cdot \lambda_v +1$, and the work cost of the dummy gadget at the end is $1$. However, the work cost of a clause gadget is not fixed: if the clause is unsatisfied, the work cost is $5 _{\!} \cdot _{\!} \lambda_c$, and if it is satisfied, then the work cost is $4 _{\!} \cdot _{\!} \lambda_c + 1$, i.e.\ it is lower by $(\lambda_c-1)$. Let us denote $\lambda'_c := (\lambda_c-1)$ for convenience. Let us denote the fixed part of this total cost by 
\[ M_0 := \lambda_s + (2 _{\!} \cdot _{\!} \lambda_v + 1) _{\!} \cdot _{\!} N + 5 _{\!} \cdot _{\!} \lambda_{c\!} \cdot _{\!} M + 1
 + \frac{\Upsilon}{2} _{\!} \cdot _{\!} g \, ; \]
recall that this is linear in $M$, i.e.\ there is a $\lambda \in O(1)$ such that $M_0 \leq \lambda \cdot M$. This means that if \textsc{sat} denotes the number of satisfied clauses in an assignment of the variables, then the corresponding clean schedule in the DAG has a total cost of exactly $M_0- \lambda'_c \cdot \textsc{sat}$. This also shows that any clean schedule has a total cost of at most $M_0$.

The technical part is to actually prove that any reasonable BSP schedule in this DAG is a clean schedule outlined above. More specifically, we show that if any schedule does not satisfy one of these properties, then we can transform it in a series of steps to obtain a clean schedule of at most the same cost.

\subsection{First observations}

We begin with some simple properties on the assignment of the long rigid chains, and the behavior of the schedule between consecutive gadgets.

\subparagraph*{The long rigid chains.}

We begin by showing that it is always suboptimal to not assign the long rigid chains to a single processor.

For this, we first discuss the choice of $\lambda_s$. Let us briefly introduce the notation $\textsc{sol} = M_0 - \lambda_s$, i.e.\ the maximal cost of a clean solution, as discussed above, minus $\lambda_s$. Also, let $\textsc{len} =  M _{\!} \cdot _{\!} \lambda_c + N _{\!} \cdot _{\!} \lambda_v + 1$, i.e.\ the length of the long rigid chains, minus $\lambda_s$. Recall that $\textsc{sol}, \textsc{len} \in O(M)$. We will need to ensure that $\textsc{sol} + \lambda_s < 2\cdot (\textsc{len} + \lambda_s)$, or in other words, $\textsc{sol} - 2 _{\!} \cdot _{\!} \textsc{len} < \lambda_s$. With $\textsc{sol} \in O(M)$ and $\textsc{len} > 0$, we can easily do this with a choice of $\lambda_s = \textsc{sol}$. With $\lambda_s \in O(M)$, we still have $M_0 \in O(M)$ for the fixed part of the cost of clean solutions. Note that $\textsc{sol} + \lambda_s < \frac{g}{4}\cdot (\textsc{len} + \lambda_s)$ also follows from the previous claim since $g=8$. We also get $\lambda_s > \textsc{len}$, since $\textsc{len} < \textsc{sol}$.

\renewcommand*{\proofname}{Proof.}

\begin{lemma}
Given a schedule $(\pi, \tau, \Gamma)$ of cost $C$, we can convert it into a schedule $(\pi', \tau', \Gamma')$ of cost $C'$ such that $C' \leq C$, and such that in $(\pi', \tau', \Gamma')$ the true chain is entirely assigned to one processor, and the false chain (and its side chain) is entirely assigned to the other.
\end{lemma}

\begin{proof}
Recall that the two long rigid chains have length $(\textsc{len}+\lambda_s)$; if either of these is critically broken, then this results in $\frac{1}{2} _{\!} \cdot _{\!} (\textsc{len}+\lambda_s)$ communication steps, and thus a total communication cost of $\frac{1}{4} _{\!} \cdot _{\!} (\textsc{len}+\lambda_s) _{\!} \cdot _{\!} g$ at least. Any clean schedule has a cost of at most $\textsc{sol}+\lambda_s$, and $\lambda_s$ was chosen to ensure $\frac{g}{4} _{\!} \cdot _{\!} (\textsc{len}+\lambda_s) > \textsc{sol}+\lambda_s$, so we could strictly reduce the cost by switching to any clean schedule. As such, the second half of each long rigid chain must be assigned entirely to a single processor. Note that $\lambda_s > \textsc{len}$, so this second half contains the segments in all the clause and variable gadgets.

We next show that the long rigid chains are not broken, and that the two chains should be assigned to different processors. Let $s$ be the first superstep that computes a node of the first clause gadget. First, assume that the second halves of the two chains are assigned to different processors. In this case, the cost of the schedule before superstep $s$ has to be at least $\lambda_s+g$: the chain segments in the initialization gadget induce a work cost of at least $\lambda_s$, and since the last layer of the initialization gadget is split between processors, at least one communication step is also required. As such, we can simply reassign all nodes in the long rigid chains to the processor owning the second half, and use a single superstep to compute the initialization gadget and communicate the two nodes in the last layer, at a cost of exactly $\lambda_s+g$. This cannot increase the cost, and ensures that the part of the schedule starting from superstep $s$ always remains valid.

On the other hand, assume the second halves of the two long rigid chains are assigned to the same processor $p_i$. Let $x$ be the number of nodes in the first halves of the two chains that are not assigned to $p_i$. These nodes together incur a communication cost of at least $\frac{x}{2} _{\!} \cdot _{\!} g$ before superstep $s$, since they are all sent to $p_i$. If we instead reassign all these nodes to $p_i$, this increases the total work cost by at most $x$, while removing these communication steps; thus it decreases the total cost if $\frac{x}{2} _{\!} \cdot _{\!} g > x$, i.e.\ $g>2$. At this point, both long rigid chains are assigned entirely to $p_i$. However, this then yields a work cost of at least $2 _{\!} \cdot _{\!} (\textsc{len} + \lambda_s)$, whereas any clean solution has a total cost of at most $\textsc{sol} + \lambda_s$.

As such, we can obtain a transformed schedule where neither of the long rigid chains is broken, and they are assigned to different processors. Assume w.l.o.g.\ that the true chain is assigned to $p_1$, and the false chain to $p_2$.

The side chain attached to the false chain is in the second half of the false chain, and thus also needs to be assigned entirely to $p_2$: if any node in the side chain is assigned to $p_1$, then as before, this results in $\frac{1}{2} _{\!} \cdot _{\!} (\textsc{len}+\lambda_s)$ communication steps from the first half of the false chain (and a total cost of at least $\frac{g}{4} _{\!} \cdot _{\!} (\textsc{len}+\lambda_s)$). Thus we can also assume that the side chain is assigned entirely to $p_2$.
\end{proof}

Note that this assignment of the long rigid chains already results in some unavoidable communication steps. Consider the connections between any two consecutive gadgets, with the layers $S_l$ and $S_0$. The assignments of the long rigid chains implies that in each such $S_l$ and $S_0$, there is at least one node assigned to $p_1$ and at least one node assigned to $p_2$. Hence in any schedule, all the nodes in $S_l$ need to be available to both processors, i.e.\ the always need to be sent to the other processor. This implies that over the whole construction, there are $\Upsilon=2 + 10 _{\!} \cdot _{\!} M + 6 _{\!} \cdot _{\!} N$ unavoidable communication steps, adding up to a communication cost of at least $\frac{\Upsilon}{2}_{\!} \cdot _{\!} g$ (the total communication cost in a clean solution).

\subparagraph*{Subschedules on gadgets.}

Furthermore, we show that the computation phases in the schedule corresponding to the gadgets can be nicely separated into disjoint segments, i.e.\ such that each computation phase computes nodes from only a single gadget. Intuitively, for the $i$-th gadget, there will be a superstep index $s^{(i)}$ (possibly after some artificial adjustments) such that all nodes of the $i$-th gadget are computed in superstep $s^{(i)}$ at the latest, and all nodes of the $(i+1)$-th gadget are computed in superstep $s^{(i)}_{\!}+_{\!}1$ at the earliest.

Consider the $i$-th gadget in our construction. Let $S_l$ be the last layer in this gadget, and $S_1$ be the first layer of the $(i+1)$-th gadget. Let $s_1$ be the last superstep when a node from the $i$-th gadget is computed on $p_1$, and let $s_2$ be the last superstep when a node from the $i$-th gadget is computed on $p_2$. Note that the long rigid chains are assigned to different processors, so these numbers are well-defined.

First assume that $s_1=s_2$. Consider the nodes computed on processor $p_j$ in this superstep, and take the last one of these according to some topological ordering; this has to be a node of $S_l$, otherwise the node has out-neighbors in the $i$-th gadget on the other processor that can only be computed after superstep $s_1$. This means that both $p_1$ and $p_2$ compute at least one node of $S_l$ in this superstep. This means that all nodes of $S_1$ require at least one value from the other processor which can be sent in superstep $s_1$ at the earliest; hence they cannot be computed before superstep $s_1+1$. Since other nodes in the $(i+1)$-th gadget are all descendants of a node in $S_1$, we can only start computing any node from the $(i+1)$-th gadget from superstep $s_1+1$. Hence we can set $s^{(i)}=s_1$.

Now assume $s_1 \neq s_2$, and let w.l.o.g.\ $s_1<s_2$. Consider the nodes computed on processor $p_2$ in superstep $s_2$, and consider the last of these from the $i$-th gadget according to some topological ordering; this has to be a node of $S_l$, otherwise there are nodes in the $i$-th gadget that are not computed at all. As such, $p_2$ computes at least one node of $S_l$ in superstep $s_2$. Furthermore, let $s'$ be the first superstep where $p_1$ computes a node from the $(i+1)$-th gadget, and let $v$ be the first node of the $(i+1)$-th gadget computed by $p_1$ in superstep $s'$, according to some topological ordering. If $s' \in S_1$, then $s'>s_2$, since a predecessor from $S_l$ computed on $p_2$ in $s_2$ first needs to be communicated to $p_1$. If $s' \notin S_1$, then $v$ has an in-neighbor $v'$ computed on $p_2$ within the $(i+1)$-th gadget. Since $v'$ is a descendant of all nodes in $S_l$, and one of these nodes is computed in $s_2$, $v'$ (and in general, any node of the $(i+1)$-th gadget) is computed in $s_2$ at earliest. Since a communication is required to send $v'$ to $p_1$ before computing $v$, we again have $s'>s_2$. This implies that processor $p_1$ does not compute anything in supersteps $s_1\!+\!1, ..., s_2$, and $p_2$ also starts computing the nodes of the $(i+1)$-th gadget in superstep $s_2$ at the earliest.

We consider two cases. If all nodes computed on $p_2$ in superstep $s_2$ are from the $i$-th gadget, then we can naturally set $s^{(i)}=s_2$, since all nodes from the $(i+1)$-th gadget are computed after $s_2$. Otherwise, if $p_2$ computes nodes from both the $i$-th and $(i+1)$-th gadget in superstep $s_2$, then we can artificially split up superstep $s_2$ into two `imaginary' computation phases $s_{2,1}$ and $s_{2,2}$, with $s_{2,1}$ containing the nodes from the $i$-th, and $s_{2,2}$ containing the nodes from the $(i+1)$-th gadget. Since no computation happens on $p_1$ in $s_2$, the total work cost of these computation phases is the same as the work cost of $s_2$; as such, in terms of work costs, we can simply analyze this superstep $s_2$ as if it had two computation phases. We can then set $s^{(i)}=s_{2,1}$ as before.

The discussion above shows that the sum of work costs over the new (potentially split) computation phases is identical to the total original work cost. This essentially allows us to analyze the work cost separately for each gadget. More specifically, it ensures the following property: if the $i$-th gadget has at least $W_i$ nodes assigned to some processor, then the entire schedule has a total work cost of at least $\sum W_i$, with the sum understood over all gadgets.

\subsection{Converting to a clean schedule}

We will modify the schedule in two phases, each consisting of several smaller steps. In the first phase, we transform the original schedule into a \emph{semi-clean} schedule, which fulfills almost all properties of a clean schedule, but allows the two rigid chains in variable gadgets to be assigned to the same processor. In this first phase, each of our steps will remove at least $\Delta t_c$ communication steps from $\Gamma$, and on the other hand, move at most $\Delta t_w$ nodes to another processor (thus intuitively, increasing work costs by at most $\Delta t_w$). We show in the end that these steps do not increase the cost altogether. We point out several properties of this technique:
\begin{itemize}[itemsep=3pt, topsep=3pt]
 \item It is a key aspect that the steps are analyzed together in the end, and not one by one: when we remove some communication steps, the total cost might not decrease immediately, since the corresponding communication phases might be dominated by the values sent in the other direction.
 \item Each of our steps remove elements from $\Gamma$ that correspond to communicating different nodes, so these are indeed distinct communication steps.
 \item Nodes in the last layers of gadgets are always communicated to the other processor, so these nodes never appear in any of our steps of removing $\Delta t_c$ communications.
\end{itemize}
We also note that the internal states during these steps are only for the analysis, and might not correspond to valid BSP schedules, but the final semi-clean schedule is valid.

Then in the second phase, we transform our semi-clean schedule into a clean schedule. This phase does not affect communication at all. Here each of our transformation steps create another valid BSP schedule that strictly increases the total work cost compared to the previous schedule.

\subparagraph*{Fixing the broken chains.}

We begin by considering any broken rigid chain in the variable gadgets. For any such chain, assume that the last node of the chain is assigned to processor $p_j$; we then reassign all other nodes in the chain to $p_j$, too. If there were originally $x$ nodes in the chain assigned to the other processor, then all of these had to be sent to $p_j$ before; as such, the transformation removes $\Delta t_c=x$ communication steps, and moves at most $\Delta t_w = x$ nodes to the other processor. Note that any nodes in the clause gadgets with an edge towards this chain also have an edge to the last node of the chain, which was already assigned to $p_j$, so the transformation does not make any new communication step necessary.

We execute the same thing with all the padding rigid chains in the clause gadgets.

We then move to the rigid chains that represent literals in the clause gadgets. Let $R_1$ be such a rigid chain, and let $R$ be the rigid chain that represents the same literal in a variable gadget. Note that $R$ is already not broken; let $p_j$ denote the processor that $R$ is assigned to. If the last node of $R_1$ is also assigned to $p_j$, then we proceed as before, reassigning all nodes in $R_1$ from the other processor to $p_j$, while removing $\Delta t_c=x$ communication steps and moving at most $\Delta t_w = x$ nodes to the other processor. On the other hand, if the last node in $R_1$ is not assigned to $p_j$, then we reassign every node in $R_1$ to $p_j$. Note that the first $(\lambda_c-1)$ nodes in $R_1$ had an edge both to a node computed on $p_j$ (last node of $R$) and a to node computed on the other processor (last node of $R_1$), so the reassignment removes $\Delta t_c=(\lambda_c-1)$ communication steps at least, and moves at most $\Delta t_w=\lambda_c$ nodes.

Note that for all of our modification steps above, it holds that $\frac{g}{2} \cdot \Delta t_c \geq \Delta t_w$. In particular, when $\Delta t_c = x$ and $\Delta t_w = x$, this holds since $g \geq 2$. When $\Delta t_c = \lambda_c-1$ and $\Delta t_w = \lambda_c$ in the clause gadgets, it holds since $g \geq 3$ and $\lambda_c \geq 3$. This implies that for the total number $\Delta T_c = \sum \Delta t_c$ of communication steps removed, and for the total number $\Delta T_w = \sum \Delta t_w$ of nodes moved to the other processor, we have that $\Delta T_c \geq \frac{2}{g} _{\!} \cdot _{\!} \Delta T_w$.

\subparagraph*{To a semi-clean schedule.} We now consider the gadgets one after another. For each gadget, we place all nodes of the gadget into a single computation phase, and send all values from its last layer in a single communication phase, creating our semi-clean solution. We analyze the resulting work cost for each gadget. In particular,
\begin{itemize}[itemsep=2pt, topsep=3pt]
 \item let $t_w^{(i)}$ denote the original work cost of the computation phases in the $i$-th gadget,
 \item let $t_w^{(i)'}$ denote the work cost in the (single) computation phase of the $i$-th gadget in our resulting semi-clean schedule,
 \item let $\Delta t_w^{(i)}$ denote the number of nodes that were moved to the other processor in the previous part, and are within the $i$-th gadget.
\end{itemize}
Clearly $\sum \Delta t_w^{(i)} = \Delta T_w$ over all gadgets, and note that no nodes were moved in the first and last gadgets, since these only contain the long rigid chain segments. More importantly, if we now have $w_0$ nodes assigned to some processor $p_j$ in the $i$-th gadget, then $t_w^{(i)} \geq w_0 - \Delta t_w^{(i)}$, since at least $w_0 - \Delta t_w^{(i)}$ nodes must have been assigned to $p_j$ originally, incurring a work cost of $w_0 - \Delta t_w^{(i)}$ in the computation phases of the gadget.

First consider the clause gadgets; recall that each one has $2 _{\!} \cdot _{\!} \lambda_c$ nodes assigned to $p_2$ (false chain and side chain), $\lambda_c$ nodes assigned to $p_1$ (true chain), $3$ rigid chains of size $\lambda_c$ and two padding chains of size $\lambda_c$ assigned to some unspecified processor, and two further nodes.
\begin{itemize}[itemsep=3pt, topsep=4pt]
 \item If all $3$ rigid chains representing literals are assigned to $p_2$, then there are currently at least $5 _{\!} \cdot _{\!} \lambda_c$ nodes assigned to $p_2$ in the gadget, so we had $t_w^{(i)} \geq 5 _{\!} \cdot _{\!} \lambda_c - \Delta t_w^{(i)}$. In this case, we reassign the two padding chains and the two extra nodes all to $p_1$, if they were not assigned to $p_1$ before. The resulting computation phase has a cost of exactly $t_w^{(i)'} = 5 _{\!} \cdot _{\!} \lambda_c$. Thus this ensures $t_w^{(i)} \geq t_w^{(i)'} - \Delta t_w^{(i)}$.
 \item If $x \geq 1$ of the rigid chains for literals are assigned to $p_1$, then we reassign $(x-1)$ of the padding chains to $p_2$, and the remaining $2-(x-1)$ to $p_1$, and reassign the two extra nodes to separate processors. With $4$ rigid chains (or chain segments) and an extra node on both $p_1$ and $p_2$, the resulting computation phase has $t_w^{(i)'} = 4 _{\!} \cdot _{\!} \lambda_c + 1$. Note that with $8 _{\!} \cdot _{\!} \lambda_c + 2$ nodes in the gadget, this is the minimal work cost for its computation phases, and hence $t_w^{(i)} \geq t_w^{(i)'}$; however, for simplicity, we again use the weaker form $t_w^{(i)} \geq t_w^{(i)'} - \Delta t_w^{(i)}$.
\end{itemize}
In both cases, our transformation also ensures that the $10$ nodes of the last layer are split equally among the processors, so we can communicate all values in a single communication phase at a cost of $5_{\!} \cdot _{\!} g$ in the end.

Now consider the variable gadgets: each one has $\lambda_v$ nodes assigned to $p_2$ (false chain), $\lambda_v$ nodes assigned to $p_1$ (true chain), $2$ rigid chains of size $\lambda_v$ assigned to some unspecified processor, and two further nodes.
\begin{itemize}[itemsep=3pt, topsep=4pt]
 \item If the $2$ rigid chains are assigned to the same processor $p_j$, then there are currently at least $3 _{\!} \cdot _{\!} \lambda_v$ nodes assigned to $p_j$ in the gadget, so we had $t_w^{(i)} \geq 3 _{\!} \cdot _{\!} \lambda_v - \Delta t_w^{(i)}$. In this case, we reassign the two extra nodes to the other processor than $p_j$. The resulting computation phase has a cost of exactly $t_w^{(i)'} = 3 _{\!} \cdot _{\!} \lambda_v$, ensuring $t_w^{(i)} \geq t_w^{(i)'} - \Delta t_w^{(i)}$.
 \item If the $2$ rigid chains are assigned to different processors, then we reassign the two extra nodes to separate processors. The resulting computation phase has $t_w^{(i)'} = 2 _{\!} \cdot _{\!} \lambda_v + 1$. With $4 _{\!} \cdot _{\!} \lambda_c + 2$ nodes in the gadget, this is the minimal cost for its computation phases, so $t_w^{(i)} \geq t_w^{(i)'}$. For simplicity, we again only use $t_w^{(i)} \geq t_w^{(i)'} - \Delta t_w^{(i)}$.
\end{itemize}
Again in both cases, the $6$ nodes of the last layer are split equally among the processors in our schedule, so we can have a single communication phase of cost $3_{\!} \cdot _{\!} g$ in the end.

Note that after we processed all gadgets, we arrive at a valid and semi-clean schedule, where the only difference from a clean schedule is that the two rigid chains in variable gadgets might be assigned to the same processor. Furthermore, if the original schedule had a total work cost of $T_w = \sum t_w^{(i)}$, and the resulting schedule has a total work cost of $T'_w = \sum t_w^{(i)'}$ (with the sum understood over all gadgets), then our observations in each gadget ensure that we have $T_w \geq T'_w - \Delta T_w$ altogether. Moreover, since the resulting schedule is semi-clean, and has a communication cost of $\frac{\Upsilon}{2}_{\!} \cdot _{\!} g$ exactly, its total cost is $T'_w + \frac{\Upsilon}{2}_{\!} \cdot _{\!} g$. In contrast to this, the original schedule had at least $\Upsilon$ communication steps for the nodes in the last layers of the gadgets, and the $\Delta T_c$ communication steps removed, adding up to a communication cost of at least $\frac{1}{2} _{\!} \cdot _{\!} g _{\!} \cdot _{\!} (\Upsilon + \Delta T_c)$. Using $\Delta T_c \geq \frac{2}{g} _{\!} \cdot _{\!} \Delta T_w$, the cost of the old schedule was at least
\[ T_w + \frac{1}{2} _{\!} \cdot _{\!} g _{\!} \cdot _{\!} (\Upsilon + \Delta T_c) \geq T_w + \Delta T_w + \frac{\Upsilon}{2} _{\!} \cdot _{\!} g \geq T'_w + \frac{\Upsilon}{2}_{\!} \cdot _{\!} g \, . \]
Thus our semi-clean schedule indeed has at most as high cost as the original schedule.

\subparagraph*{To a clean schedule.}

From the BSP schedule above, it only remains to ensure that in each variable gadget, the rigid chains are assigned to different processors. As such, consider any variable gadget where this does not hold, and both rigid chains are assigned to $p_j$. We then arbitrarily select one of the rigid chains $R$, and assign it to the other processor; at the same time, we pick one of the extra nodes in the gadget, and reassign it from the other processor to $p_j$. Before this step, the total work cost in the single supersteps of the variable gadget was $3 _{\!} \cdot _{\!} \lambda_v$; after the transformation, the work cost is only $2 _{\!} \cdot _{\!} \lambda_v + 1$. This means that the work cost in the gadget is reduced by $(\lambda_v -1)$. Note that the reassignment of the extra node ensures that the last layer of the gadget still has $3$ nodes on both processors, so the cost $3 _{\!} \cdot _{\!} g$ of the communication phase at the end remains unchanged.

Simultaneously to each such transformation, we consider all rigid chains $R_i$ in the clause gadgets that are connected to $R$, and also reassign all such $R_i$ to the other processor. Besides this, we also modify the padding chains/extra nodes in the clause gadget of $R_i$ to ensure that it corresponds to a clean schedule, and to maintain the previous work cost. More specifically, consider the $3$ rigid chains representing literals in the clause gadget of $R_i$.
\begin{enumerate}[topsep=2pt]
 \item If all $3$ of the rigid chains were assigned to $p_1$, then both padding chains were assigned to $p_2$, and the extra nodes to different processors. In this case, besides reassigning $R_i$ to $p_2$, we reassign one of the padding chains to $p_1$. 
 \item If $2$ of the rigid chains were assigned to $p_1$, then the two padding chains were assigned to different processors, as well as the two extra nodes.
\begin{enumerate}[topsep=1pt]
  \item If we are reassigning $R_i$ to $p_1$, we also reassign a padding chain from $p_1$ to $p_2$.
  \item If we are reassigning $R_i$ to $p_2$, we also reassign a padding chain from $p_2$ to $p_1$.
\end{enumerate}
 \item If $1$ of the rigid chains was assigned to $p_1$, then both padding chains were assigned to $p_1$, and the extra nodes to different processors.
 \begin{enumerate}[topsep=1pt]
  \item If we are reassigning $R_i$ to $p_1$, we also reassign a padding chain from $p_1$ to $p_2$.
  \item If we are reassigning $R_i$ to $p_2$, we also reassign an extra node from $p_2$ to $p_1$.
\end{enumerate}
 \item If none of the rigid chains were assigned to $p_1$, then the two padding chains were assigned to $p_1$, as well as the two extra nodes. In this case, besides reassigning $R_i$ to $p_1$, we reassign an extra node to $p_2$.
\end{enumerate}
Cases (1), (2a), (2b) and (3a) leave the work cost in the clause gadget unchanged. Case (4) actually reduces the work cost in the clause gadget by $(\lambda_c-1)$. However, case (3b) might actually increase the work cost in the gadget by $(\lambda_c-1)$. Note that all cases ensure that we still have $5$ nodes on both processors in the last layer $S_l$ afterwards, and hence the cost of the communication phase remains $5 _{\!} \cdot _{\!} g$.

Altogether, reassigning a rigid chain $R$ in a variable gadget decreases the work cost in the variable gadget by $(\lambda_v -1)$, and in the clause gadgets, it increases the work cost by at most $B \cdot (\lambda_c-1)$ altogether, since any variable appears in at most $B$ clauses. We have ensured $\lambda_v > B _{\!} \cdot _{\!} \lambda_c$, hence $(\lambda_v-1) > B _{\!} \cdot _{\!} (\lambda_c-1)$; this implies that the reassignment of every chain $R$ in the variable gadget decreases the total cost of the schedule. As such, we can execute these one after another, making sure that the two rigid chains are assigned to different processors in all variable gadgets. Our adjustments in the corresponding clause gadgets make sure that at the end of this process, we obtain a clean schedule that describes a valid solution of $3$SAT. The cost of this clean schedule is at most as much as that of the original schedule.

\subsection{Reduction}

The conversion in the previous sections implies two things: (i) we can assume that the optimal schedule is a clean schedule, (ii) we can assume that any approximation algorithm also returns a clean schedule, since otherwise, we can simply further reduce its cost and obtain a clean schedule (thus considering an improved version of the same algorithm).

This construction and conversion provides a natural L-reduction between the two problems. The optimum value in the original MAX-$3$SAT($B$) problem is at least $\frac{7}{8} _{\!} \cdot _{\!} M$, while the optimum cost in the derived BSP scheduling problem is at most $M_0 \leq \lambda _{\!} \cdot _{\!} M$, so the two optima are indeed within a constant factor of $\frac{7}{8} _{\!} \cdot _{\!} \lambda$. Furthermore, if $\textsc{sat}_{\textsc{opt}}$ denotes the number of satisfied clauses in the optimal MAX-$3$SAT($B$) solution, then the transformation shows that the optimal cost of BSP scheduling our DAG is $C_{\textsc{opt}}=(M_0- \lambda'_c \cdot \textsc{sat}_{\textsc{opt}})$. Assume that an algorithm returns a BSP schedule of cost $C_{\textsc{sol}}$. As shown before, we can transform this into a clean schedule without increasing its cost, and this clean schedule corresponds to a valid solution of MAX-$3$SAT($B$): if this variable assignment satisfies $\textsc{sat}_{\textsc{sol}}$ clauses, then the cost of the clean BSP schedule is exactly $(M_0-\lambda'_c \cdot \textsc{sat}_{\textsc{sol}}) \leq C_{\textsc{sol}}$. As such, altogether we have
\[ C_{\textsc{sol}} - C_{\textsc{opt}} \geq (M_0-\lambda'_{c\!} \cdot _{\!} \textsc{sat}_{\textsc{sol}}) - (M_0- \lambda'_{c\!} \cdot _{\!} \textsc{sat}_{\textsc{opt}}) = \lambda'_{c\!} \cdot _{\!}(\textsc{sat}_{\textsc{opt}\!} - \textsc{sat}_{\textsc{sol}}) \, . \]
This completes the L-reduction. Since MAX-$3$SAT($B$) is known to be APX-hard, this implies that BSP scheduling is also APX-hard, and hence does not allow a PTAS. We note that since $P _{\!} \in _{\!} O(1)$, and assigning all nodes to the same processor and superstep is a $P$-approximation, the problem is also within APX, and hence it is APX-complete.

\section{Communication scheduling} \label{app:CS}

We now present our proofs on the hardness of the communication scheduling problem.

As a side note, we point out that a slightly similar problem to CS also is studied in the work of \cite{recent3}, introduced as an intermediate scheduling task (named UMPS) in their reduction. However, in UMPS, only the assignment $\pi$ fixed, the assignment to time steps is not; furthermore, UMPS is studied in an entirely different setting to ours, with no communication costs at all and non-constant $P$.

\subsection{CS with 2 processors} \label{app:subs:greedy}

We begin by showing that with $P=2$ processors, there is a relatively simple algorithm to find the optimal communication schedule in polynomial time.

\renewcommand*{\proofname}{Proof of Lemma \ref{th:CSDP}}

\begin{proof}
Consider the following greedy algorithm: we iterate through the supersteps $1$, ..., $S$ in order. In the current superstep $s$, let $\Lambda_{1,2}$ denote the set of values $u$ that satisfy the following properties: (i) they were already computed on $p_1$ before ($\pi(u)=p_1$, $\tau(u) \leq s$), (ii) they are required later on $p_2$ (there is $(u,v) \in E$ with $\pi(v)=p_2$ and $\tau(v)>s$), and (iii) they have not been sent from $p_1$ to $p_2$ so far in our algorithm. Furthermore, assume that this set is organized as a list of values, sorted in increasing order according to the first time they are needed on $p_2$. That is, if for two values $u_1, u_2 \in \Lambda_{1,2}$ it is the successors $v_1$, $v_2$ on $p_2$ that have a minimal $\tau(v_1)$ and $\tau(v_2)$, respectively, with $\tau(v_1) < \tau(v_2)$, then $u_1$ will appear earlier in the list $\Lambda_{1,2}$ than $u_2$. Finally, let $\Lambda'_{1,2} \subseteq \Lambda_{1,2}$ denote the subset of this list that is immediately needed in superstep $(s+1)$, i.e.\ there is a $(u,v) \in E$ with $\pi(v)=p_2$ and $\tau(v)=s+1$; note that these values appear at the beginning of $\Lambda_{1,2}$. Symmetrically, let $\Lambda_{2,1}$ and $\Lambda'_{2,1}$ denote the same values in the other direction.

Let us assume w.l.o.g. that $|\Lambda'_{1,2}| \geq |\Lambda'_{2,1}|$. Our greedy algorithm will send exactly $C_{comm}\,\!^{(s)} = |\Lambda'_{1,2}|$ values from $p_1$ to $p_2$ in superstep $s$: the values contained in $\Lambda'_{1,2}$. In the other direction, from $p_2$ to $p_1$, we will send the first $C_{comm}\,\!^{(s)}$ of values of $\Lambda_{2,1}$ according to its ordering (but of course at most $|\Lambda_{2,1}|$ values if $|\Lambda_{2,1}|<C_{comm}\,\!^{(s)}$). Note that this indeed results in a communication cost of $C_{comm}\,\!^{(s)}$ for the superstep, and all the values required in the computation phase of superstep $(s+1)$ are indeed transmitted.

The algorithm clearly runs in polynomial time. The number of supersteps is at most $n$, and in each superstep, we can update the lists $\Lambda_{1,2}$ and $\Lambda_{2,1}$ in at most $O(n)$ time.

It remains to show that this algorithm indeed finds the optimal communication schedule. In any concrete superstep $s$, consider the values communicated by the greedy method above; we show that there is at least one optimal solution that executes these communication steps in superstep $s$. The optimality of our algorithm then follows from an induction in $s$.

Given a concrete superstep $s$, let $\Gamma$ denote the best solution that sends the same values in superstep $s$ as our greedy algorithm, and assume for contradiction that there is an alternative solution $\Gamma'_1$ that results in a lower total cost. Firstly, let us define $\Delta=C_{comm}\,\!^{(s)} - \max(|\Lambda'_{1,2}|, |\Lambda'_{2,1}|)$ in $\Gamma'_1$. If $\Delta$ is positive, then let us take $\Delta$ arbitrary values from both $\Lambda_{1,2} \setminus \Lambda'_{1,2}$ and $\Lambda_{2,1} \setminus \Lambda'_{2,1}$, and instead of communicating them in superstep $s$, let us assign the same communication steps to superstep $(s+1)$. Note that by the definition of $\Lambda'$, none of these values are required in superstep $(s+1)$, so we can indeed do this. Furthermore, this modified solution $\Gamma'_2$ reduces $C_{comm}\,\!^{(s)}$ by exactly $\Delta$, and only increases $C_{comm}\,\!^{(s+1)}$ by $\Delta$ at most, so the total cost does not increase either.

Furthermore, it is clear that the values $\Lambda'_{1,2}$ and $\Lambda'_{2,1}$ need to be transmitted in superstep $s$. Let $v_1$, ..., $v_k$ denote the further values that are sent by our algorithm from $p_2$ to $p_1$ in superstep $s$; recall that these are the first $k=|\Lambda'_{1,2}| - |\Lambda'_{2,1}|$ values in $\Lambda_{2,1} \setminus \Lambda'_{2,1}$. Assume that $\Gamma'_2$ transfers the values $\hat{v}_1$, ..., $\hat{v}_{k'}$ from $p_2$ to $p_1$ in superstep $s$ instead; assume that these are ordered according to the same rule as in our $\Lambda$ lists. Note that $k' \leq k$, since $\Gamma'_2$ has  $C_{comm}\,\!^{(s)} = \max(|\Lambda'_{1,2}|, |\Lambda'_{2,1}|)$. Now let us further modify $\Gamma'_2$ into $\Gamma'_3$ as follows: for all $i \in [k']$, we take the communication steps in $\Gamma'_2$ where $v_i$ and $\hat{v}_i$ are sent from $p_2$ to $p_1$, and we exchange these two steps for every such $i$.

The resulting $\Gamma'_3$ is still a valid communication schedule. On the one hand, sending $v_i$ instead of $\hat{v}_i$ in superstep $s$ is naturally a valid step, as also demonstrated by solution $\Gamma$. On the other hand, both the $v_i$ and the $\hat{v}_i$ values are ordered according to the time they are required on $p_1$, so $\hat{v}_i$ is always required at a later (or the same) superstep as $v_i$; this means that sending $\hat{v}_i$ in the step where originally $v_i$ was sent is also a valid step. Moreover, $\Gamma'_3$ has the same communication cost as $\Gamma'_2$, since the number of values communicated in each superstep remains unchanged. Finally, if $\Gamma'_3$ happens to have $k'<k$, we can also move the communication steps sending the $(k'+1)$-th, ..., $k$-th node of $\Lambda_{2,1}$ (from $p_2$ to $p_1$) to superstep $s$, without increasing $C_{comm}\,\!^{(s)}$.

This ensures that $\Gamma'_3$ now indeed transfer the same values in superstep $s$ as $\Gamma$. However, $\Gamma'_3$ was obtained from $\Gamma'$ without increasing the total cost; this contradicts the optimality of $\Gamma$.
\end{proof}

\subsection{NP-hardness in DB, FS and FB} \label{app:subs:3Dm}

Next we show that in general, the CS problem is NP-hard if communication happens according to the DB, FS or FB models.

\renewcommand*{\proofname}{Proof of Theorem \ref{th:commsched}}

\begin{proof}
In these models, NP-hardness can be shown through a reduction from a special case of the $3$D-matching problem. In this problem, we are given three classes of nodes $X$,$Y$,$Z$ of the same size $N=|X|=|Y|=|Z|$, and a set $\mathcal{H}$ of $M$ triplets $(x,y,z) \in X \times Y \times Z$, and the question is whether we can find $N$ triplets in $\mathcal{H}$ that form an \textit{exact cover}, i.e.\ they are disjoint (and hence they cover $X$, $Y$ and $Z$ entirely). This problem is known to be NP-complete \cite{GJ79}.

Given a $3$D-matching problem, we convert it into a CS problem on $P \! = \! 7$ or $P \! = \! 8$ processors, depending on the model. Three of the processors, named $p_X$, $p_Y$ and $p_Z$, will correspond to the classes $X$,$Y$,$Z$ and another three processors $p_{\widetilde{X}}$, $p_{\widetilde{Y}}$ and $p_{\widetilde{Z}}$ will represent them in \textit{reverse ordering} (to be discussed later). We will also have a source processor $p_0$ which needs to send values to all these six processors, without making the communication cost too large. In the FS model, we also have an auxiliary processor $p_a$, and instead of sending the values directly, $p_0$ will need to send them to $p_a$ first, and $p_a$ will need to relay them to the remaining processors.

As a general tool in our construction, for any superstep $s$, we can add a node $v$ that is assigned to some processor $p_1$ in superstep $s$, and draw an edge from this to another node that is assigned to some other processor $p_2$ in superstep $(s+1)$. This node can only be communicated form $p_1$ to $p_2$ in superstep $s$; as such, it does not have any degree of freedom in our CS problem, and indeed can be understood as a fixed cost that contributes $1$ to both $C_{sent}\,\!^{(s,p_1)}$ and $C_{rec}\,\!^{(s,p_2)}$. By adding a set of these immediately needed values between any pair of processors, we can essentially ensure that we have a desired fixed value for $C_{sent}\,\!^{(s,p)}$ and $C_{rec}\,\!^{(s,p)}$ for any specific $p$ or $s$, adding up to a fixed unavoidable communication cost over the whole schedule. The main idea of our construction is to convert the input $3$D-matching to a CS problem where the maximal allowed cost is set to this fixed communication cost. Hence intuitively speaking, the solvability of our CS problem will depend on whether the remaining (not fixed) communication steps in our construction can be scheduled in such a way that the communication cost is not increased above this fixed cost in any superstep; in other words, for the communication steps that can be scheduled flexibly, we need to find appropriate supersteps where the $C_{sent}\,\!^{(s,p)}$ and $C_{rec}\,\!^{(s,p)}$ values of the sending and receiving processors are not yet ``saturated'', i.e.\ they can both be increased without increasing the total cost $C_{comm}\,\!^{(s)}$ in the superstep.

The DAG schedule in our CS construction will be split into two parts: an \textit{initialization} part and a \textit{matching} part, consisting of $S_1$ and $S_2$ supersteps, respectively. Let us denote the nodes of $X$, $Y$ and $Z$ by $x_i$, $y_i$ and $z_i$, respectively, for $i \in [N]$. For each specific triple $(x_i, y_j, z_k) \in \mathcal{H}$, a construction will have a distinct node $u_{i,j,k}$ that is assigned to processor $p_0$ and superstep $1$ that represents this specific triple. We then draw an edge from this node $u_{i,j,k}$ to an arbitrary node in the following processors/supersteps:
\begin{itemize}[leftmargin=17pt, itemsep=3pt]
 \item processors $p_X$, $p_Y$ and $p_Z$ in superstep $S_1 + 6 \cdot i + 1$,
 \item processors $p_{\widetilde{X}}$, $p_{\widetilde{Y}}$ and $p_{\widetilde{Z}}$ in superstep $S_1 + 6 \cdot (N+1-i) +1$,
\end{itemize}
setting $S_2=6 \cdot N$.

Intuitively, the initialization part will allow us to transfer $M-N$ (but not more!) of the values $u_{i,j,k}$ to all the six other processors where they are required, already until superstep $S_1$, without increasing the communication cost in any of the supersteps $1, ..., S_1$. The specific details of how this is done depends on the concrete model. This will mean that in the matching part, we need to transfer the remaining $N$ values $u_{i,j,k}$ to the six other processors without increasing the communication costs in steps $S_1+1, ..., S_2$. Our construction will ensure that this is possible if and only if the $N$ triplets corresponding to these $N$ remaining nodes $u_{i,j,k}$ form an exact cover of $X$,$Y$,$Z$. As such, our construction will allow a communication schedule with the fixed immediate costs if and only if the original $3$D-matching problem is solvable.

The construction details for the matching part are identical in all of the models. In each superstep $s \in \{S_1+1, ..., S_2\}$, we ensure that $C_{sent}\,\!^{(s,p)}=1$ and $C_{rec}\,\!^{(s,p)}=1$ for all of the processors, with two exceptions: firstly, we have $C_{sent}\,\!^{(s,p_0)}=0$ for $p_0$, and secondly, we have $C_{rec}\,\!^{(s,p')}=0$ for another processor $p'$ that depends on the remainder of $s$ when divided by $6$; in particular, we set 
\begin{itemize}[leftmargin=17pt, itemsep=3pt]
 \item $C_{rec}\,\!^{(s,p_X)}=0$ if $s-S_1$ gives $1$ modulo $6$,
 \item $C_{rec}\,\!^{(s,p_Y)}=0$ if $s-S_1$ gives $2$ modulo $6$,
 \item $C_{rec}\,\!^{(s,p_Z)}=0$ if $s-S_1$ gives $3$ modulo $6$,
 \item $C_{rec}\,\!^{(s,p_{\widetilde{X}})}=0$ if $s-S_1$ gives $4$ modulo $6$,
 \item $C_{rec}\,\!^{(s,p_{\widetilde{Y}})}=0$ if $s-S_1$ gives $5$ modulo $6$,
 \item $C_{rec}\,\!^{(s,p_{\widetilde{Z}})}=0$ if $s-S_1$ gives $0$ modulo $6$.
\end{itemize}
This can indeed be easily implemented with fixed communication steps in any superstep $s$: we simply add an ``imaginary'' edge from $p_0$ to the desired target node $p'$, we arbitrarily extend this to a perfect matching (i.e.\ each processor sends a single fixed value to another processor), and then discard this imaginary edge. This results in a set of fixed communications where in terms of costs, we indeed have $C_{sent}\,\!^{(s,p)}=1$ and $C_{rec}\,\!^{(s,p)}=1$ for all of the processors, with the exception of $C_{sent}\,\!^{(s,p_0)}=0$ and $C_{rec}\,\!^{(s,p')}=0$. As such, in any of the models, there is only one possible communication step in superstep $s$ that does not increase the communication cost: sending a single value from $p_0$ to $p'$.

Now assume that at the end of the initialization part, there are $(M-N)$ of the values $u_{i,j,k}$ that are already available on all processors, and the remaining $N$ values are still available on $p_0$ only.

\begin{lemma}
The remaining $N$ values $u_{i,j,k}$ can be communicated without increasing the communication costs if and only if the corresponding $N$ triplets form an exact cover.
\end{lemma}

\renewcommand*{\proofname}{Proof}

\begin{proof}
This is not hard to see: our construction allows us to send exactly one value from $p_0$ to all other processors in the first $6$ steps, one in the second $6$ steps, and so on. As such, if we have an exact cover, then this contains all of $x_1$, $y_1$ and $z_1$ exactly once, and hence there is a single value among the $u_{i,j,k}$ that is needed on $p_X$ in superstep $S_1+7$, a single value that is needed on $p_Y$ in superstep $S_1+7$, and a single value that is needed on $p_Z$ in superstep $S_1+7$; we can transfer these from $p_0$ in supersteps $S_1+1$, $S_1+2$ and $S_1+3$, respectively. Similarly, the exact cover contains all of $x_N$, $y_N$ and $z_N$ exactly once, so for $p_{\widetilde{X}}$, $p_{\widetilde{Y}}$ and $p_{\widetilde{Z}}$, there is a single value that is needed on these processors in superstep $S_1+7$; we can transfer these from $p_0$ in supersteps $S_1+4$, $S_1+5$ and $S_1+6$.

The same technique can be continued for the rest of the values: all of $x_2$, $y_2$, $z_2$ and $x_{(N-1)}$, $y_{(N-1)}$, $z_{(N-1)}$ are contained in the cover exactly once, so apart from the values already sent, there is exactly one more value that is required on all six receiving processors by superstep $S_1+13$. These values can be transferred in steps $S_1+7, ..., S_1 + 12$. Following this pattern, we can communicate all the values by the time they are required, without increasing the communication cost in any of the supersteps.

For the reverse direction, assume that the values can be communicated without increasing the cost; this implies two things about the original triplets. Firstly, due to processor $p_X$, we know that our triplets contain $x_1$ at most once, the nodes in $\{ x_1, x_2\}$ at most twice, the nodes in $\{ x_1, x_2, x_3\}$ at most three times, and so on. The same holds for the nodes $y_i$ and $z_i$. Furthermore, due to processor $p_{\widetilde{X}}$, our triplets contain $x_N$ at most once, the nodes $\{ x_N, x_{(N-1)}\}$ at most twice, the nodes $\{ x_N, x_{(N-1)}, x_{(N-2)}\}$ at most three times, and so on (similarly for $y_i$ and $z_i$).

One can observe that in order for $N$ triplets to satisfy these conditions, they indeed need to form an exact cover. In particular, if the triplets are not an exact cover, then there must a lowest-indexed node $x_i$ (w.l.o.g. we assumed it is in class $X$) which is not contained in the triplets exactly once (i.e.\ it either does not appear at all, or it appears more than once). If $x_i$ appears more than once, then the nodes $\{ x_1, ..., x_i \}$ appear altogether more than $i$ times, which is a contradiction. On the other hand, if $x_i$ does not appear at all, then the nodes $\{ x_{(i+1)}, ..., x_N \}$ appear altogether more than $(N-i)$ times, which is again a contradiction.
\end{proof}

It remains to discuss the initialization part; we first discuss it for the broadcast models $DB$ and $FB$. In these models, the initialization will consist of $S_1=M-N$ supersteps, and each will allow us to transfer a single value from $p_0$ to all the remaining $6$ processors. For this, we simply need to add a single fixed value that has to be sent from e.g.\ $p_X$ to $p_0$ in each superstep $s \in [S_1]$; this sets the communication cost to $C_{comm}\,\!^{(s)}=1$, and ensures that $p_0$ can only send exactly one value in each of these supersteps without increasing the cost. Since the remaining six processors all have $C_{rec}\,\!^{(s,p)}=0$ so far, they can all receive this broadcasted value without increasing the cost. As such, by the end of the initialization, $p_0$ can only send $(M-N)$ of the values, and these can be received by all other processors. Note that while it would also be possible to communicate between the remaining six processors without a cost increase, this does not offer any advantage, since the only values that have left processor $p_0$ up to a given superstep $s$ are already known by all the six processors by superstep $s$.

In the $FS$ model, we ensure the same behavior using an auxiliary processor $p_a$ and altogether $S_1=7 \cdot (M-N)$ supersteps. Similarly to the matching part, we ensure in each of these $S_1$ supersteps that $C_{comm}\,\!^{(s)}=1$, and there is only a single processor that can send a value and a single processor that can receive a value without increasing the cost. In particular, in the first $(M-N)$ supersteps, we ensure that $C_{sent}\,\!^{(s,p)}=1$ and $C_{rec}\,\!^{(s,p)}=1$ for all $p \in [P]$, with two exceptions: we have $C_{sent}\,\!^{(s,p_0)}=0$ and $C_{rec}\,\!^{(s,p_a)}=0$. This ensures that the only allowed communication step in the first $(M-N)$ supersteps is to send a single value from $p_0$ to $p_a$. Then for the next $6 \cdot (M-N)$ supersteps, we again set $C_{sent}\,\!^{(s,p)}=1$ and $C_{rec}\,\!^{(s,p)}=1$ for all $p \in [P]$, with two exceptions: we set $C_{sent}\,\!^{(s,p_a)}=0$ for $p_a$, and we set $C_{rec}\,\!^{(s,p')}=0$ for another processor $p'$ in a periodic fashion again: 
\begin{itemize}[leftmargin=17pt, itemsep=3pt]
 \item $C_{rec}\,\!^{(s,p_X)}=0$ if $s-(M-N)$ gives $1$ modulo $6$,
 \item $C_{rec}\,\!^{(s,p_Y)}=0$ if $s-(M-N)$ gives $2$ modulo $6$,
 \item $C_{rec}\,\!^{(s,p_Z)}=0$ if $s-(M-N)$ gives $3$ modulo $6$,
 \item $C_{rec}\,\!^{(s,p_{\widetilde{X}})}=0$ if $s-(M-N)$ gives $4$ modulo $6$,
 \item $C_{rec}\,\!^{(s,p_{\widetilde{Y}})}=0$ if $s-(M-N)$ gives $5$ modulo $6$,
 \item $C_{rec}\,\!^{(s,p_{\widetilde{Z}})}=0$ if $s-(M-N)$ gives $0$ modulo $6$.
\end{itemize}
This ensures that for all other six processors $p'$, there are exactly $(M-N)$ rounds where the only allowed operation is to send a single value from $p_a$ to $p'$.

This implies that in the initialization part, $p_0$ will never be able to send a value to any of the six processors without increasing cost; instead, it has to send the values through $p_a$. In the first $(M-N)$ rounds, $p_0$ can communicate any chosen $(M-N)$ values to $p_a$, but not more; then $p_a$ can relay these to the remaining six processors in the next $6 \cdot (M-N)$ rounds without increasing cost.

We note that the reason why the same proof strategy does not work for the DS model is in this initialization part: the step of $p_0$ sending a value to e.g.\ $p_X$ and $p_Y$ are independent in the DS model, and hence there is no straightforward way to ensure that $p_0$ indeed has to send $(M-N)$ complete triplets to the remaining processors, and has no option to ``break them up'' and only send the desired parts of each triplet.

Altogether, this completes the reduction: our construction shows that in the resulting CS problem, there only exists a communication schedule with the predefined cost ($M+5_{\!}\cdot_{\!}N$ in the DB and FB models, or $7_{\!}\cdot_{\!}M - N$ in the FS model) if the original $3$D-matching problem is solvable, i.e.\ if there are $N$ triplets that form an exact cover.
\end{proof}

\subsection{With communication weights} \label{app:subs:commw}

On the other hand, with communication weights, it is more straightforward to show that the problem is already NP-hard for $P\!=\!2$.

\renewcommand*{\proofname}{Proof of Lemma \ref{lem:CSweights}}

\begin{proof}
Consider the $3$-partition problem with numbers $a_1$, ..., $a_{3 \cdot m'}$. We build a DAG that is split into $(m'+1)$ supersteps. Let us place a node into all $(m'+1)$ supersteps on both $p_1$ and $p_2$ so that they are non-empty; we will refer to this node on processor $p$ in superstep $s$ as $v_{p,s}$. To make the choice of $\pi$ and $\tau$ more reasonable, we can also draw an edge form $v_{p_i,s}$ to $v_{p_i,(s+1)}$ for all $i \in \{ 1, 2 \}$, $s \in [m']$.

Then for all $s \in [m']$, we draw an edge from $v_{p_1,s}$ to $v_{p_2,(s+1)}$, and choose a communication weight of $w_{comm}(v_{p_1,s}) = T$. Since the value of $v_{p_1,s}$ can only be communicated from $p_1$ to $p_2$ in superstep $s$, this will ensure that $C_{sent}\,\!^{(s,p_1)}=T$ for all $s \in [m']$.

On the other hand, we add $m$ distinct nodes $u_1$, ..., $u_{m}$ to processor $p_2$ and superstep $1$, we draw an edge from each of them to $v_{p_1,(m'+1)}$, and we set the communication weights of these nodes to $w_{comm}(u_i)=a_i$. This implies that these nodes can be freely transferred in any of the supersteps $1$, ..., $m'$ from $p_2$ to $p_1$. However, in order to not increase the cost $C_{comm}\,\!^{(s)}$ above $T$ for any $s \in [m']$, we must ensure that we communicate a set of values in each superstep that have total weight at most $T$. Since the total weight of the nodes is $m' \cdot T$, this means that we have to transfer a total weight of exactly $T$ from $p_2$ to $p_1$ in each superstep.

As such, a schedule with communication cost $m' \cdot T$ exist if and only if the original $3$-partition problem was solvable, which completes the reduction.
\end{proof}

\subsection{Comments on latency}

So far, we have considered the CS problem without latency; indeed since the number of supersteps is already decided in this setting, any communication schedule induces the same latency cost in total, so this does not influence the relative quality of the different solutions.

Alternatively, we might be interested in a variant of the CS problem where if there are no communications happening in supersteps $s$, then this latency cost is not added to the total (and hence we essentially have the option to merge supersteps $s$ and $(s+1)$ by an appropriate communication schedule, and save this latency cost). We note that the hardness results above also carry over to this setting with latency. In particular, the constructions in Sections~\ref{app:subs:commw} and \ref{app:subs:3Dm} ensure that there is a fixed communication happening in each superstep, so the total latency is equal in all solutions anyway. On the other hand, the greedy approach in Section \ref{app:subs:greedy} can easily become suboptimal if $L>0$, so Lemma~\ref{th:CSDP} is more challenging to extend to this case.

\section{ILP representation} \label{app:ILP}

%This section discusses the ILP representation in Proposition~\ref{prop:ILP}.

\subsection{ILP formulation details}

The main idea of the ILP representation has already been outlined in Section~\ref{sec:ILP}. The binary variables $\textsc{comp}_{v, p, s}$ allow us to indicate whether $v$ was computed on processor $p$ in the computation phase of superstep $s$. For each $v \in V$, we add a linear constraint
\[  \sum_{\substack{p \in [P] \\ \, s \in [S]}} \, \textsc{comp}_{v, p, s} \: = \: 1  \]
to ensure that each node is indeed assigned to exactly one processor and superstep. This constraint can be relaxed to an inequality (the sum needs to be $\geq 1$) to address the case when duplication is also allowed. Note that this also allows a node to be computed multiple times on the same processor, but this will never happen in an optimal solution. Alternatively, if we wish to specifically exclude this case, we can add further constraints for each fixed $p$ to ensure that the sum over all $s \in [S]$ is at most $1$.

On the other hand, the binary variables $\textsc{pres}_{v, p, s}$ indicate whether processor $p$ is already aware of the value of $v$ at the end of the computation phase of superstep $s$. As mentioned before, in case of e.g. the broadcast models, the correctness of these variables can be ensured with the condition
\[ \textsc{pres}_{v, p, s} \, \leq \, \textsc{pres}_{v, p, (s-1)} + \textsc{comp}_{v, p, s} + \textsc{rec}_{v, p, (s-1)} \]
for each $(v,p,s) \in V \times [P] \times [S]$. These constraints ensure that if $\textsc{pres}_{v, p, s}=1$, then at least one of $\textsc{pres}_{v, p, (s-1)}$, $\textsc{comp}_{v, p, s}$ and $\textsc{rec}_{v, p, (s-1)}$ must also be $1$, i.e.\ the variable was already present on $p$ before, or it was computed in the computation phase of superstep $s$, or received in the communication phase of superstep $(s-1)$. For the special case of $s=1$, we need to simplify the condition to $\textsc{pres}_{v, p, s} \, \leq \, \textsc{comp}_{v, p, s}$.

A similar constraint on the presence variables can also be applied in the DS model with the variables $\textsc{comm}_{v, p_1, p_2, s}$; in particular, we need to change the constraint to
\[ \textsc{pres}_{v, p, s} \, \leq \, \textsc{pres}_{v, p, (s-1)} + \textsc{comp}_{v, p, s} + \sum_{\substack{p' \in [P] \\ p' \neq p}} \, \textsc{comm}_{v, p', p, (s-1)} \, .\]

These variables already provide a straightforward way to encode the precedence constraints: a node $v$ can only be computed on $p$ in superstep $s$ if all of its predecessors are also available on $p$ by the end of the computation phase of superstep $s$. To express this, we can add a constraint $\textsc{comp}_{v, p, s} \leq \textsc{pres}_{u, p, s}$ for all $(u,v) \in E$ and all $p \in [P]$, $s \in [S]$; with our binary variables, this ensures that if we have $\textsc{comp}_{v, p, s}=1$, then we must also have $\textsc{pres}_{u, p, s}=1$.

If we introduce separate integer variables for all $C_{work}\,\!^{(s,p)}$ and $C_{work}\,\!^{(s)}$, then the work costs are also easy to compute from the variables $\textsc{comp}_{v, p, s}$. We can simply write
\[ C_{work}\,\!^{(s,p)} = \sum_{v \in V} \, \textsc{comp}_{v, p, s} \]
to obtain the work cost on a specific superstep, and then add a constraint $C_{work}\,\!^{(s)} \geq C_{work}\,\!^{(s,p)}$ for all $p \in [P]$, $s \in S$. This ensures that $C_{work}\,\!^{(s)}$ will indeed be larger than the work cost on each specific processor in $s$, and in the optimal solution, it will take the value of exactly the maximum. Note that this expression of work costs is straightforward to extend with work weights as coefficients in these linear constraints:
\[ C_{work}\,\!^{(s,p)} = \sum_{v \in V} \, w_{work}(v) \cdot \textsc{comp}_{v, p, s} \, . \]

Finally, the analysis of communication costs and constraints depends on the concrete communication model. As discussed, we use the binary variables $\textsc{sent}_{v, p, s}$ and $\textsc{rec}_{v, p, s}$ to capture communication in the DB and FB models, and the binary variables $\textsc{comm}_{v, p_1, p_2, s}$ to capture communication in the FS model (the DS model will be discussed separately later). The correctness of the communication steps (i.e.\ that processors only send values that they already possess) can be ensured by the constraints $\textsc{sent}_{v, p, s} \leq \textsc{pres}_{v, p, s}$ for all $v \in V$, $p \in [P]$, $s \in [S]$ in the broadcast models. The same thing can be ensured by $\textsc{comm}_{v, p, p', s} \leq \textsc{pres}_{v, p, s}$ in FS, for all $v \in V$, $p,p' \in [P]$, $s \in [S]$. Furthermore, in the broadcast models, we must ensure that if $\textsc{rec}_{v, p, s}=1$, then we have $\textsc{send}_{v, p', s}=1$ for some other $p'$, i.e.\ some processor is actually broadcasting the value $v$; this can be done by adding the constraint
\[ \textsc{rec}_{v, p, s} \leq \sum_{\substack{p' \in [P] \\ p' \neq p}} \, \textsc{sent}_{v, p', s} \]
for each $(v,p,s) \in V \times [P] \times [S]$.

Another thing to ensure in the direct transfer models (DS and DB) is that a value is indeed sent by the processor that computed it. One simple way to express this is to introduce an auxiliary binary variable $\textsc{home}_{v, p}$ for all $v \in V$, $p \in [P]$ to indicate whether $v$ was computed on $p$; for this, we can simply add the constraint
\[  \sum_{s \in [S]} \, \textsc{comp}_{v, p, s} \: = \: \textsc{home}_{v, p}  \]
for all $v \in V$, $p \in [P]$. Then besides requiring $\textsc{sent}_{v, p, s} \leq \textsc{pres}_{v, p, s}$ in the DB model, we can also add the constraint $\textsc{sent}_{v, p, s} \leq \textsc{home}_{v, p}$ (for all $v \in V$, $p \in [P]$, $s \in [S]$).

The communication costs $C_{sent}\,\!^{(s,p)}$, $C_{rec}\,\!^{(s,p)}$, $C_{comm}\,\!^{(s,p)}$ and $C_{comm}\,\!^{(s)}$ can again be expressed with separate variables and some simple constraints. The relations between these variables are naturally ensured with $C_{comm}\,\!^{(s,p)} \geq C_{sent}\,\!^{(s,p)}$, $ \; C_{comm}\,\!^{(s,p)} \geq C_{rec}\,\!^{(s,p)}$ and $C_{comm}\,\!^{(s)} \geq C_{comm}\,\!^{(s,p)}$ for all $p \! \in \! [P]$, $s \! \in \! [S]$. In fact, the $C_{comm}\,\!^{(s,p)}$ variables are not even needed as an intermediate step, we can simply write $C_{comm}\,\!^{(s)} \geq C_{sent}\,\!^{(s,p)}$ and $C_{comm}\,\!^{(s)} \geq C_{rec}\,\!^{(s,p)}$. Similarly to the case of work costs, the constraints formally only ensure that $C_{comm}\,\!^{(s)}$ is higher than these values, but in the optimal solution it is always set exactly to their maximum.

In the broadcast models, we further need to add
\[ C_{sent}\,\!^{(s,p)} = \sum_{v \in V} \textsc{sent}_{v, p, s} \: \quad \text{and} \: \quad C_{rec}\,\!^{(s,p)} = \sum_{v \in V} \textsc{rec}_{v, p, s} \]
for all $p \in [P]$, $s \in [S]$. In the FS model, we need to add
\[ C_{sent}\,\!^{(s,p)} = \sum_{\substack{v \in V \\ p' \in [P] \\ p' \neq p}} \textsc{comm}_{v, p, p', s} \]
and 
\[ C_{rec}\,\!^{(s,p)} = \sum_{\substack{v \in V \\ p' \in [P] \\ p' \neq p}} \textsc{comm}_{v, p', p, s} \]
for all $p \in [P]$, $s \in [S]$. Once again, these constraints can be easily extended with node-based communication weights (or even edge-based ones in FS) by adding them as simple coefficients: for instance, the constraint on $C_{sent}\,\!^{(s,p)}$ in the broadcast models becomes
\[ C_{sent}\,\!^{(s,p)} = \sum_{v \in V} \, w_{comm}(v) \cdot \textsc{sent}_{v, p, s} \, .\].

With the cost variables implemented, the objective function of the linear program can be naturally expressed as the actual BSP objective function. In the simpler case of no latency ($L=0$), the objective is simply to minimize
\[ \sum_{s \in [S]} \, C_{work}\,\!^{(s)} \, + \, g \cdot C_{comm}\,\!^{(s)} \, . \]
The optimal ILP solution to this problem naturally provides the optimal BSP schedule.

Extending this formulation with latency requires an extra technical step: we introduce a binary variable $\textsc{used}_s$ for each superstep $s \in [S]$ to indicate if there is indeed some communication happening in the communication phase of superstep $s$. We can set this variable to the correct value by adding the constraints
\[ \textsc{used}_s \geq \textsc{sent}_{v, p, s} \]
in the broadcast models for all $(v,p,s) \in V \times [P] \times [S]$, or by adding
\[ \textsc{used}_s \geq \textsc{comm}_{v, p, p', s} \]
for all $v \in V$, $p, p' \in [P]$, $p \neq p'$, $s \in [S]$ in the FS model. This then allows us to modify the ILP objective to minimizing
\[ \sum_{s \in [S]} \, C_{work}\,\!^{(s)} \, + \, g \cdot C_{comm}\,\!^{(s)} \, + \, L \cdot \textsc{used}_s \, . \]

Note that this extra technical step in case of latency is essentially required because the optimal number of supersteps is not known in advance; as such, the straightforward way to ensure that we are not excluding the optimal solution is to select $S=n$, and run the ILP solver on the formulation obtained from this choice. If a solution in fact uses $S<n$ supersteps, then the natural way to represent this is to number these supersteps consecutively from $1$ to $S$. However, there is no straightforward way to encode this in the ILP representation, and hence solutions with $S<n$ supersteps can be represented in multiple different ways in our ILP, possibly leaving an arbitrary subset of the supersteps $\{1, ..., n \}$ empty, or having consecutive supersteps where no communication happens at all (and hence they could be merged into a single superstep in the natural solution). The $\textsc{used}_s$ variables ensure that in case of latency, we charge no extra cost for these unnatural representations. We point out that in practice, an ILP-solver based approach might benefit from more sophisticated search strategies to find the optimal $S$ value.

Finally, let us discuss communication in the DS model separately. As mentioned before, a naive implementation could again use the variables $\textsc{comm}_{v, p_1, p_2, s}$ as in FS, and add extra constraints $\textsc{comm}_{v, p, p', s} \leq \textsc{home}_{v, p}$ similarly to DB; however, this increases the number of variables to $O(n \cdot P^2 \cdot S)$ in this model, too. Instead, consider the following more sophisticated approach. Since the DS model only allows direct transfer, the variables $\textsc{rec}_{v, p, s}$ in fact already determine the communications happening in superstep $s$, since each variable can only be sent from a single processor. As such, we can again apply the binary variables $\textsc{rec}_{v, p, s}$, and besides these, we also introduce a positive integer variable $\textsc{senttimes}_{v, p, s}$, which counts the number of times $v$ is sent by processor $p$ in the communication phase of superstep $s$. Firstly, we can add the constraints
\[ \textsc{senttimes}_{v, p, s} \, \leq \, p \cdot \textsc{home}_{v, p} \]
for all $(v,p,s) \in V \times [P] \times [S]$ to ensure that $\textsc{senttimes}_{v, p, s}=0$ if $v$ is not computed on $p$, and at most $p$ otherwise. More importantly, we add the constraints
\[ \textsc{senttimes}_{v, p, s} + p \cdot (1-\textsc{home}_{v, p}) \, \geq \, \sum_{\substack{p' \in [P] \\ p' \neq p}} \textsc{rec}_{v, p', s} \]
for all $(v,p,s) \in V \times [P] \times [S]$. If $\textsc{home}_{v, p}=0$, then this constraint has no effect: since $p \cdot (1-\textsc{home}_{v, p}) = p$, it is satisfied for any choice of $\textsc{senttimes}_{v, p, s}$ and $\textsc{rec}_{v, p', s}$. However, if $\textsc{home}_{v, p}=1$, then the constraint ensures that $\textsc{senttimes}_{v, p, s}$ must be at least as large as the number of processors receiving $v$ in superstep $s$ (and there is no motivation to increase it any larger). Given these variables, we can use the alternative definition
\[ C_{sent}\,\!^{(s,p)} = \sum_{v \in V} \textsc{senttimes}_{v, p, s}  \]
for all $p \in [P]$, $s \in [S]$ to obtain the correct communication cost in the DS model. Finally, note that these variables also result in minor changes for the remaining communication constraints in the DS model: to ensure that we only communicate values that are already available, we now need to use $\textsc{senttimes}_{v, p, s} \leq p \cdot \textsc{pres}_{v, p, s}$.

Altogether, the ILP formulations above use $O(n \cdot P \cdot S)$ variables in the DS, DB and FB models, and $O(n \cdot P^2 \cdot S)$ variables in the FS model. Note, however, that reducing the number of variables in the DS model also has a slight drawback. In particular, the formulations in the DB, FS and FB models ensure that the vast majority of variables are binary, and there are only $O(P \cdot S)$ variables (for cost measurement) that can take arbitrary integer values. In contrast to this, when introducing the $\textsc{senttimes}_{v, p, s}$ variables into the DS model, the number of non-binary variables increases to $O(n \cdot P \cdot S)$. However, recall that $\textsc{senttimes}_{v, p, s}$ is in fact also restricted to the integer interval $\{0, ..., P \}$, so its domain is also not significantly larger when $P$ is a small constant.

The number of linear constraints, on the other hand, is dominated by the precedence constraints in the DAG in most cases, which result in $O(|E| \cdot P \cdot S)$ constraints. In particularly sparse (non-connected) DAGs, we might have $|E|<n$, when the remaining constraints become dominant; hence more formally, the number of constraints in the broadcast models is
\[ O((n+|E|) \cdot P \cdot S) \, . \]
The only exception to this is FS, where we also use $O(n \cdot P^2 \cdot S)$ constraints to ensure the correctness of the communication steps; hence in this model, the number of constraints is
\[ O((n \cdot P +|E|) \cdot P \cdot S) \, . \]

\subsection{Brief overview of ILPs for other models}

We note that multiple methods have been studied before to model scheduling problems as an ILP; some overviews are available in e.g.\ \cite{ILPscheduling1, ILPscheduling2}. These works often consider other variants of the DAG scheduling problem; however, the same general techniques can sometimes be adapted to different model variants.

In particular, the straightforward ILP representation for classical models is a time-indexed formulation where similarly to our ILP above, there is a variable $\textsc{comp}_{v,p,t}$ to indicate if $p$ was computed on $p$ in time step $t$. This allows for a simple expression of the constraints, but the number of variables scales with the maximal possible makespan $n$, or even the sum of work weights in the weighted case. This can be significantly larger than the factor $S$ in our representation.

For some scheduling problems, e.g.\ so-called resource-constrained scheduling, this method also allows a formulation that does not scale with $P$. However, in case we have work weights, even such a formulation can require $O(n \cdot \sum_{v \in v} w_{work}(v))$ variables, whereas in our ILP, the number of variables does not scale with the weights at all. There are also more sophisticated representations, e.g.\ the event-based representation of~\cite{ILPscheduling2}, which only requires $O(n^2)$ variables even in case of work weights; it might be possible to also adapt these techniques to classical scheduling models. However, even $O(n^2)$ can be a higher than the number of variables we require in our ILP, e.g.\ if the number of supersteps is relatively small (i.e.\ $P \cdot S \leq n$), for example, due to a relatively high latency $L$. 

\end{document}